\documentclass[entropy,tutorial,accept,pdftex,moreauthors]{Definitions/mdpi}
\let\origmdpihl\hl
\renewcommand{\hl}[1]{\texorpdfstring{\origmdpihl{#1}}{#1}}

\firstpage{1}
\pubvolume{1}
\issuenum{1}
\articlenumber{0}
\pubyear{2026}
\copyrightyear{2026}
\externaleditor{Chintha Tellambura} 
\datereceived{6 August 2026}
\dateaccepted{12 October 2026}
\datepublished{ }
\graphicspath{{Figures/}}
\usepackage{tikz-3dplot}
\usetikzlibrary{shadings}
\usetikzlibrary{decorations.pathmorphing}
\usetikzlibrary{shapes.geometric}
\usepackage{amsmath,amssymb,amsthm}
\usepackage{bm}
\usepackage{booktabs}
\usepackage{mathtools}
\usepackage{tabularx}
\usepackage{tikz}

\newcommand{\rev}[1]{\textcolor{black}{#1}}
\usepackage{bbm}    
\usetikzlibrary{positioning}
\usetikzlibrary{calc} 
\newtheorem{Connection}{Connection}

\newcommand{\bq}{\mathbf{q}}
\newcommand{\bx}{\mathbf{x}}

\newcommand{\bz}{\mathbf{z}}
\newcommand{\bg}{\mathbf{g}}
\newcommand{\bc}{\mathbf{c}}

\newcommand{\bd}{\mathbf{d}}

\newcommand{\balpha}{\boldsymbol{\alpha}}
\newcommand{\be}{\mathbf{e}}
\newcommand{\bu}{\mathbf{u}}
\newcommand{\bv}{\mathbf{v}}
\newcommand{\bw}{\mathbf{w}}
\newcommand{\bb}{\mathbf{b}}
\newcommand{\bX}{\mathbf{X}}
\newcommand{\bY}{\mathbf{Y}}
\newcommand{\bA}{\mathbf{A}}
\newcommand{\bPi}{\boldsymbol{\Pi}}
\newcommand{\bC}{\mathbf{C}}
\newcommand{\bE}{\mathbf{E}}
\newcommand{\bF}{\mathbf{F}}
\newcommand{\bG}{\mathbf{G}}
\newcommand{\bJ}{\mathbf{J}}
\newcommand{\bh}{\mathbf{h}}

\newcommand{\bD}{\mathbf{D}}
\newcommand{\bH}{\mathbf{H}}
\newcommand{\bI}{\mathbf{I}}

\newcommand{\bP}{\mathbf{P}}
\newcommand{\bB}{\mathbf{B}}
\newcommand{\bepsilon}{\boldsymbol{\epsilon}}
\newcommand{\bsigma}{\boldsymbol{\sigma}}
\newcommand{\bQ}{\mathbf{Q}}
\newcommand{\bS}{\mathbf{S}}
\newcommand{\bT}{\mathbf{T}}
\newcommand{\bU}{\mathbf{U}}
\newcommand{\bW}{\mathbf{W}}
\newcommand{\bV}{\mathbf{V}}
\newcommand{\bR}{\mathbf{R}}
\newcommand{\bPhi}{\boldsymbol{\Phi}}
\newcommand{\bK}{\mathbf{K}}
\newcommand{\bM}{\mathbf{M}}
\newcommand{\bZ}{\mathbf{Z}}
\newcommand{\bmu}{\boldsymbol{\mu}}
\newcommand{\bSigma}{\boldsymbol{\Sigma}}
\newcommand{\bLambda}{\boldsymbol{\Lambda}}
\newcommand{\bPsi}{\boldsymbol{\Psi}}

\newcommand{\R}{\mathbb{R}}
\newcommand{\E}{\mathbb{E}}
\newcommand{\col}{\mathrm{col}}
\newcommand{\Tr}{\mathrm{Tr}}

\newcommand{\Cov}{\mathrm{Cov}}
\newcommand{\Var}{\mathrm{Var}}

\Title{From Common Information to Canonical Correlation Analysis and Beyond: Bridging Information Theory, Signal Processing, and Machine Learning}

\Author{Paris A. Karakasis* and Walid Saad }
\AuthorNames{Paris A. Karakasis and Walid Saad}

\address[1]{%
Institute for Advanced Computing, Bradley Department of Electrical and Computer Engineering, Virginia Polytechnic Institute and
State University, Alexandria, VA 22305, USA; \{parisk, walids\}@vt.edu}

\corres{\hspace{-0.82em}Correspondence: parisk@vt.edu}

\abstract{%
Many problems in signal processing, machine learning, and neuroimaging
provide several jointly observed measurements of one process, each mixing a
component common to all of them with a component private to itself, and require
the common component to be recovered without supervision. This problem has been
studied in several communities, largely independently. Statisticians proposed
canonical correlation analysis (CCA), introduced by Hotelling in 1936.
Information theorists proposed two notions of common information, due to Gács
and Körner and to Wyner. Engineers developed subspace methods for sensor arrays
and multiuser communications to design multiview receivers. This tutorial
presents the main results and tools each community developed for this problem in
a single exposition, and explains the connections that emerge where they meet,
\rev{distinguishing exact equivalences from special cases and analogies.} We develop the information-theoretic foundation from first principles. We then introduce CCA and its extension to several views, followed by the engineering thread from
blind multiuser detection to CCA as a blind receiver. Next, we turn to nonlinear
CCA, covering kernel, deep, and probabilistic formulations, the conditions under
which the common factor is identifiable, and connections to self-supervised
learning. We conclude with a discussion of further extensions and open
directions. The exposition is self-contained, so that readers from each of these
communities can follow the parts developed by the others.}


\begin{document}

\fancypagestyle{plain}{\fancyhf{}\renewcommand{\headrulewidth}{0pt}\renewcommand{\footrulewidth}{0pt}}
\fancyhead[L]{\footnotesize Accepted for publication in \emph{Entropy}}
\fancyfoot[R]{}

\section{Introduction}
\label{sec:intro}

\subsection{A Question That Crosses Disciplines}
\label{subsec:crossing-disciplines}

A wide range of applications, from wireless communications and neuroimaging to
speech processing and machine learning, rely on observing the same phenomenon
through several measurements. These measurements, which we call views, may come
from different instruments, from different locations, or from different
modalities, and each view records something the others miss. What should be
extracted from such multiview measurements depends on the application. In some
applications, each view is incomplete on its own, and the goal is to merge the
views into a single representation that describes the phenomenon more fully than
any one of them does alone. This process is known as multiview fusion. In other applications, the opposite holds. Each view captures the phenomenon completely, but most of what it records is unrelated to the phenomenon, and the goal is to keep only the part that is due to the phenomenon. Prior information that points to the interesting part of the recordings is often unavailable, and the only side information is that jointly observed samples of the different views share something. Since the phenomenon is the only thing
common to all views, the task becomes to recover what the views have in common.
This problem, which we call common factor recovery, is the subject of this
tutorial. The problem is important because it arises across engineering and the
sciences, typically when the signal of interest is weak and no supervision is
available, and it has been posed independently in fields that share little
else. We open with three instances of it, from wireless communications,
neuroimaging, and machine learning.

\textbf{{Wireless communications.} 
} A fundamental problem in wireless communications is
to recover a transmission that is received buried under others that are orders of
magnitude stronger. For example, consider two base stations, each equipped with
several antennas, receiving uplink transmissions from the users in their vicinity, as in Figure~\ref{fig:celledge}. The transmissions of users near the boundary between the two cells, the cell-edge users, reach both stations, but at low power, since those users are far from either one. A cell-center user close
to a given base station is received there at high power and lies near the noise floor at the other. Every user has to be decoded, and the ones near a station present no difficulty. The cell-edge users are the challenging case, since their signals sit underneath those of the nearby users.

\begin{figure}[H]
\centering
\includegraphics[width=0.85\linewidth]{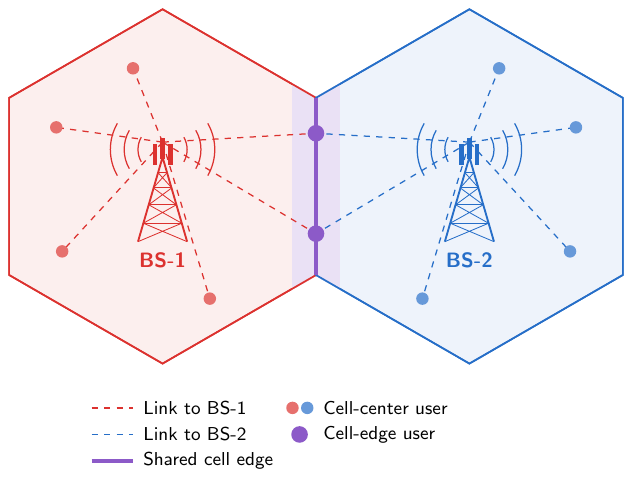}
\caption{{Two} 
 neighboring cells, each served by a multi-antenna base station. A user near the center of a cell is received strongly at its own station and is barely heard at the other, so only one link is drawn for it. The two cell-edge users are received at low power by both stations, so each of them carries a link to BS-1 and a link to BS-2, and they are the challenging case for decoding.}
\label{fig:celledge}
\end{figure}

\textbf{Neuroimaging.} A fundamental problem in neuroimaging is to estimate the spatiotemporal response to a stimulus. This is usually done by extracting task-related pairs of temporal signatures and brain activation maps from recorded brain signals. For example, consider multiple subjects performing the same visuo-motor task while functional magnetic resonance imaging (fMRI) records their brain activity. {{A visuo-motor} 
 task is one in which the subject produces a movement in response to a visual stimulus, so that visual and motor areas are engaged together. A common instance is pressing a button whenever a checkerboard pattern appears on a screen.} For each subject, the measurements can be arranged in a matrix of size $\text{space} \times \text{time}$ by stacking the time series of all volume elements, or voxels. Each recording mixes the response to the task, which is shared across subjects, with spontaneous activity private to each individual subject. The goal is to infer the task-related temporal and spatial components of the response from the recordings without knowing what the stimulus was or what the two components look like. The only assumption available is that the measurements of all subjects share them.

\textbf{Multimodal machine learning.} A fundamental problem in multimodal learning is to
build a representation of two modalities that keeps what they agree on and discards
what belongs to only one. For example, consider a speaker recorded in two modalities
while talking. In the first one, sensors track the positions of the tongue, lips and
jaw as the vocal tract moves. In the second one, a microphone records the resulting
sound. Several common factors characterize both, among them the phoneme being
produced and the speaker's identity, speaking rate, age and sex. Each recording also carries
variation the other modality does not. For example, the articulatory traces carry sensor noise and involuntary movement, and the waveform carries the acoustics of the room. Each also carries measurement noise independent of the other. Learning from both modalities isolates the factors they share. A few labels are then enough to build a representation of the factor a downstream task needs.

Naturally, one wonders whether common factor recovery, in these three examples and in general, is just an ordinary estimation task. As we will see, the problem is much more challenging, for three technical reasons.

The first is that we do not know which part of the measurements corresponds to the common factor and which to the private ones, since we observe only their mixture. In the most general case, we also lack any kind of side information about their shape, statistics, or any other
characteristics. The only thing we know is the coupling of the joint measurements, i.e., that jointly observed multiview samples share something. Infinitely many valid decompositions are consistent with the same measurements, each assigning some factors to both views and others to only one. As will become clear, what counts as common and private has to be defined more rigorously than this. 

The second reason is the absence of supervision. The cell-edge users are not identified in advance, their number can even be unknown, and their symbol sequences may be partially or fully unknown. In fMRI, the expected brain response is unknown, though a canonical shape is often assumed. In a controlled experiment the stimulus timing is known but is deliberately not used, so that the recovered response is not constrained by the assumed model. Under naturalistic stimulation, such as watching a film or listening to continuous speech, there is no design to use at all. The speech recordings may come with no labels, or with very few. 

The third is that the common factor need not dominate in the measurements. For example, in the wireless cellular network use case, cell-edge users are received at lower power than the interference, so any criterion that ranks structure by energy discards exactly their transmissions. 

Any methodology must therefore provide four elements, namely a rigorous definition of the shared structure, an objective that allows us to design a well-posed mathematical problem for the extraction of the common factor, an estimator that computes it from finite samples, and conditions under which the estimator returns the defined object rather than an artifact of the data. This tutorial is organized around those four requirements.

\subsection{Four Communities, One Problem}

Remarkably, the problem set out above has been studied by four communities, using different
semantics, and the resulting bodies of work developed with few explicit
connections between them.

\textbf{Statistics.} Hotelling introduced canonical correlation analysis (CCA) in
{1936}
~\cite{hotelling1936}, seeking the linear combinations of two sets of variables
that correlate most strongly with each other. His examples were reading and arithmetic
scores in schoolchildren, and prices and quantities across related commodities.
Kettenring~\cite{kettenring1971} later cataloged five criteria for extending the formulation to more than two sets, two of them his own, with all five coinciding for two sets of variables and differing beyond that.

\textbf{Information theory.} Gács and Körner asked in 1973 what two random variables
hold in common in the strict sense of a quantity that each can compute on its own. They showed that such a quantity can be far smaller than the mutual information, vanishing entirely for most distributions of practical interest~\cite{gacs1973}. Two years later, Wyner posed a different question, namely how much information a latent variable must carry in order to render the two random variables conditionally independent, and answered it with the common
information~\cite{wyner1975}. The Gray--Wyner network~\cite{grawyner1974} gave that quantity an operational meaning as the minimum rate on a shared communication link,  subject to lossless reconstruction at a total rate equal to the joint entropy.
Tandon, Sankar, and Poor~\cite{tandon2011} later extended the Gray--Wyner
network to several sources under a privacy requirement.

\textbf{Signal processing.} Wang and Poor showed in 1998 that a receiver can estimate the subspace spanned by the users' signatures directly from the received data, and use that estimate to suppress interference from users whose signatures it does not know~\cite{wangpoor1998}. Two decades later, Ibrahim and Sidiropoulos showed that two receivers observing a common subset of users can recover the subspace that subset spans using CCA, with no knowledge of any signature~\cite{ibrahim2020}.

{\color{black}\textbf{Machine learning.} Lai and Fyfe published a first kernel version of CCA in 2000~\cite{laifyfe2000}, and Akaho proposed in 2001 the formulation in use today, in which the linear combinations of each view are replaced with functions in a reproducing kernel Hilbert space~\cite{akaho2001}. In 2013, Andrew et~al. replaced the kernel with a neural network per view, which yielded deep CCA~\cite{andrew2013}. Later, the same correlation objectives formed a basis for current self-supervised learning, where the two views are constructed from one sample rather than observed~\cite{zbontar2021barlow, bardes2022vicreg}.}

CCA remained largely confined to statistics for a long time. Borga observed in 2001 that CCA was surprisingly little known in learning and signal processing, and wrote a tutorial addressed to those two audiences~\cite{borga2001}. Tutorials aimed at machine learning~\cite{hardoon2004, uurtio2017} and neuroscience~\cite{wang2018} followed. These treatments establish how CCA is formulated and how to compute it efficiently, which leaves two questions open. The first is what the computed quantity means. A large canonical correlation says that some linear combination of the first view agrees with some linear combination of the second. Going beyond this requires an assumed generative model that names what the common factor is, such as a cell-edge user or a response common to several subjects. The second is whether the recovered structure corresponds to anything real. A high correlation between two transformed views shows that the transformations agree. However, it does not necessarily show that the views share something concrete, since the agreement can be manufactured by the transformations
themselves, especially when the class they are drawn from is rich. Answering
that requires more than a generative model. It requires conditions under which what is recovered is a common factor rather than an artifact, which a correlation criterion by itself does not~supply.

Attempts to answer the first question were made twice, and neither answer identifies the common factor. Information theory answered it as a rate. Gács and Körner and Wyner each wrote down a quantity in bits, and the Gray--Wyner network turns Wyner's into the rate that a link shared by two receivers must carry~\cite{gacs1973, wyner1975, grawyner1974}. These quantities settle how much the two views have in common and what describing it costs. However, they say nothing about what the common factor is. Bach and Jordan answered the question with a generative model~\cite{bach2005}. Each view is its own affine function of one shared latent vector plus zero-mean Gaussian noise with its own arbitrary covariance. Under this assumption, the maximum-likelihood estimates of the model parameters are given by the canonical directions, and the posterior mean of the shared latent vector given either view is a linear function of the canonical variates of that view. Naming the common factor, as a cell-edge user or a task response, remains application specific. The second question was answered much more recently~\cite{ibrahim2019, sorensen2021, lyu2022, karakasis2023}, and in the later ones one key condition is the same conditional independence that Wyner wrote down in 1975.  In this tutorial, we take those information-theoretic definitions as the
foundation. We develop them first and then introduce CCA. For jointly Gaussian
views, the canonical correlations determine Wyner's common information exactly,
whereas the link between CCA and the Gács--Körner common information is only
one-sided. We then trace the line from subspace methods in communications
through CCA to the identifiability theory that came later, in both linear and
nonlinear settings.

\subsection{Tutorial Objectives and Contributions}

The main objective of this tutorial is, thus, to make four bodies of literature accessible to a
reader who wants to understand how shared structure appears across multiple
jointly observed sources. Each of the four communities has developed a
substantial body of theory and tools, largely without reference to the others. We present the main principles behind common information and canonical
correlation in detail, together with examples of applications, so that a reader
can see how the concepts of the four literatures relate, apply the resulting
tools, and assess what the outcome of an analysis means. A further objective is
to encourage progress on these topics, and for that reason the tutorial also
discusses open problems and research directions.

Specifically, the tutorial makes the following contributions.
\begin{itemize}
\item We develop the Gács--Körner and Wyner common information from first
principles, together with the Gray--Wyner network, their extensions to several
and to continuous variables, and worked examples for binary and Gaussian
sources (Section~\ref{sec:common_info}).
\item We present CCA from Hotelling's formulation, show that for jointly
Gaussian views the canonical correlations determine both the mutual information
and Wyner's common information, and give the conditions under which CCA and
generalized CCA recover the common subspace (Section~\ref{sec:cca}).
\item We follow the signal processing thread from the blind multiuser detection
of Wang and Poor to CCA-based receivers, for cell-edge users and for a second
view designed at the transmitter (Section~\ref{sec:sp}).
\item We cover kernel, deep, and probabilistic formulations of nonlinear CCA,
explain their trivial solutions, state the conditions under which the common
factor is identifiable, and relate these conditions to the feasible sets of the
Gács--Körner and Wyner common information. We also connect these formulations
to self-supervised learning (Section~\ref{sec:nonlinear}).
\item Throughout, boxes labeled Connection relate the information-theoretic and
the statistical developments, each labeled with the type of relation it states,
and Section~\ref{sec:open} collects further extensions and open problems.
\end{itemize}

\subsection{Relation to Existing Tutorials \& Monographs}

The prior art focuses mainly on one of the two main topics, common information and canonical correlation. On the statistics and
learning side, several tutorials on \mbox{CCA~\cite{borga2001, hardoon2004,
uurtio2017, wang2018, yang2021survey}} differ in which parts of the subject they
cover. Borga's tutorial~\cite{borga2001} covers the statistical formulation and
its links to related second-order methods, with information theory entering only
through the mutual information of Gaussian views. Hardoon
et~al.~\cite{hardoon2004} introduce CCA and kernel CCA to machine learning
researchers and sketch an extension of CCA to more than two sets. Uurtio
et~al.~\cite{uurtio2017} add regularized and sparse variants together with a
brief account of the deep one. Wang et~al.~\cite{wang2018} write a primer for
neuroscientists, with worked examples that link brain imaging to behavior,
psychopathology and genetics. Yang et~al.~\cite{yang2021survey} give the
broadest coverage of the formulations, sorting them into multi-view,
probabilistic, deep, kernel, discriminative, sparse and locality preserving
families, and pairing each with representative applications and datasets. None
of them treats the information-theoretic definitions as the foundation of the
subject. Only Yang et~al. and Wang et~al. raise identifiability, Yang et~al. for
Bayesian CCA and Wang et~al. for linear CCA when the number of samples is small
relative to the number of variables. However, neither examines the conditions under which the recovered factor is guaranteed to be a common factor.

On the information-theoretic side, the monograph of Yu and
Tan~\cite{yu2022common} surveys Wyner's and the Gács--Körner
common information in great depth. The authors cover the operational meanings of
both notions in distributed source simulation and common randomness extraction,
as well as their extensions, together with noise stability and the associated
functional inequalities. The Hirschfeld--Gebelein--Rényi maximal correlation
enters that monograph as a tool in the analysis of the
Gács--Körner common information and of noise stability, but the
monograph does not reach CCA or the identifiability of the common factor.

{\color{black}The monograph of Huang et~al.~\cite{huang2024universal} starts from the
information-theoretic side and shares the most common ground with the present
tutorial. That monograph builds on the modal decomposition of the joint
distribution of a pair of variables, mainly on finite alphabets, together with a
local information geometry perspective. From there, the monograph devotes a chapter to Hotelling's CCA for jointly Gaussian vectors under linear features, and relates
the features of the modal decomposition to both Wyner's and the Gács--Körner
common information. In addition, the monograph surveys only briefly the kernel and deep formulations of CCA as well as multivariate extensions. Finally, the monograph does not address the identifiability of the common factor, and the detection and estimation covered there are statistical formulations, outside the wireless communications setup. Table~\ref{tab:tutorials} places the present tutorial alongside these treatments.}

\begin{table}[H]
\caption{Scope of existing treatments of common information and CCA. A check
mark means the topic is covered substantively, and Partial means it is touched
on without being developed.\label{tab:tutorials}}

\begin{adjustwidth}{-\extralength}{0cm}
\begin{tabularx}{\fulllength}{lCCCCCCC}
\toprule
\textbf{Work} & \textbf{Info. Theory} & \textbf{Statistics} & \textbf{Sig. Proc.}
& \textbf{Multi-View CCA} & \textbf{Nonlinear CCA} & \textbf{Deep CCA} & \textbf{Identif.} \\\midrule
Borga~\cite{borga2001} & Partial & {\checkmark} 
 & Partial & & & & \\
Hardoon et~al.~\cite{hardoon2004} & & \checkmark & & Partial & \checkmark & & \\
Uurtio et~al.~\cite{uurtio2017} & & \checkmark & & & \checkmark & Partial & \\
Wang et~al.~\cite{wang2018} & & \checkmark & & Partial & Partial & Partial & Partial \\
Yang et~al.~\cite{yang2021survey} & & \checkmark & & \checkmark & \checkmark &
\checkmark & Partial \\
Yu and Tan~\cite{yu2022common} & \checkmark & & & & & & \\
\rev{Huang et~al.~\cite{huang2024universal}} & \checkmark & \checkmark & & Partial &
\checkmark & Partial & \\
This work & \checkmark & \checkmark & \checkmark & \checkmark & \checkmark &
\checkmark & \checkmark \\
\bottomrule
\end{tabularx}
\end{adjustwidth}
\end{table}

\subsection{Organization}

Hereinafter, we move from the definitions of what is common, in
information-theoretic and statistical terms, to the methods that recover the
common factor from data, first under linear and then under nonlinear mixing.
Each section is outlined below.

\begin{enumerate}
\item[$\bullet$] In Section~\ref{sec:common_info}, we build the
information-theoretic foundation. The section covers the Gács--Körner and Wyner
common informations, the Gray--Wyner network, the multi-user extension of
Tandon, Sankar, and Poor, the passage from discrete to continuous sources
through a lossy coding formulation, and partial information decomposition as a
complementary account of what sharing means. The section assumes no background
in signal processing or statistics.

\item[$\bullet$] In Section~\ref{sec:cca}, we introduce CCA through Hotelling's
formulation and computation. Then, we present the relation between the canonical correlations and the mutual information of Gaussian views~\cite{scharf2000, borga2001}, and the
probabilistic interpretation of Bach and Jordan~\cite{bach2005}. This is where common information and canonical correlation meet for Gaussian views. Next, we give conditions under which, in the noiseless case, the subspace recovered by
CCA is the common subspace. Finally, we extend CCA and these conditions to more than two views through generalized CCA.

\item[$\bullet$] In Section~\ref{sec:sp}, we reach CCA from multiuser detection
in wireless communications, where a receiver must separate the signals of
several users that arrive superimposed. We start from code-division multiple
access, where each user is assigned a code known to the single-antenna receiver, and present the subspace receivers of Wang and Poor~\cite{wangpoor1998}, which need only the code of the desired user. We then turn to multi-antenna receivers, which replace the code with a few known symbols of that user. Next, we present the work of Ibrahim and Sidiropoulos~\cite{ibrahim2020}, in which two base stations receiving the same cell-edge users provide two views and CCA recovers the subspace of these weak users regardless of how strong the interference from the cell-center users is, without codes or known symbols. Finally, we show how the transmitter can create the second view by repeating its data~\cite{ibrahim2024mmse}, in which case CCA yields the linear minimum mean-square-error (MMSE) detector without channel estimation.

\item[$\bullet$] In Section~\ref{sec:nonlinear}, we relax the assumption that
the shared structure is linear. The section covers optimal scaling, kernel CCA,
deep CCA and its multiview extensions, and the probabilistic and variational
formulations of CCA. For the correlation-based methods, once the class of
admissible transformations is rich enough, agreement between the transformed
views on the training samples can be produced by the transformations themselves
rather than by anything the views share. The probabilistic and variational
formulations face different difficulties, which we also discuss. The
section then presents two sets of conditions under which the common factor is
identifiable, and closes with self-supervised learning, where an objective of
the same structure reappears with the views constructed rather than given.

\item[$\bullet$] In Section~\ref{sec:open}, we collect further extensions that the preceding sections do not develop, such as tensor CCA and a unification of generalized CCA with subspace clustering. We provide them as pointers for the reader who wants to go further. We then state a series of open problems and research directions, from relaxing the conditions under which the common factor is identifiable, to distinguishing the factor of interest from a confounder that the views share.
\end{enumerate}

Readers with different interests can take different paths through the sections. However, we encourage every path to pass through Section~\ref{sec:cca}. A reader from signal processing can skip Section~\ref{sec:common_info}, together with the connections to that section in Section~\ref{sec:cca}, since Section~\ref{sec:sp} relies on the CCA formulation and the identifiability conditions of Section~\ref{sec:cca}. Section~\ref{sec:sp} refers to Section~\ref{sec:common_info} only in connecting remarks. A reader from information theory may find more interest in Section~\ref{sec:common_info} and then Section~\ref{sec:cca}, where the connections to the Gács--Körner and Wyner common information are made. A reader from machine learning can go from Section~\ref{sec:cca} directly to Section~\ref{sec:nonlinear}, which also draws on the maximal correlation and the conditional independence of Section~\ref{sec:common_info}. Section~\ref{sec:open} draws mainly on Sections~\ref{sec:cca} and~\ref{sec:nonlinear}.

Finally, many objects that the remaining sections work with are related but not necessarily interchangeable. \rev{Table~\ref{tab:terms} summarizes the terms used for them throughout the tutorial. Moreover, boxes labeled Connection link the information-theoretic and the statistical developments. Each box states whether the link is a one-sided implication, a consequence or an equivalence that holds only in a special case, such as jointly Gaussian views, or an analogy that shares structure without equating the quantities.}

\begin{table}[h]
{\color{black}
\caption{Terminology used throughout the tutorial.\label{tab:terms}}
\begin{tabularx}{\textwidth}{l>{\raggedright\arraybackslash}X}
\toprule
\textbf{Term} & \textbf{Meaning} \\
\midrule
Shared structure & Informal term for whatever the views have in common, used
only before one of the definitions below is adopted. \\
\midrule
Common information & A number, measured in bits and computed from the joint
distribution of the views, namely the Gács--Körner or the Wyner common
information. The number says how much the views share, not what they share. \\
\midrule
Common factor & A latent random vector shared by all views in a multiview
generative model, such that the views are conditionally independent given the
common factor. Whatever else enters a view is private to that view and forms
its private factor. \\
\midrule
Common subspace & The column space of the matrix whose rows are the values of the common factor at the training samples. \\
\midrule
Common representation & What the outputs of the learned encoders of the views
are made to agree on.  \\
\bottomrule
\end{tabularx}}
\end{table}

\section{The Information-Theoretic View of Common Information}
\label{sec:common_info}

\subsection{What Does ``Shared'' Mean?}

In Section~\ref{subsec:crossing-disciplines} we considered three problems. In the first, we are interested in combining the signals received at two base stations to decode the cell-edge users, which both stations receive at low power. In the second, several subjects perform the same task during fMRI scanning sessions, and we are interested in estimating the response to the task that is shared by all subjects. In the third, a speaker is recorded by vocal-tract sensors and a microphone, and we are interested in a representation of the factors the two recordings share,
such as the phoneme being spoken. One may notice that all three are instances of the same, more abstract problem. In each, several jointly observed views are available. Each view is generated by a factor common to all views together with factors private to that view, and we are interested in recovering the common one. Many other problems share this form. Two microphones record one speaker, two cameras watch one scene, and gene expressions and clinical measurements come from one patient. The views may differ in form, because each is recorded by a different instrument, at a different location, or from a different subject. In this section, we lay the foundations for studying the shared structure of two or more random variables, which information theorists quantify as common information. This is the first step toward the question that runs through the
tutorial. \emph{{What do the views share, and how can we recover what they share?} 
}

Two examples show what the answer can look like. Suppose we install two perception systems at an intersection, to support autonomous driving and traffic monitoring. For each detected object, such a system returns a label from its own label set. The two systems come from different manufacturers, and their label sets differ significantly. We therefore want a third system, on top of the two, that aligns their outputs. Unlike the perception systems, the third system never processes sensor data. The third system only applies one map to the labels of each perception system, and the two maps must send the labels of the same object to the same common label. Notice that this requirement alone is met by the trivial pair of maps that sends every label to a single common label, such as ``object'', and possibly by coarse labelings such as vehicle and person. Among all such pairs of maps, we want the pair whose common labeling describes the detected objects as finely as possible. The maps are learned from the pairs of labels that the two systems assign to the same object as the object passes through the intersection, as illustrated in Figure~\ref{fig:perception}.

\begin{figure}[H]
\centering
\includegraphics[width=0.8\textwidth]{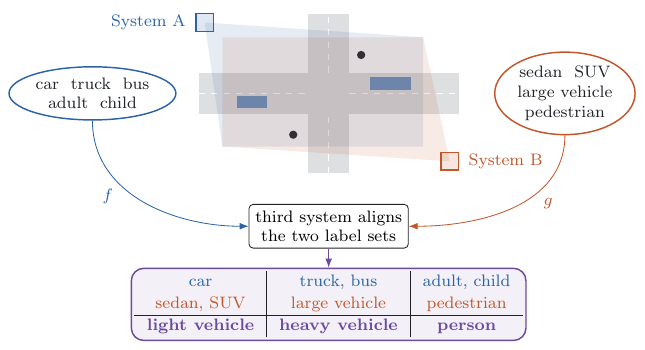}
\caption{Two perception systems at an intersection label the same objects with different label sets. A third system maps the labels of System A through $f$ and the labels of System B through $g$ to a common label set. The finest labeling on which the two maps always agree is light vehicle, heavy vehicle and person.}
\label{fig:perception}
\end{figure}

A second example involves two single-channel sources of information affected by
a common factor, such as a pair of microphones recording the same speaker, as illustrated in Figure~\ref{fig:microphones}. Each captures the speaker's voice, the common signal, together with its own background noise, which we assume to be independent of the voice and of the noise at the other microphone. The two recordings are correlated because of the shared voice and differ because of the private noise. The speaker's voice is therefore what the two recordings share, accounting for their entire dependence. Once the voice is known, one recording carries no further information about the other.

\begin{figure}[H]
\centering
\includegraphics[width=0.5\textwidth]{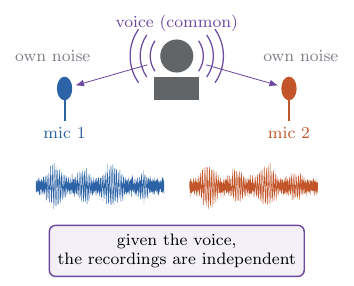}
\caption{One speaker recorded by two microphones. Each recording contains the voice, which is common to both, and the noise of its own microphone. The recordings are dependent only through the voice, so given the voice the two recordings are conditionally independent.}
\label{fig:microphones}
\end{figure}

The two examples answer our question in different ways. In the first, each
perception system can compute the common label from its own output, and both
systems arrive at the same answer. What is shared is then the richest labeling
that each view can determine on its own. In the second, neither recording
determines the voice exactly, since each recording also carries its own noise.
The voice is shared in a weaker sense. Once the voice is given, the two
recordings become conditionally independent. What is shared, in this case, is then a factor that explains the dependence between the views, and the most parsimonious such factor is sought. The two readings of ``shared'' lead to different quantities,
and information theory makes each precise. Gács and Körner formalized the first,
and Wyner the second. Both quantities are measured in bits and derive their
meaning from coding problems. We develop them in turn, beginning with the
earlier one, due to Gács and Körner.

\subsection{Gács--Körner Common Information: The Highest-Entropy Shared Function}
\label{sec:gk}

The first attempt at a formal definition of common information was given by Gács and Körner~\cite{gacs1973} in 1973. For two discrete random variables $X$ and $Y$, their notion asks for the highest-entropy random variable $W$ that can be computed \emph{deterministically} from both $X$ and $Y$ separately. Mathematically, this can be formulated {as} 
\begin{linenomath}
\begin{equation}
\begin{split}
      K_{GK}(X;Y) &= \max_W H(W)\\
    \text{ s.t. }~  W=&f(X)=g(Y),
\end{split}
  \label{eq:gk}
\end{equation}
\end{linenomath}
for deterministic functions $f$ and $g$. Since $X$ and $Y$ are discrete, $f$ and $g$ are simply lookup tables. Every feasible $W$ is an early instance of what was later called common
randomness~\cite{csiszar2011}.

To build some further intuition, we consider two simpler examples. First, let $A$, $B$ and $C$ be independent uniform bits, and define $X \triangleq (A,B)$ and $Y \triangleq (A,C)$, so that both $X$ and $Y$ take values in the alphabet $\{00,01,10,11\}$. Then $W = A$ is an optimal nontrivial shared labeling, giving $K_{GK}(X;Y) = H(A) = 1$~bit $= I(X;Y)$. This is because the entire dependence between $X$ and $Y$ is mediated by the shared bit $A$, which each of $X$ and $Y$ can determine on its own. Note that any invertible relabeling of $A$ is also a valid optimal solution. Second, let $X \triangleq A \cdot B$ and $Y \triangleq A \cdot C$, where $\cdot$ is the AND operation. When $A = 0$, both $X$ and $Y$ equal $0$ regardless of $B$ and $C$, so $X$ and $Y$ are positively correlated and $I(X;Y) > 0$. Yet $K_{GK}(X;Y) =0$, since both $X$ and $Y$ take values in $\{0,1\}$, and all four joint values of $(X,Y)$ occur with positive probability. Hence, no deterministic and nontrivial mappings $f(X)$ and $g(Y)$ can agree with probability one.

More generally, the whole spectrum is defined by two extremes and the intermediate case. {\color{black}At one extreme, $X$ and $Y$ have equally large alphabets, and each is a deterministic function of the other. Then, optimal mappings $f$ and $g$ are simply bijective relabelings, i.e., permutations of the values of $X$ and $Y$, respectively.} Hence, $W$ carries as much information as $X$ and $Y$ themselves, and $K_{GK}(X;Y) = I(X;Y) = H(X) = H(Y)$. At the other extreme, $f$ and $g$ could both be constant, mapping all values to a single label. This always satisfies the constraint $W = f(X) = g(Y)$ trivially, but gives $H(W) = 0$ and carries no information, as in the second example, where this was the only feasible option. The interesting case lies in between, where $f$ and $g$ define a coarse grouping of the values of $X$ and $Y$, respectively, into matching clusters. Then, $W$ simply carries the cluster index. The maximization of $H(W)$ searches for the finest such grouping that each of $X$ and $Y$ can determine on its own, i.e., the richest shared partition of the two alphabets that is consistent with the joint~distribution.

This intuition has a precise measure-theoretic counterpart in $K_{GK}(X;Y) = H(X \wedge Y)$, the entropy of the meet of the sigma-algebras generated by $X$ and $Y$. Here, recall that the sigma-algebra generated by a random variable is the collection of all events that can be determined from that variable alone. The meet of two sigma-algebras is the largest sigma-algebra contained within both, i.e., the greatest collection of events that can be determined from $X$ alone and simultaneously from $Y$ alone. Hence, the atoms of the meet form the finest partition of the sample space that both $X$ and $Y$ determine, and $H(X \wedge Y)$ measures how much randomness this partition carries. Since the meet is always unique, an optimal $W$ is unique up to relabeling, namely the index of the atoms of the meet. We call this optimal $W$ the common part of $X$ and $Y$ and write $W^*$. Figure~\ref{fig:sigma} illustrates the atoms of $\sigma(X)$, $\sigma(Y)$, their join, and their meet.\vspace{-6pt}
 \begin{figure}[H]
  \centering
   \includegraphics[width=\linewidth]{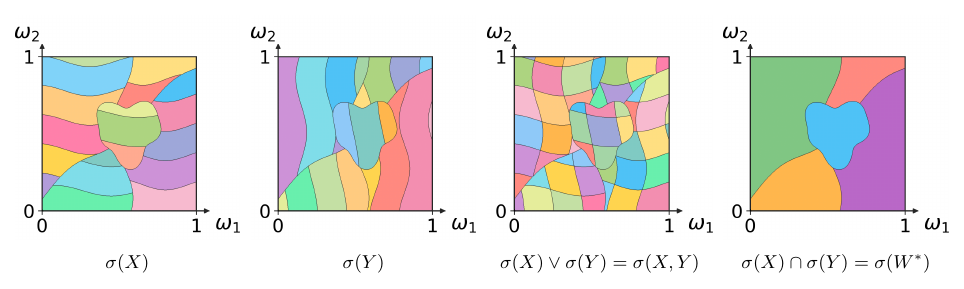}
\caption{{Partitions} 
 of the sample space $\Omega=[0,1]^2$, with outcome $\omega=(\omega_1,\omega_2)$, into the atoms of $\sigma(X)$ and $\sigma(Y)$, the sigma-algebras generated by $X$ and $Y$, of their join $\sigma(X)\vee\sigma(Y)=\sigma(X,Y)$, and of their meet $\sigma(X)\cap\sigma(Y)$. Each color marks one atom{, that is, the set of all outcomes at which the random variable generating the partition takes one particular value}. The join is the coarsest partition that refines both $\sigma(X)$ and $\sigma(Y)$. The meet is the finest partition coarser than both, and the five atoms of the meet index the values of the common part $W^*$.}
\label{fig:sigma}
\end{figure}

In the perception example of Figure~\ref{fig:perception}, an outcome $\omega$ is a detected object, and $X$ and $Y$ are the labels that Systems A and B assign to that object. Consider a sedan passing through the intersection. System A labels the sedan as car, so the object lies in the atom $\{X=\text{car}\}$ of $\sigma(X)$, which collects every object that System A calls a car. System B labels the same object as sedan, so the object lies in the atom $\{Y=\text{sedan}\}$ of $\sigma(Y)$. The join $\sigma(X,Y)$ records both labels, and the object lies in the atom $\{X=\text{car},\,Y=\text{sedan}\}$. The meet records only what both systems agree on. The object lies in the atom light vehicle, which equals $\{X=\text{car}\}$ and also equals $\{Y=\text{sedan}\}\cup\{Y=\text{SUV}\}$. Hence each system determines from its own label that the object is a light vehicle, while only System B can tell a sedan from an SUV. Figure~\ref{fig:sigma} shows generic partitions rather than the label sets of Figure~\ref{fig:perception}, so the numbers of atoms differ.

As hinted above, the Gács--Körner common information satisfies
\begin{linenomath}
\begin{equation}
    K_{GK}(X;Y) \leq I(X;Y),
    \label{eq:gk_mi}
\end{equation}
\end{linenomath}
{\color{black}with equality if and only if $X$ and $Y$ are conditionally independent given the common part $W^*$~\cite{gacs1973}.} In general, the inequality can be strict, and the gap can be arbitrarily large, as $K_{GK}(X;Y)$ can vanish entirely even when $I(X;Y)$ is arbitrarily large. This vanishing is not a technicality but a fundamental limitation of the Gács--Körner notion, and is precisely the observation encapsulated in the title of their paper, ``Common Information is Far Less Than Mutual Information''~\cite{gacs1973}.

\subsubsection{{Connection} 
 to Maximal Correlation}

The definition in~\eqref{eq:gk} relies on the Shannon entropy $H(W)$, which is infinite for any continuous random variable. Quantizing such a $W$ at resolution $\Delta$ gives an entropy of approximately $h(W) - \log\Delta$, where $h$ denotes differential entropy, and this value diverges as $\Delta \to 0$. Replacing $H(W)$ by $h(W)$ does not resolve the issue. The common part $W^*$ is defined only up to relabeling, since any invertible function of a common part is again a common part, whereas differential entropy is not invariant under such maps. For instance, $h(aW) = h(W) + \log|a|$ for any $a \neq 0$, so the maximum of $h(W)$ subject to $W = f(X) = g(Y)$ is unbounded. Differential entropy can also be negative. Hence, once $X$ and $Y$ share a continuous common part, the entropy objective no longer measures how rich the shared structure is. A more tractable question to study is whether a nontrivial common part exists at all.

Witsenhausen~\cite{witsenhausen1975} extended the Gács--Körner framework beyond finite alphabets to general measurable spaces and provided a structural characterization of when $K_{GK}(X;Y) > 0$. A natural quantity for this purpose is the Hirschfeld--Gebelein--Rényi \emph{maximal correlation}~\cite{hirschfeld1935,gebelein1941,renyi1959} 
\begin{equation}
\begin{split}
S(X,&Y) \triangleq ~\sup_{f,g} ~\mathbb{E}[f(X)g(Y)] \\
\text{s.t.} \quad &\mathbb{E}[f(X)] = \mathbb{E}[g(Y)] = 0, \\
&\mathbb{E}[f(X)^2] = \mathbb{E}[g(Y)^2] = 1.
\end{split}
\end{equation} Intuitively, $S(X,Y)$ measures how correlated functions of $X$ and $Y$ can be made. When $S(X,Y) = 1$ and the supremum is attained, there exist functions $f(X)$ and $g(Y)$ that are perfectly correlated. 

Witsenhausen showed that, under some mild conditions, this is precisely the condition for the joint distribution to decompose into a nontrivial block structure as
\begin{linenomath}
\begin{equation}
    K_{GK}(X;Y) > 0 \iff S(X,Y) = 1.
    \label{eq:gk_maxcorr}
\end{equation}
\end{linenomath}
These conditions are satisfied when at least one of $X$ and $Y$ has a finite alphabet, or when the Radon--Nikodym derivative of their joint distribution with respect to the product of their marginals is square-integrable under that product. Therefore, $S(X,Y) = 1$ provides a useful criterion for determining whether a nontrivial common random variable $W$ exists at all. However, a single pair of functions achieving $S(X,Y) = 1$ does not necessarily identify the \emph{optimal} $W^*$ of Gács and Körner, whose entropy equals $K_{GK}(X;Y)$. In the general case, such functions $f$ and $g$ specify a random variable $c(W^*)$ for some deterministic function $c$, which is not necessarily invertible. As a result, $c(W^*)$ may not convey the full Gács--Körner common information. Some authors, such as Lancaster~\cite{lancaster1957}, call $S(X,Y)$ the maximum
canonical correlation ({{the} name is not accidental; Section~\ref{sec:cca} 
explains the direct connection to CCA}). 

For jointly Gaussian $(X,Y)$ with Pearson correlation coefficient $\rho$, the maximal correlation is $S(X,Y) = |\rho|$~\cite{maung1941,lancaster1957}. Since $|\rho| < 1$ for any nondegenerate Gaussian pair, we have $K_{GK}(X;Y) = 0$, even though $I(X;Y) = -\tfrac{1}{2}\log(1-\rho^2)$ can be arbitrarily large. The Gács--Körner notion thus captures no common structure in jointly Gaussian
pairs. Interestingly, the analysis of Witsenhausen is geometric. $S(X,Y)$ is the cosine of the smallest angle between the space of zero-mean functions of $X$ and that of $Y$, and $K_{GK}(X;Y) > 0$ exactly when the two spaces share a direction~\cite{witsenhausen1975}. This is an early instance of the notion of common subspaces spanned by two or more views. As we will show in Sections~\ref{sec:cca}--\ref{sec:nonlinear}, common subspaces are central to the methods discussed in this tutorial.

\subsubsection{Extensions: Multiple Views and Relaxed Constraints}

The definition of common information naturally extends to $K$ random variables $X_1, \ldots, X_K$ as
\begin{linenomath}
\begin{equation}
    K_{GK}(X_1;\ldots;X_K) = \max H(W)
    \quad \text{s.t.} \quad
    W = f_k(X_k) \; \forall\, k,
    \label{eq:gk_K}
\end{equation}
\end{linenomath}
where $W$ must be simultaneously recoverable as a deterministic function of each $X_k$ alone. Tyagi et~al.\ \cite{tyagi2011when} show that this maximum is attained by a recursive construction. They define the \emph{maximum common function} $\mathrm{mcf}(X_1,\ldots,X_K)$ by applying the two-variable case repeatedly as $\mathrm{mcf}(X_1,\ldots,X_K) = \mathrm{mcf}(\mathrm{mcf}(X_1,\ldots,X_{K-1}), X_K)$, and show that any random variable recoverable from each $X_k$ individually has entropy at most $H(\mathrm{mcf}(X_1,\ldots,X_K))$. The same structural characterization applies as before, as $K_{GK}(X_1;\ldots;X_K) > 0$ if and only if the joint distribution decomposes into a common block structure across all $K$ variables, and vanishes otherwise.

To the best of our knowledge, the extension to $K$ continuous views was first considered by Zhang et~al.~\cite{zhang2024}, directly for random vectors rather than scalar random variables. They propose a practical framework for discovering the Gács--Körner decomposition from high-dimensional continuous samples. Their approach reframes the problem as one of disentanglement ({{we} will cover more related works in Section \ref{sec:nonlinear}}). 
 Specifically, each view $X_k$ is decomposed into a common component $C$, shared across all views, and a view-specific unique component $U_k$, where $C$ satisfies the deterministic Gács--Körner constraint $C = \phi_k(X_{k})$ for all $k$, $C$ and all $U_k$ are mutually independent, and each $X_k=f_k(C, U_k)$ for invertible functions $f_k$. The independence constraint is enforced by minimizing the total correlation $\mathrm{TC}(C, U_1, \ldots, U_K)$, estimated using a matrix-based R\'{e}nyi $\alpha$-order entropy functional~\cite{yu2021measuring} that avoids distributional estimation and variational approximation. 

A complementary relaxation, in the discrete setting, is due to Salamatian et~al.~\cite{salamatian2020}, who propose an approximate relaxation that allows a disagreement cost $\delta$ between the function $f(X_1)$ and $X_2$, measured by $H(f(X_1) \mid X_2) \leq \delta$. Their formulation recovers the original Gács--Körner formulation at $\delta=0$, while their motivation arose from applications in distributed compression. Kleinman et~al.~\cite{kleinman2023} instead target the two-view continuous setting and relax the deterministic constraint itself, replacing $W = f(X_1) = g(X_2)$ with the requirement that the conditional distribution of $W$ given each view is the same, i.e., $p(W|X_1) = p(W|X_2)$ for all $(X_1, X_2) \sim p(X_1, X_2)$. The resulting stochastic Gács--Körner common information is
\begin{linenomath}
\begin{equation}
    \tilde{K}_{GK}(X_1; X_2) = \max_{W} I(X_k; W)
    \quad \text{s.t.} \quad
    p(W|X_1) = p(W|X_2),
    \label{eq:gk_stochastic}
\end{equation}
\end{linenomath}
where the two mutual information values $I(W; X_1)$ and $I(W; X_2)$ are equal under the constraint, so either can serve as the objective. When $W$ is a deterministic function of $X_k$, we have $H(W|X_k) = 0$, so $I(X_k; W) = H(W)$ and the original Gács--Körner formulation is recovered as a special case. This relaxation is amenable to optimization via a variational autoencoder, where the stochastic constraint is enforced by penalizing the KL divergence between the posteriors $p(W|X_1)$ and $p(W|X_2)$, enabling learning from high-dimensional continuous data.

\subsubsection{Applications}
The Gács--Körner common information has found applications across several areas. In distributed source coding, the common part of two sources needs to be transmitted only once, which Salamatian et~al.~\cite{salamatian2020} exploit to build low-complexity alternatives to Slepian--Wolf coding. In cryptography, Tyagi et~al.~\cite{tyagi2011when} use the maximum common function to characterize the capacity of an aided secret key model, which in turn decides when a function of correlated signals can be computed without revealing its value to an eavesdropper. In machine learning, the Gács--Körner framework provides a principled information-theoretic foundation for multi-view representation learning, where the goal is to separate shared from view-specific factors of variation~\cite{kleinman2023, zhang2024}. Hence, although the notion dates back to 1973, it still receives attention from the research community, which keeps extending the theory built on it and developing tools that apply it to modern problems.

\subsubsection{Limitations and the Need for a Richer Notion}
The Gács--Körner common information, while mathematically elegant, captures only deterministic shared structure and is blind to any probabilistic common information. As discussed above, $K_{GK}(X;Y) = 0$ for most distributions of practical interest, including all nondegenerate jointly Gaussian pairs, even when $I(X;Y)$ is arbitrarily large. This fundamental limitation motivates the search for a richer notion of common information, one that allows the shared variable $W$ to be probabilistically rather than deterministically related to each view. Instead of requiring $W = f(X) = g(Y)$ exactly, one could ask for the minimum-rate description of $W$ that renders $X$ and $Y$ conditionally independent. This is precisely the question that Wyner's common information, introduced in the next subsection,~addresses.

\subsection{Wyner's Common Information}
\label{sec:wyner}

Wyner~\cite{wyner1975} proposed a more flexible definition. Instead of requiring the common variable to be computable from each observation alone, Wyner required a common factor that \emph{explains} the dependence between them.

\begin{Definition}[Wyner's common information~\cite{wyner1975}]
The common information of two discrete jointly distributed random variables $X$ and $Y$ is
\begin{linenomath}
\begin{equation}
  C(X;Y) = \min_{W:\;X \perp Y \mid W} I(X,Y;\,W),
  \label{eq:wyner}
\end{equation}
\end{linenomath}
where the minimum is over all random variables $W$ such that $X$  and $Y$ are conditionally independent given $W$, and is achieved by some $W$ with alphabet size $|\mathcal{W}| \leq |\mathcal{X}| \cdot |\mathcal{Y}|$.
\end{Definition}

To see what this objective measures, note that under the constraint $X \perp Y \mid W$, it can be shown that
\begin{linenomath}
\begin{equation}
I(X,Y;\,W) = I(X;W) + I(Y;W) - I(X;Y).
\end{equation}
\end{linenomath}
Since $I(X;Y)$ does not depend on the choice of $W$, minimizing $I(X,Y;W)$  is equivalent to minimizing $I(X;W) + I(Y;W)$, i.e., the total amount  of information that $W$ carries about both random variables. Therefore, Wyner's formulation of common information finds the most parsimonious $W$ that renders $X$ and $Y$ conditionally independent, since such a $W$ fully explains their dependence while carrying as little information as possible.

From a linear-algebra perspective, it can be helpful to think of the alphabet size of $W$, $|\mathcal{W}|$, as the number of rank-one terms of a nonnegative matrix factorization (NMF) of the joint probability mass function (PMF) $P_{XY}(x,y)$. Specifically, the conditional independence condition $X \perp Y \mid W$ corresponds to 
\begin{linenomath} 
\begin{equation}
    P_{XY}(x,y) = \sum_{w} P_W(w)\, P_{X|W}(x\mid w)\, P_{Y|W}(y\mid w),
\end{equation}
\end{linenomath}
which is precisely an NMF of $P_{XY}$ with rank $|\mathcal{W}|$. Let $\text{rank}_+(P_{XY})$ be the minimum $|\mathcal{W}|$ for which such a factorization exists, i.e., the nonnegative rank of $P_{XY}$. It is known that $\text{rank}_+(P_{XY}) \leq \min(|\mathcal{X}|,|\mathcal{Y}|)$, since setting $W = X$ (or $Y$) always yields a feasible decomposition with $|\mathcal{X}|$ (or $|\mathcal{Y}|$) rank-one terms. However, a $W$ achieving the minimum in Wyner's definition does not need to coincide with the one achieving the minimum nonnegative rank, as the former minimizes $I(X,Y;W)$ over all valid decompositions, while the latter minimizes only $|\mathcal{W}|$. A detailed treatment of the connections between common information and the nonnegative rank can be found in~(\cite{yu2022common}{,} Section~7). 

The NMF of $P_{XY}$ is not unique in general. Uniqueness, up to relabeling of $W$, holds when both factor matrices are sufficiently scattered, a condition that favors sparse conditional PMFs $P_{X|W}$ and $P_{Y|W}$~\cite{huang2014,fu2019}. Checking whether given factor matrices are sufficiently scattered is NP-complete in general~(\cite{huang2014}, Remark~2), although simple necessary conditions, such as zero entries in the factor matrices, are easy to verify. These conditions compare factorizations with the same number of rank-one terms, whereas Wyner's problem compares all valid factorizations through $I(X,Y;W)$, so NMF uniqueness alone does not certify that a minimizer of~\eqref{eq:wyner} is unique.

The constraint $X \perp Y \mid W$ also yields a lower bound on $C(X;Y)$. For every feasible $W$, the chain rule and the data-processing inequality give $I(X,Y;W) \ge I(X;W) \ge I(X;Y)$. Minimizing over all feasible $W$ in~\eqref{eq:wyner} then gives $C(X;Y) \ge I(X;Y)$. Together with~\eqref{eq:gk_mi}, the two notions of common information are therefore ordered as
\begin{linenomath}
\begin{equation}
  K_{GK}(X;Y) \leq I(X;Y) \leq C(X;Y).
  \label{eq:ci_chain}
\end{equation}
\end{linenomath}
To see when the two notions coincide, consider adding the Gács--Körner constraint \linebreak $W = f(X) = g(Y)$ to~\eqref{eq:wyner}. For any $W$ that satisfies both constraints, we have that \linebreak $I(X;Y) = H(W) + I(X;Y|W) = H(W) \le K_{GK}(X;Y) \le I(X;Y)$, which forces \linebreak $H(W) = K_{GK}(X;Y)$. Hence, $W = W^*$ up to relabeling, and such a $W$ exists only in the equality case of~\eqref{eq:gk_mi}, namely when $X$ and $Y$ are conditionally independent given $W^*$. In that case, $W^*$ is feasible for~\eqref{eq:wyner} with $I(X,Y;W^*) = H(W^*) = I(X;Y)$, so equality holds throughout~\eqref{eq:ci_chain}. 

{\color{black}An equivalent characterization is due to M.~Kaplan, as reported in~(\cite{wyner1975}, Remark~F). Equality holds throughout~\eqref{eq:ci_chain} if and only if it is possible to write $X = (X',V)$ and $Y = (Y',V)$, where $X'$ and $Y'$ are conditionally independent given $V$. 
Here, $X = (X',V)$ means that each value of $X$ corresponds to a unique pair of values of $X'$ and $V$, and similarly for $Y$. Hence, $V$ is computable from each of $X$ and $Y$ and plays the role of the common part $W^*$, while $X'$ and $Y'$ are the private parts of the two views. In terms of the NMF presented above, after permuting its rows and columns, $P_{XY}$ corresponds to a block-diagonal decomposition and on each block admits a product-distribution representation, as illustrated in Figure~\ref{fig:block_diag}. The rank-one terms of the factorization then have pairwise disjoint row and column supports, and $W^*$ indexes them, so each of $X$ and $Y$ can determine $W^*$ from its own value. Outside this special case, the inequalities are strict.}

\vspace{-12pt}
\begin{figure}[H]
\centering
\includegraphics[width=0.8\textwidth]{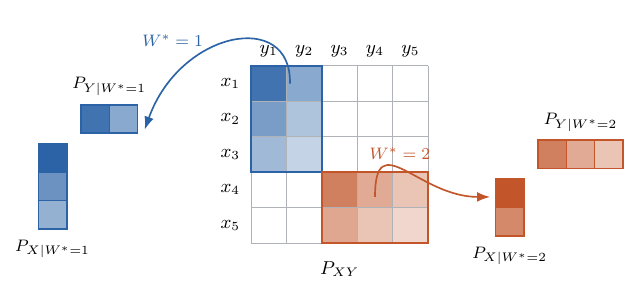}
\caption{{\color{black}{Joint} 
 PMF $P_{XY}$ in the case $K_{GK}(X;Y) = I(X;Y) = C(X;Y)$, with rows and columns permuted so that $P_{XY}$ is block-diagonal{. The block in blue corresponds to $W^*=1$, the block in orange to $W^*=2$, and darker shades indicate larger probabilities}. Each block equals $P_{W^*}(w)$ times the outer product of the conditional PMFs $P_{X|W^*=w}$ and $P_{Y|W^*=w}$, as the arrows indicate, which reflects the conditional independence of $X$ and $Y$ given $W^*$. The two blocks occupy disjoint rows and disjoint columns, so each of $X$ and $Y$ determines the common part $W^*$ from its own value. Hence, $X$ always reveals the block in which $Y$ lies, but nothing more about $Y$, and vice versa.}}
\label{fig:block_diag}
\end{figure}

\subsubsection*{Example: the Doubly Symmetric Binary Source}

A concrete case in which the three quantities differ is the \emph{doubly symmetric binary source} (DSBS), taken from~\cite{wyner1975}. Let $X$ and $Y$ be binary random variables related by a binary symmetric channel (BSC) with crossover probability $a_0 \in \left(0, \frac{1}{2}\right)$, as appears in Figure~\ref{fig:dsbs}. Assume that $X$ is a uniform bit and that $Y = X \oplus Z$, where $\oplus$ stands for modulo-2 addition and $Z \sim \mathrm{Bernoulli}(a_0)$ is independent of $X$. In this case, $H(X) = H(Y) = 1$ and $I(X;Y) = 1 - h(a_0)>0$, where $h(\lambda) = -\lambda \log \lambda - (1-\lambda)\log(1-\lambda)$ is the binary entropy function. Since all four pairs $(x,y)$ occur with positive probability, no nonconstant functions $f(X)$ and $g(Y)$ can agree with probability one. Hence, the Gács--Körner common information vanishes, i.e., $K_{GK}(X;Y) = 0$, even though $I(X;Y) > 0$.

\vspace{-6pt}
\begin{figure}[H]
\centering
\includegraphics[width=0.5\textwidth]{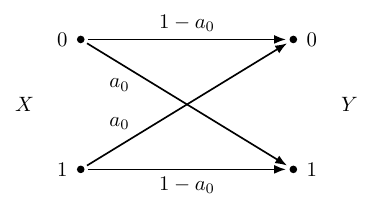}\vspace{-6pt}
\caption{The binary symmetric channel with crossover probability $a_0$ that relates $X$ and $Y$ in the doubly symmetric binary source. Each bit of $X$ is received correctly as $Y$ with probability $1-a_0$ and flipped with probability $a_0$.}
\label{fig:dsbs}
\end{figure}

Wyner showed that an optimal $W$, for this setup, is the intermediate variable in a cascade of two BSCs, each with crossover probability $a_1 = \tfrac{1}{2}(1 - \sqrt{1 - 2a_0})$, chosen so that the cascade is equivalent to a single BSC with crossover probability $a_0$, as shown in Figure~\ref{fig:wyner_dsbs}. In terms of the channel noise, the cascade reads
\begin{linenomath}
\begin{equation}
    W = X \oplus U, \qquad Y = W \oplus V,
\end{equation}
\end{linenomath}
where $U$ and $V$ are independent $\mathrm{Bernoulli}(a_1)$ random variables, independent of $X$. Hence, $Y = X \oplus (U \oplus V)$, where $U \oplus V \sim \mathrm{Bernoulli}\big(2a_1(1-a_1)\big)$ and $2a_1(1-a_1) = a_0$, which recovers the DSBS. The resulting common information is
\begin{linenomath}
\begin{equation}
    C(X;Y) = 1 + h(a_0) - 2h(a_1).
\end{equation}
\end{linenomath}
One can verify that $C(X;Y) > I(X;Y)$ for all $0 < a_0 < \frac{1}{2}$, providing a concrete instance where the inequality is strict. As $a_0 \to 0$, $X$ and $Y$ become identical and $C(X;Y) \to 1 = H(X) = H(Y)$, while as $a_0 \to \frac{1}{2}$, $X$ and $Y$ become independent and $C(X;Y) \to 0$. Hence, the DSBS illustrates the full chain $K_{GK}(X;Y) = 0 < I(X;Y) < C(X;Y)$ in a single concrete example. Finally, note that the optimal $W$ is not a function of $(X,Y)$. Given $X = Y = 0$, for instance, $W = 0$ when $U = V = 0$ and $W = 1$ when $U = V = 1$, and both events have positive probability.

Recall that the Gács--Körner definition required $W$ to be deterministically recoverable from each of $X$ and $Y$ individually, and preferably nontrivial. An interesting question is to ask what happens if we restrict Wyner's formulation to $W$'s that are \emph{jointly} determined by $X$ and $Y$, i.e., there exists a deterministic function $f$ such that $W = f(X,Y)$, and $W$ renders $X$ and $Y$ conditionally independent. Under this assumption, the conditional entropy $H(W|X,Y)$ vanishes, and the mutual information reduces to
\[
I(X,Y;W) = H(W) - H(W|X,Y) = H(W),
\]
so that Wyner's objective becomes an entropy and the problem reads\vspace{12pt}
\begin{linenomath}
\begin{equation}
\begin{split}
      G_f(X;Y) &= \min_{W} H(W)\\
      \text{ s.t. }  \;X \perp  Y &\mid W,\; W = f(X,Y).
\end{split}
\label{exact_wyner}
\end{equation}
\end{linenomath}
Among all functions of $(X,Y)$ that render $X$ and $Y$ conditionally independent, $G_f(X;Y)$ is the entropy of the most parsimonious one. Note that this is not $C(X;Y)$. The constraint $W = f(X,Y)$ shrinks the feasible set of~\eqref{eq:wyner}, so $G_f(X;Y) \ge C(X;Y)$, and the inequality can be strict.

\vspace{-6pt}
\begin{figure}[H]
\centering
\includegraphics[width=0.8\textwidth]{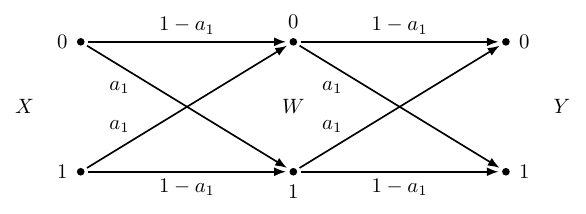}\vspace{-6pt}
\caption{Wyner's optimal common variable $W$ for the doubly symmetric binary source, shown as the intermediate variable in a cascade of two binary symmetric channels, each with crossover probability $a_1 = \tfrac{1}{2}\big(1-\sqrt{1-2a_0}\big)$. Since $2a_1(1-a_1) = a_0$, the cascade is equivalent to the single binary symmetric channel of Figure~\ref{fig:dsbs}, and $X$ and $Y$ are conditionally independent given $W$.}
\label{fig:wyner_dsbs}
\end{figure}

As the doubly symmetric binary source shows, Wyner's optimal $W$ need not
be a function of $(X,Y)$ at all. Dropping that constraint but keeping the entropy objective gives the \emph{common entropy}
\begin{linenomath}
\begin{equation}
  G_{CE}(X;Y) = \min_{W:\;X \perp Y \mid W} H(W) ,
  \label{eq:exact_ci}
\end{equation}
\end{linenomath}
of Kumar et~al.~\cite{kumar2014}. The common entropy is within one bit of the \emph{exact common information}, the minimum description length of the common randomness needed to generate $(X,Y)$ exactly. {{Kumar} et~al.~\cite{kumar2014} measure the exact common information by the minimum expected length of a prefix-free code for $W$, which lies between $G_{CE}(X;Y)$ and \mbox{$G_{CE}(X;Y) + 1$}. The extra bit is the cost of integer codeword lengths, as in Shannon's zero-error source coding theorem, and becomes negligible per symbol when many realizations are encoded jointly.} Since $I(X,Y;W) \le H(W)$ for every $W$, and $G_f$ minimizes the same entropy over a subset of the feasible set of $G_{CE}$, the chain~\eqref{eq:ci_chain} extends to $K_{GK}(X;Y) \le I(X;Y) \le C(X;Y) \le G_{CE}(X;Y) \le G_f(X;Y)$. In summary, what the views can extract in common lies below $I(X;Y)$, what is needed to explain their dependence lies above $I(X;Y)$, and the simplest explanation of their dependence that both views determine jointly lies highest of all. 

In the equality case of~\eqref{eq:ci_chain}, $W^*$ is also feasible for~\eqref{exact_wyner} and~\eqref{eq:exact_ci} with objective value $H(W^*) = I(X;Y)$, so $K_{GK}(X;Y) = I(X;Y) = C(X;Y) = G_{CE}(X;Y) = G_f(X;Y)$. Gács and Körner maximize the entropy of a common function, while Wyner minimizes the information carried by a factor that explains the dependence, and the two meet exactly when the common part explains the entire dependence~(\cite{csiszar2011}, Chapter~16).

\subsection{The Gray--Wyner Network: An Operational Meaning}

Wyner showed that $C(X;Y)$, beyond its mathematical elegance, also has a precise operational meaning through the \emph{Gray--Wyner network}~\cite{grawyner1974}, illustrated in Figure~\ref{fig:gray_wyner}. {{A quantity} is said to have an \emph{operational meaning} when it is the answer to a concrete coding problem rather than only the value of an optimization problem one has written down. The standard example is the entropy $H(X)$, defined as $-\sum_x p(x)\log p(x)$ and given operational meaning by Shannon's source coding theorem as the minimum rate at which $X$ can be losslessly compressed. Wyner's $C(X;Y)$ is defined by the minimization in~\eqref{eq:wyner}. This section gives $C(X;Y)$ an operational meaning as a rate in the Gray--Wyner network}. 

\begin{figure}[H]
\centering
\begin{tikzpicture}[
    block/.style={rectangle, draw, thick, rounded corners, 
                  minimum width=1.8cm, minimum height=1cm, align=center},
    arrow/.style={->, thick},
    line/.style={thick}
]
\node[block] (enc) {Encoder};
\node[block, right=4cm of enc, yshift=2cm]  (dec1) {Receiver 1};
\node[block, right=4cm of enc, yshift=-2cm] (dec2) {Receiver 2};
\node[left=1.5cm of enc] (input) {$(X^n, Y^n)$};
\draw[arrow] (input) -- (enc);
\draw[arrow] (enc.north) -- ++(0,1.5cm) -- node[above, pos=0.5] {$R_1$} (dec1.west);
\coordinate (split) at ($(dec1.south)!0.5!(dec2.north)$);
\draw[line] (enc.east) -- node[above] {$R_0$} (enc.east -| split);
\draw[arrow] (enc.east -| split) -- (dec1.south);
\draw[arrow] (enc.east -| split) -- (dec2.north);
\draw[arrow] (enc.south) -- ++(0,-1.5cm) -- node[below, pos=0.5] {$R_2$} (dec2.west);
\node[right=1.5cm of dec1] (xout) {$X^n$};
\node[right=1.5cm of dec2] (yout) {$Y^n$};
\draw[arrow] (dec1) -- (xout);
\draw[arrow] (dec2) -- (yout);
\end{tikzpicture}
\caption{The Gray--Wyner source coding network. An encoder observes $(X^n, Y^n)$ and describes them to two receivers using a common link of rate $R_0$ and two private links of rates $R_1$ and $R_2$. Wyner's common information $C(X;Y)$ is the minimum $R_0$ among achievable rate triples with \mbox{$R_0 + R_1 + R_2 = H(X,Y)$.}
\label{fig:gray_wyner}}
\end{figure}
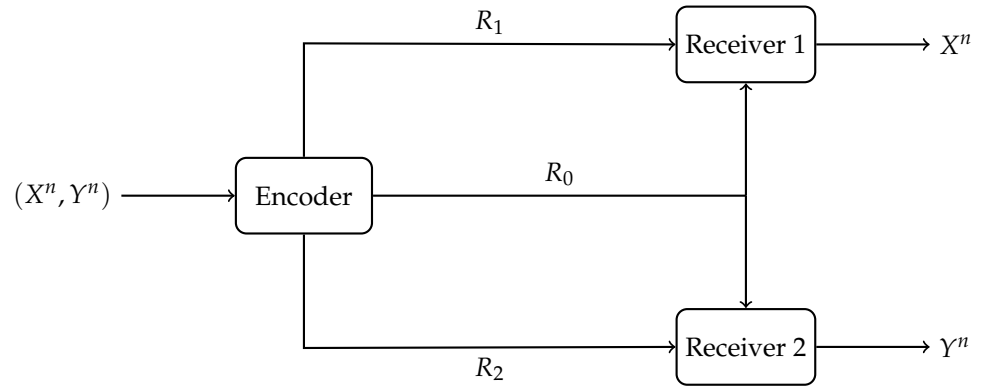

Given $n$ i.i.d.\ jointly observed samples of $(X,Y)$, we consider the problem of encoding them for lossless reconstruction at two separate receivers. Receiver~1 wishes to reconstruct $X^n$ and Receiver~2 wishes to reconstruct $Y^n$. The two receivers are connected to a common transmitter via a network of three channels. One common channel is shared by both receivers, and two private channels are dedicated to the two receivers, one for each. The rate of each channel, measured in bits per sample, is defined as $R_0$ for the common channel and $R_1$, $R_2$ for the private channels. Receiver~1 thus obtains $R_0 + R_1$ bits per sample and Receiver~2 obtains $R_0 + R_2$, while the transmitter sends $R_0 + R_1 + R_2$ in total. By Shannon's source coding theorem, the optimal total rate for lossless reconstruction is $H(X,Y)$ bits per sample. The interest in this network lies in the fact that the total rate can remain at its minimum $H(X,Y)$ while each receiver obtains less than $H(X,Y)$, which saves bandwidth, storage and decoding effort at the receivers at no extra cost to the transmitter. The goal is to characterize all achievable rate triples $(R_0, R_1, R_2)$ that allow lossless reconstruction. Among these, we seek, at the minimum total rate, the encoding that keeps the common rate $R_0$, which represents the information that must necessarily be shared between the two receivers, as small as possible.

Gray and Wyner showed that $H(X,Y)$ remains optimal in the distributed setting described above, since every achievable triple satisfies $R_0 + R_1 + R_2 \ge H(X,Y)$, and some triples attain equality~(\cite{grawyner1974}, Theorems~2 and~3). Since the common channel is the only one received by both receivers, $R_0$ is precisely the rate of information shared between them, while $R_1$ and $R_2$ carry information private to each receiver. Returning to the question raised above, namely the smallest common rate $R_0$ that allows lossless reconstruction at the minimum total rate, Wyner showed that the answer is precisely $C(X;Y)$,
\begin{linenomath}
\begin{equation}
  \min\{R_0 : (R_0, R_1, R_2) \text{ achievable at total rate }
  H(X,Y)\} = C(X;Y),
\end{equation}
\end{linenomath}
revealing that $C(X;Y)$ is the irreducible amount of information that must be shared between the two receivers for lossless reconstruction at the minimum total rate.

\subsection{Extension of Wyner's Common Information to Multiple Variables}
\label{sec:wyner_K}

The two-variable definition of Wyner's common information admits a natural extension to $K$ random variables. Liu, Xu, and Chen~\cite{liu2010} proposed defining the common information of $K$ jointly distributed random variables $X_1, \ldots, X_K$ by imposing the conditional independence structure through the introduction of an auxiliary random variable~$W$.

\begin{Definition}[$K$-variable Wyner common 
information~\cite{liu2010}]
The common information of $K$ discrete jointly distributed 
random variables $X_1, \ldots, X_K$ is
\begin{linenomath}
\begin{equation}
  C(X_1; \ldots; X_K) 
  = \inf_{W} ~I(X_1, \ldots, X_K;\, W),
  \label{eq:wyner_K}
\end{equation}
\end{linenomath}
where the infimum is over all joint distributions of $(X_1, \ldots, X_K, W)$ that preserve the marginal $P_{X_1 \cdots X_K}(x_1, \ldots, x_K)$ and satisfy the full conditional independence
\begin{linenomath}
\begin{equation}
  P_{X_1 \cdots X_K|W}(x_1, \ldots, x_K \mid w) 
  = \prod_{k=1}^K P_{X_k|W}(x_k \mid w).
  \label{eq:wyner_K_ci}
\end{equation}
\end{linenomath}
\end{Definition}

For $K=2$, the full conditional independence $P_{X_1, X_2|W}(x_1, x_2 \mid w) = P_{X_1|W}(x_1 \mid w) P_{X_2|W}(x_2 \mid w)$ reduces to $X_1 \perp X_2 \mid W$, recovering Wyner's original definition. As in the two-variable case, the conditional independence constraint~\eqref{eq:wyner_K_ci} admits an interpretation as a nonnegative tensor factorization (NTF). The joint PMF $P_{X_1 \cdots X_K}(x_1, \ldots, x_K)$ can be viewed as a $K$-way nonnegative tensor of dimensions $|\mathcal{X}_1| \times \cdots \times |\mathcal{X}_K|$, and the conditional independence condition corresponds to the canonical polyadic (CP) decomposition~\cite{hitchcock1927,carroll1970,harshman1970}
\begin{linenomath}
\begin{equation}
  P_{X_1 \cdots X_K}(x_1, \ldots, x_K)
  = \sum_{w=1}^{|\mathcal{W}|} P_W(w) \prod_{k=1}^K P_{X_k|W}(x_k \mid w),
\end{equation}
\end{linenomath}
with $|\mathcal{W}|$ rank-one factors. As noted by Lim and Comon~\cite{lim2009}, this tensorial generalization of NMF corresponds probabilistically to the naive Bayes model, shown in Figure~\ref{fig:naive_bayes}, in which $X_1, \ldots, X_K$ are conditionally independent given a latent variable $W$. The nonnegative rank $\mathrm{rank}_+(P_{X_1 \cdots X_K})$ is the minimum $|\mathcal{W}|$ for which such a decomposition exists.

\begin{figure}[H]
\centering
\begin{tikzpicture}[
  every node/.style={font=\small},
  lat/.style={circle,draw,thick,minimum size=9mm,inner sep=0pt},
  obs/.style={circle,draw,thick,fill=black!8,minimum size=9mm,inner sep=0pt}]
\node[lat] (W)  at (0,1.7) {$W$};
\node[obs] (X1) at (-3,0)  {$X_1$};
\node at (-1.5,0) {$\cdots$};
\node[obs] (Xk) at (0,0)   {$X_k$};
\node[obs] (XK) at (3,0)   {$X_K$};
\node at (1.5,0) {$\cdots$};
\draw[->,thick] (W) -- (X1);
\draw[->,thick] (W) -- (Xk);
\draw[->,thick] (W) -- (XK);
\end{tikzpicture}
\caption{The naive Bayes model. A single latent variable $W$ generates every
observation, and no edge joins the $X_k$ to one another, so they are mutually
independent once $W$ is given. This is the conditional
independence~\eqref{eq:wyner_K_ci} that Wyner's construction requires, and
writing it out on a discrete joint PMF gives the nonnegative tensor
factorization above, with one rank-one factor per value of $W$.
\label{fig:naive_bayes}}
\end{figure}
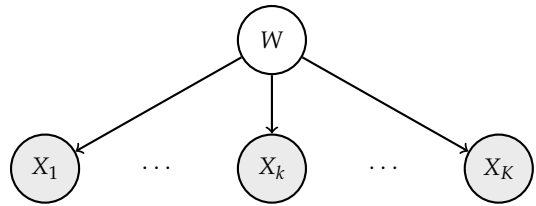

As in the two-variable case, for an optimal $W^*$ of~\eqref{eq:wyner_K} we have $|\mathcal{W}^*| \leq \prod_{k=1}^K |\mathcal{X}_k|$, and combined 
with the above, we obtain
\begin{linenomath}
\begin{equation}
  \mathrm{rank}_+(P_{X_1 \cdots X_K}) \leq |\mathcal{W}^*| \leq \prod_{k=1}^K |\mathcal{X}_k|.
  \label{NTF_rank}
\end{equation}
\end{linenomath}
Interestingly, $\mathrm{rank}_+(P_{X_1 \cdots X_K})$ admits a tighter upper bound than $\prod_{k=1}^K |\mathcal{X}_k|$. To see this, assume without loss of generality that $|\mathcal{X}_1| \leq \cdots \leq |\mathcal{X}_K|$ and write
\begin{linenomath}
\small
\begin{equation*}
  P(x_1, \ldots, x_K) = \sum_{x_1' \in \mathcal{X}_1} \cdots \hspace{-4mm}
  \sum_{x_{K-1}' \in \mathcal{X}_{K-1}} \hspace{-4mm}P(x_1', \ldots, x_{K-1}')
  \prod_{k=1}^{K-1} \mathbf{1}[x_k = x_k'] \cdot 
  P(x_K \mid x_1', \ldots, x_{K-1}'),
\end{equation*}
\end{linenomath}
where the indices have been dropped from $P$ for space economy. Each term in this sum is a weighted rank-one conditional PMF. By letting $P_W(w) = P_{X_1 \cdots X_{K-1}}(x_1', \ldots, x_{K-1}')$, we obtain an implicit random variable $W$ which takes at most $\prod_{k=1}^{K-1} |\mathcal{X}_k|$ values. Since this holds for any subset of $K-1$ alphabets and $X_K$ has the largest alphabet, we get $\mathrm{rank}_+(P_{X_1 \cdots X_K}) \leq \prod_{k=1}^{K-1} |\mathcal{X}_k|$. As a result, the two ends of \eqref{NTF_rank} are always separated by a factor of at least $|\mathcal{X}_K|$.

Another interesting observation is that the CP decomposition can be essentially unique, i.e., unique up to permutations of the rank-one factors, under mild sufficient conditions~\cite{kruskal1977, sidiropoulos2000, domanov2017, kargas2018}. The most well-known of these is Kruskal's condition~\cite{kruskal1977, sidiropoulos2000}, which states that uniqueness is guaranteed when the Kruskal ranks $\kappa_k$ of the factor matrices $\mathbf{M}_k(x_k, w):=P_{X_k|W}(x_k \mid w)$, $k = 1, \ldots, K$, satisfy
\begin{linenomath}
\begin{equation}
\sum_{k=1}^K \kappa_k \geq 2|\mathcal{W}| + K - 1.
\end{equation}
\end{linenomath}
The Kruskal rank of a factor matrix $\mathbf{M}_k$ is defined as the largest integer $\kappa_k$ such that every subset of columns of size $\kappa_k$ is linearly independent. Hence, it can be at most equal to $\min(|\mathcal{X}_k|, |\mathcal{W}|)$. This maximum is attained almost surely for factors whose entries are sampled independently from absolutely continuous distributions on $\mathbb{R}_{+}$, with each column subsequently normalized by its sum. This construction covers most distributions of practical interest, essentially excluding only degenerate cases where mass concentrates on a lower-dimensional set. Indeed, since normalization does not affect linear dependence, the Kruskal rank of $\mathbf{M}_k$ matches that of its unnormalized counterpart. For the latter, linear dependence of any $\min(|\mathcal{X}_k|, |\mathcal{W}|)$ columns forces a nontrivial determinantal polynomial to vanish~\cite{stegeman2006}, which happens with probability zero.

While the definition is a direct and elegant generalization, it has a counterintuitive property. One would expect that the common information shared among a growing set of variables should not increase. Intuition suggests that something truly common to all variables cannot increase as more variables are included, much like the intersection of sets cannot increase as more sets are added. However, $C(X_1; \ldots; X_K)$ is monotonically \emph{non-decreasing} in $K$~\cite{liu2010}. This follows from the observation that if $X_1, \ldots, X_K$ are conditionally independent given $W$, then any subset $X_{i_1}, \ldots, X_{i_L}$ is also conditionally independent given $W$, but not conversely. Therefore, the feasible set of $W$'s can only remain the same or shrink as more variables are added. Moreover, for any $W$ feasible for the full set, $I(X_{i_1}, \ldots, X_{i_L}; W) \le I(X_1, \ldots, X_K; W)$, since adding variables cannot decrease mutual information. Combining these two observations leads to
\begin{linenomath}
\begin{equation}
  C(X_{i_1}; \ldots; X_{i_L}) 
  \leq C(X_1 ;\ldots; X_K) 
  \quad \text{for any } 
  \{i_1, \ldots, i_L\} \subseteq \{1, \ldots, K\}.
\end{equation}
\end{linenomath}
This monotonicity arises because the full conditional independence constraint~\eqref{eq:wyner_K_ci} becomes more demanding as $K$ grows, forcing $W$ to carry a richer description of the joint distribution. As a concrete illustration, provided in~\cite{liu2010}, consider $X_1 = (X', U, V)$, $X_2 = (Y', V, T)$, and $X_3 = (Z', T, U)$, where $X', Y', Z', U, V, T$ are mutually independent. Then $C(X_1; X_2) = H(V)$, $C(X_1; X_3) = H(U)$, and $C(X_2; X_3) = H(T)$, but $C(X_1; X_2; X_3) = H(U, V, T) = H(U) + H(V) + H(T)$, which exceeds every pairwise value.

It is worth noting that the Gács--Körner common information, discussed in Section~\ref{sec:gk}, does exhibit the expected non-increasing behavior when extended to $K$ variables. Adding more variables can only further constrain the deterministic function that all parties must agree on, so $K_{GK}(X_1; \ldots; X_K)$ can only shrink or remain the same as $K$ grows. Thus the two classical notions of common information exhibit opposite monotonicity in the multivariable setting. The Gács--Körner quantity captures the intersective, shared-structure view, while generalized Wyner common information captures the cost of fully explaining all dependencies, which naturally grows with the number of variables. 
In the next subsection, we present the work of Tandon, Sankar, and Poor~\cite{tandon2011}, motivated by multi-user privacy in the Gray--Wyner network, whose resulting common information is monotonically non-increasing in $K$ and can be viewed as a natural generalization of the Gács--Körner common information to $K$ variables.

\subsection{Common Information Under Multi-User Network Privacy}
\label{sec:poor_ext}

The Gray--Wyner network extends naturally to $K$ jointly observed sources $X_1, \ldots, X_K$ at a single encoder and $K$ separate receivers, where Receiver~$k$ losslessly reconstructs $X_k$ from the common channel of rate $R_0$ and a private channel of rate $R_k$. In the same work that introduced~\eqref{eq:wyner_K}, Liu, Xu, and Chen~\cite{liu2010} showed that $C(X_1; \ldots; X_K)$ is the minimum common rate $R_0$ among achievable rate tuples with total rate $R_0 + \sum_{k=1}^K R_k = H(X_1, \ldots, X_K)$. In many practical settings, however, privacy across users is also a concern. Tandon, Sankar, and Poor~\cite{tandon2011} studied this network under a privacy requirement, measured by the equivocation $E_k$ of Receiver~$k$, i.e., the per-sample uncertainty that remains about the sources other than $X_k$ given what Receiver~$k$ observes.

To state the results appropriately, we first consider the rate-equivocation region from~\cite{tandon2011}. The region $\mathcal{R}$ of achievable rate-equivocation $(K+2)$-tuples $(\{R_k\}_{k=1}^K, R_0, \Delta)$ is the union over all auxiliary random variables $W$, arbitrarily correlated with all $X_1$, \ldots, $X_K$, of the tuples satisfying
\begin{linenomath}
\begin{align}
  R_0 &\geq I(X_1, \ldots, X_K;\, W), \label{eq:gw_R0}\\
  R_k &\geq H(X_k \mid W),  \quad k = 1, \ldots, K, \label{eq:gw_Rk}\\
  \Delta &\leq \sum_{k=1}^{K} H(X_1, \ldots, X_K \mid W, X_k),
    \label{eq:gw_Delta}
\end{align}
\end{linenomath}
where the first inequality restricts the feasible $W$'s to those carrying no more mutual information about the sources than a common rate of $R_0$ allows, the second restricts them to those leaving at most $R_k$ bits of residual uncertainty in each source $X_k$, and the third restricts them even further to those achieving a total equivocation, i.e., the aggregate uncertainty that all receivers retain about the sources not intended for them, of at least $\Delta$.

To explore such regions, the authors in \cite{tandon2011} introduce two definitions of common information for $K$ correlated variables, each motivated by a specific engineering goal. The first addresses the following question. If maximal privacy is required, so that no decoder learns more than necessary about the sources not intended for it, what is the largest common rate that the system can afford? Mathematically speaking, this question can be formulated as follows
\begin{Definition}[\cite{tandon2011}] The privacy-optimal common information $C_1$ is the largest common rate $R_0$ compatible with maximal total equivocation,
\begin{linenomath}
\begin{equation}
  C_1 = \max\left\{ R_0 : 
    (\{R_k\}_{k=1}^K, R_0, \Delta_{\max}) 
    \in \mathcal{R} \right\},
\end{equation}
\end{linenomath}
where $\Delta_{\max} = \sum_{k=1}^{K} H(X_1, \ldots, X_K \mid X_k)$ is the 
largest total equivocation achievable in $\mathcal{R}$, i.e., the maximal possible sum of the residual uncertainties each decoder retains about the remaining sources when given only its own.
\end{Definition}

Here, a higher common rate is desirable because the common link is shared by all 
decoders and therefore reduces the total bandwidth needed across the private links. At the same time, a high common rate may force $W$ to carry information about the unintended sources, eroding privacy. $C_1$ is precisely the largest rate at which this trade-off still permits maximal privacy.

The second definition addresses efficiency rather than privacy. Recall that the 
total rate spent delivering source $X_k$ is $R_0 + R_k$, split between the shared 
and private links, and that lossless reconstruction requires this total to be at 
least $H(X_k)$. We say no bandwidth is wasted on $X_k$ when the total meets the 
reconstruction floor $H(X_k)$ exactly, i.e., \linebreak $R_k = H(X_k) - R_0$. Under this 
condition, increasing $R_0$ puts more of each source's description on the shared 
link and less on its private one, but only up to a point. Once $R_0$ exceeds that point, the shared link has to carry information that some receiver does not need. The following definition focuses on the maximal rate $R_0$ at which no bandwidth is wasted on any source.

\begin{Definition}[\cite{tandon2011}] The efficiency-optimal common information $C_2$ is the largest common rate $R_0$ at which every source is described at exactly its own entropy,
\begin{linenomath}
\begin{equation}
  C_2 = \max\left\{ R_0 : 
    (\{H(X_k) - R_0\}_{k=1}^K, R_0) 
    \in \mathcal{R}_{G\text{-}W} \right\},
\end{equation}
\end{linenomath}
where $\mathcal{R}_{G\text{-}W}$ is the Gray--Wyner rate region without the equivocation constraint, consisting of all $(K+1)$-tuples $(\{R_k\}_{k=1}^K, R_0)$ satisfying~\eqref{eq:gw_R0} and~\eqref{eq:gw_Rk} for some auxiliary random variable $W$.
\end{Definition}

Interestingly, the main result of~\cite{tandon2011} is that privacy-optimal common information $C_1$ and efficiency-optimal common information $C_2$ coincide, i.e.,
\begin{linenomath}
\begin{equation}
  C_1 = C_2 = C_{\text{TSP}}(X_1; \ldots; X_K) 
  = \max_{\substack{W:\; W - X_k - \bar{X}_{\setminus k},\\
  k = 1, \ldots, K}} 
  I(X_1, \ldots, X_K;\, W),
  \label{eq:poor_ci}
\end{equation}
\end{linenomath}
where $\bar{X}_{\setminus k}$ represents all sources except $X_k$, and the Markov chain $W - X_k - \bar{X}_{\setminus k}$ requires that $W$ be conditionally independent of the remaining sources given $X_k$ alone. In other words, $W$ can only provide information about the other sources through $X_k$, which is precisely the structural condition that prevents decoder~$k$ from learning anything about them beyond what $X_k$ already reveals. Since this must hold for every $k$ simultaneously, an optimal $W$ can capture only information that is available from every source individually, i.e., information that any single $X_k$ suffices to provide. The equivalence in (\ref{eq:poor_ci}) shows that privacy and efficiency are not in conflict here, as the largest common rate that wastes no bandwidth is exactly the largest common rate compatible with maximal privacy.

This intersective character distinguishes $C_{\text{TSP}}$ from the $K$-variable extension of Wyner's common information $C(X_1; \ldots; X_K)$~\cite{liu2010}, which instead requires the full conditional independence $P_{X_1\cdots X_K|W}(x_1, \ldots, x_K \mid w) = \prod_{k=1}^K P_{X_k|W}(x_k \mid w)$ and takes the form of a minimization. The 
two quantities satisfy~\cite{tandon2011}
\begin{linenomath}
\begin{equation}
  C_{\text{TSP}}(X_1; \ldots; X_K) \leq \min_{i \neq j} I(X_i; X_j)
  \leq \max_{i \neq j} I(X_i; X_j) 
  \leq C(X_1; \ldots; X_K),
\end{equation}
\end{linenomath}
and can differ strictly. Moreover, $C_{\text{TSP}}$ is monotonically non-increasing in $K$, as one would expect from a measure of commonality, whereas $C$ is monotonically non-decreasing in $K$, since requiring full conditional independence across more variables demands a richer description. Therefore, $C_{\text{TSP}}$ can be placed closer to the $K$-variable Gács--Körner extension of Section~\ref{sec:gk} than to Wyner's notion. This resemblance is structural rather than definitional, since $C_{\text{TSP}}$ arises from a 
rate--privacy formulation rather than from deterministic~\mbox{recoverability}.

\subsection{Extension of Wyner's Common Information to Continuous Variables}
\label{sec:continuous_ci}

The definitions developed so far were stated for discrete sources, but the central
quantities are written through mutual information, which can be defined for continuous variables as well. Wyner's common information $C(X;Y)$, its $K$-variable extension $C(X_1;\dots;X_K)$, and the Tandon--Sankar--Poor quantity
$C_{\mathrm{TSP}}(X_1;\dots;X_K)$ therefore carry over to continuous sources unchanged, and the conditional-independence constraints that underlie them remain meaningful there. The extension, however, is not free. For it to be successful, three things must hold, and each can fail for continuous sources in a way it never does for discrete~ones.

The first is that the mutual information $I(X;Y)$ must exist as a finite quantity.
For discrete variables with finite alphabets, $I(X;Y)$ is always finite. For continuous variables, $I(X;Y)$ may be infinite if the joint distribution is concentrated in a way that makes the comparison against independence ill-behaved. In practice, this is a mild condition. Rigorous conditions require that the joint law $P_{XY}$ be absolutely continuous with respect to the product of the marginals $P_X \otimes P_Y$, so that the Radon--Nikodym derivative $\frac{dP_{XY}}{d(P_X \otimes P_Y)}$ exists. The mutual information is then the relative entropy
\begin{linenomath}
\begin{equation}
  I(X;Y) = D\!\left(P_{XY} \,\big\|\, P_X \otimes P_Y\right)
         = \int \log \frac{dP_{XY}}{d(P_X \otimes P_Y)} \, dP_{XY},
\end{equation}
\end{linenomath}
and existence in the useful sense means this integral converges, i.e., the
information density $\log \frac{dP_{XY}}{d(P_X \otimes P_Y)}$ is
$P_{XY}$-integrable. When joint and marginal densities exist, the derivative is
the likelihood ratio $p_{XY}/(p_X p_Y)$ and the requirement reduces to the
convergence of $\int p_{XY} \log \frac{p_{XY}}{p_X p_Y}\,dx\,dy$. Absolute continuity fails when $P_{XY}$ places positive mass on a set of zero $P_X \otimes P_Y$ measure, as when $Y = f(X)$ with both $X$ and $Y$ continuous, in which case the divergence is $+\infty$ and $I(X;Y)$ ceases to be finite.

The second is that $C(X;Y)$ itself must be defined and finite, and this is not
\linebreak guaranteed by the first condition. The common information is an infimum,
\linebreak $C(X;Y) = \inf_{X-W-Y} I(X,Y;W)$, over all latent variables $W$ that render $X$
and $Y$ conditionally independent. A finite $I(X;Y)$ does not guarantee the finiteness of this infimum by itself, as $I(X;Y)$ provides only a lower bound of $C(X;Y)$. For discrete sources with finite alphabets, the choice $W = (X,Y)$ gives the free bound $C(X;Y) \le H(X,Y)$, so finiteness is automatic. For continuous sources that same choice may give infinite cost, and the bound is lost. In practice, finiteness of $C$ must therefore be established from the structure of the two sources rather than assumed.

The third is that $C(X;Y)$ must carry an operational meaning, and this is where the
discrete argument fails most visibly. For discrete sources with finite alphabets, $C(X;Y)$ is the smallest rate on the shared link of the Gray--Wyner network at which both sources are reconstructed losslessly, with the total rate equal to the joint entropy $H(X,Y)$. Such a source has finite entropy and, thus, lossless asymptotic reconstruction is attainable at finite rate and the question is well posed. For a continuous source, the same question is no longer well posed, since describing a real number exactly takes infinitely many bits. Hence, no finite rate, common or private, can reconstruct a continuous source losslessly, and the interpretation that gave $C(X;Y)$ its operational meaning disappears.

The way to restore an operational meaning is to ask a gentler question. Instead of seeking lossless reconstruction, one can investigate whether it is possible to allow each decoder a controlled distortion, i.e., a tolerance on the reconstruction error. For continuous sources this is the natural finite-rate question. This is the path taken by Xu, Liu, and Chen~\cite{xu2011continuous}, who place the Gray--Wyner network under a fidelity criterion. The encoder now observes a block of $n$ i.i.d.\ source pairs $(X^n, Y^n)$, where $X^n = (X_1, \dots, X_n)$ collects $n$ independent draws and $X_i$ stands for the $i$-th symbol (and likewise for $Y$). Let $d_X$ and $d_Y$ be two per-symbol distortion measures. Receiver~1 has to reconstruct $X^n$ as $\hat X^n$ within distortion $\Delta$, and Receiver~2 to reconstruct $Y^n$ as $\hat Y^n$ within $\Delta$, as well, rather than losslessly. A rate triple $(R_0, R_1, R_2)$ is \emph{achievable at distortion} $\Delta$ if, for all sufficiently large blocklengths $n$, there is a code of these rates whose average per-symbol distortions satisfy
\begin{linenomath}
\begin{equation}
  \tfrac{1}{n}\textstyle\sum_{i=1}^n \mathbb{E}\, d_X(X_i, \hat{X}_i) \le \Delta,
  \qquad
  \tfrac{1}{n}\textstyle\sum_{i=1}^n \mathbb{E}\, d_Y(Y_i, \hat{Y}_i) \le \Delta.
\end{equation}
\end{linenomath}
Let $R_{XY}(\Delta, \Delta)$ be the joint rate--distortion function, i.e., the
least total rate $R_0 + R_1 + R_2$ at which both sources can be reconstructed within $\Delta$. {{For} possibly nonsymmetric distortions, the joint rate--distortion function is $R_{XY}(\Delta_1, \Delta_2) = \min_{P_{\hat{X}\hat{Y}|XY}} I(X, Y; \hat{X}, \hat{Y})$, the minimum taken over all reconstructions $(\hat{X}, \hat{Y})$ with $\mathbb{E}\,d_X(X, \hat{X}) \le \Delta_1$ and $\mathbb{E}\,d_Y(Y, \hat{Y}) \le \Delta_2$. That this minimum equals the least total rate defined above is the rate--distortion theorem; see~\cite{cover2006} for the rate--distortion background.} {Returning to the symmetric case, the} object of interest is the smallest common rate attained on the minimum-total-rate plane, 
\vspace{-3pt}
\begin{adjustwidth}{-\extralength}{0cm}
\begin{linenomath}
\begin{equation}
  C_3(\Delta) = \inf\Bigl\{ R_0 :\ (R_0, R_1, R_2)\ \text{achievable at distortion }
  \Delta,\ \ R_0 + R_1 + R_2 = R_{XY}(\Delta, \Delta) \Bigr\}.
  \label{eq:C3_def}
\end{equation}
\end{linenomath}
\end{adjustwidth}
{Strictly, the operating point is reached only asymptotically. The common rate, the total rate, and the two distortions may exceed $R_0$, $R_{XY}(\Delta,\Delta)$, and $\Delta$, respectively, by an arbitrarily small amount for sufficiently large $n$}~(\cite{xu2011continuous}, Definitions~2 and~3).

The simplest version of the result in \cite{xu2011continuous} is the vanishing-distortion limit. Assume that the joint distribution of $(X,Y)$ has full support and that the distortion measure is normal, i.e., $d(x, \hat{x}) = 0$ if and only if $x = \hat{x}$, so that zero distortion is exactly lossless reconstruction. Then
\begin{linenomath}
\begin{equation}
  \lim_{\Delta \to 0} C_3(\Delta) = C(X;Y),
  \label{eq:ci_limit}
\end{equation}
\end{linenomath}
which means that the common information can be recovered as the limit of these lossy problems as the tolerance vanishes. For each $\Delta > 0$, $C_3(\Delta)$ is the common rate of a genuine lossy coding problem. For a discrete source with finite alphabets the value $C(X;Y)$ is already attained at $\Delta = 0$, since that is itself a valid lossless problem. For a continuous source $\Delta = 0$ is unattainable, as exact reconstruction needs infinite rate, so \eqref{eq:ci_limit} is a statement about the vanishing-distortion limit. As we will see next, when $C(X;Y)$ is finite,  $C(X;Y)$ can nonetheless be attained at positive distortion.



This asymptotic statement can be sharpened, and the sharpening was developed from two complementary angles before being unified. Xu, Liu, and Chen first approached $C_3(\Delta)$ from the operational side~\cite{xu2011continuous}, showing that $C_3(\Delta)$ does not merely approach $C(X;Y)$ but can stay exactly equal to $C(X;Y)$ over a whole range of distortions. Such a plateau exists in a neighborhood of the origin whenever the joint distribution has full support and the distortion is normal, and for symmetric distortion its edge is $\gamma = R_{XY}^{-1}(C(X;Y))$. Building on this work, Viswanatha, Akyol, and Rose~\cite{viswanatha2012} approached the same quantity from the single-letter side. The characterization in~\cite{xu2011continuous} is asymptotic, holding in the blocklength limit $n \to \infty$. Instead, Viswanatha, Akyol, and Rose gave an equivalent expression for $C_3$ as an optimization problem over an auxiliary random variable at a single symbol, with no blocklength, and hence computable directly from the joint distribution. They characterize $C_3$ through the optimal rate--distortion reconstructions $(X^{*},Y^{*})$ at distortion $(\Delta_1,\Delta_2)$, proving
\begin{linenomath}
\begin{equation}
\begin{split}
    C_3(\Delta_1,\Delta_2) &= \inf_{U}~ I(X,Y;U) \\
    \text{s.t.}~~~ X^{*}&\perp Y^{*}\mid U , \\
                 (X,Y) &\perp U \mid (X^{*},Y^{*}).
\end{split}
\end{equation}
\end{linenomath}
Here $C_3(\Delta_1,\Delta_2)$ is the general, possibly non-symmetric, version of the common rate, and $C_3(\Delta)=C_3(\Delta,\Delta)$ the symmetric case used above.
The first constraint makes $U$ a common description of the two reconstructions, and
the second confines $U$ to see the sources only through those reconstructions. Here $X^{*}$ and $Y^{*}$ are produced by any channel that achieves the joint rate--distortion function at $(\Delta_1,\Delta_2)$, so they depend on both distortions.
They also noted that convexity and monotonicity of $C_3$ in the distortion are hard to establish in general, which foreshadows the difficulty of pinning down the plateau's extent.

Xu, Liu, and Chen consolidated both viewpoints in a single rigorous analysis in~\cite{xu2016lossy}, one covering discrete and continuous alphabets alike. They also made precise how far the plateau extends for some cases. The edge characterization $\gamma = R_{XY}^{-1}(C(X;Y))$ stated above for the symmetric case holds under a sufficient condition, which they identified as successive refinability~\cite{equitz1991} of the pair $(X,Y)$ and of each source separately. The idea can be understood through the coding scheme that realizes the plateau, which is a layered one. Wyner's variable $W$ carries a common message at rate $C(X;Y)$ and gives each receiver a coarse reconstruction of its source. Then the private messages bring the reconstruction error down to the target distortion. A source is successively refinable when the two-stage scheme reaches the target distortion at the same total rate an optimal one-shot code would spend, i.e., there is no penalty for going in two stages. For such sources, the common rate $C(X;Y)$ and the private refinements together meet the joint budget $R_{XY}(\Delta,\Delta)$ exactly, so the scheme stays on the minimum-total-rate plane with the common rate fixed at $C(X;Y)$. This lasts until the budget itself falls to $C(X;Y)$, so the plateau reaches $\gamma = R_{XY}^{-1}(C(X;Y))$ exactly. When the source is not successively refinable, the refinement carries a penalty, so the common rate $C(X;Y)$ and the private refinements together exceed the joint budget $R_{XY}(\Delta,\Delta)$. The two-stage scheme is then no longer optimal, so the plateau is only guaranteed up to some $\gamma \le R_{XY}^{-1}(C(X;Y))$. The binary and Gaussian examples below are both successively refinable, so their edges follow this formula exactly. Therefore, for continuous sources, $C(X;Y)$ is more than a limit one can only approach. It is the exact rate the two receivers must share for lossy reconstruction, over a whole range of fidelities.

\subsubsection{Locating the Plateau: Binary and Gaussian Sources}

The doubly symmetric binary source (DSBS) makes the plateau concrete in the discrete case~\cite{xu2011continuous}. Its joint distribution is fixed by a binary 
symmetric channel, with $X$ a uniform bit and $Y$ its output under crossover 
probability $a_0$. This channel only sets the correlation between the two sources. 
In the Gray--Wyner network they are treated as correlated sources transmitted to two receivers, one reconstructing $X$ and the other $Y$. In this case, the common information is $C(X;Y) = 1 + h(a_0) - 2h(a_1)$, where $a_1 = \tfrac{1}{2}\big(1 - \sqrt{1-2a_0}\big)$ and $h$ is the binary entropy. Each receiver reconstructs under a Hamming distortion, which for binary sources is the average bit-error probability. Taking the two distortions equal to a common $\Delta$, it can be shown that the smallest common rate stays equal to $C(X;Y)$ for every $\Delta \le a_1$ and drops below $C(X;Y)$ only for $\Delta > a_1$. Because the DSBS is successively refinable, the plateau reaches its edge exactly. To locate the edge, note that at $\Delta = a_1$ the joint rate--distortion function equals the common information,
\begin{equation}
  R_{XY}(a_1, a_1) = 1 + h(a_0) - 2h(a_1) = C(X;Y),
  \label{eq:dsbs_edge}
\end{equation}
so $\gamma = R_{XY}^{-1}(C(X;Y)) = a_1$.

The continuous counterpart is given by the bivariate Gaussian. Let $X$ and $Y$ be 
zero-mean jointly Gaussian with covariance 
$\left[\begin{smallmatrix} 1 & \rho \\ \rho & 1 \end{smallmatrix}\right]$ and 
$\rho > 0$. Here all three conditions hold. The mutual information exists, the 
common information is finite for every $\rho < 1$, and Xu, Liu, and Chen give   $C(X;Y)$\vspace{12pt}
in closed form~\cite{xu2011continuous, xu2016lossy},
\begin{equation}
  C(X;Y) = \frac{1}{2}\log\frac{1+\rho}{1-\rho}.
  \label{eq:gauss_scalar}
\end{equation}
Under symmetric squared-error distortion, it can be shown that the smallest common rate $C_3(\Delta)$ stays equal to $C(X;Y)$ for every $\Delta \le 1-\rho$ and drops below $C(X;Y)$ only beyond. As a result, the plateau has edge $\gamma = 1-\rho$, the continuous analogue of $a_1$.

The Gaussian source also shows why lossless reconstruction had to 
be abandoned. Its symmetric joint rate--distortion function has three 
branches~\cite{berger1971},
\begin{equation}
  R_{XY}(\Delta, \Delta) =
  \begin{cases}
    \dfrac{1}{2}\log\dfrac{1-\rho^2}{\Delta^2}, 
      & 0 < \Delta \le 1-\rho, \\[2mm]
    \dfrac{1}{2}\log\dfrac{1+\rho}{2\Delta - (1-\rho)},
      & 1-\rho \le \Delta \le 1, \\[2mm]
    0, & \Delta \ge 1.
  \end{cases}
  \label{eq:gauss_rdf}
\end{equation}
On the low-distortion branch the rate is $\tfrac{1}{2}\log\frac{1-\rho^2}{\Delta^2}$, which diverges as $\Delta \to 0$. Zero distortion would require infinite rate, so $\Delta=0$ is not an attainable operating point, exactly the obstruction that condition three anticipated. The binary case differs, since the joint rate--distortion function of the DSBS stays finite at $\Delta = 0$, where $R_{XY}(0,0) = H(X,Y) = 1 + h(a_0)$. A discrete source with a finite alphabet can be reconstructed exactly at finite rate, but a continuous one cannot. This is why the plateau matters here. Its operating points are all at positive distortion, and across those up to $1-\rho$ the common rate equals $C(X;Y)$.

\subsubsection{Going Beyond Wyner's Common Information}

Thus far, we have treated $X$ and $Y$ as given and asked how efficiently they can be
compressed. The reverse question is also important. How would one engineer two
variables to exhibit a desired correlation, and how much shared randomness does that
take? This is the problem of distributed channel synthesis, analyzed by Cuff~\cite{cuff2013}, where Wyner's common information $C(X;Y)$ is the rate the two sides must communicate to build the correlation without sharing any randomness.
A different direction varies how faithfully the synthesized distribution must match
the target $P_{XY}$. Measuring that mismatch by a R\'enyi divergence of order $\alpha$, which for $\alpha > 1$ is stronger than the Kullback--Leibler one, gives the R\'enyi common information of Yu and Tan~\cite{yu2018,yu2018corrections}. As $\alpha$ grows, this quantity interpolates between Wyner's common information, which is recovered for every $\alpha \le 1$, and the exact common information of Kumar, Li, and El Gamal~\cite{kumar2014}, reached at $\alpha=\infty$~\cite{yu2020}. Finally, when the common information of two variables is infinite, a real-valued vector $\mathbf{w}$ is needed to carry the shared randomness. Li, Hanna, Diggavi, and Fragouli~\cite{li2025} consider the common information dimension, $\dim(\mathbf{w})$, and compute it in closed form for jointly Gaussian vectors. These directions lie beyond the scope of this tutorial, and we only point to them for interested readers. 

\subsubsection{Extension to Many Continuous Random Variables}

The passage from discrete to continuous carries over to more than two sources. The
$K$-variable Wyner common information $C(X_1; \ldots; X_K)$ is defined as in
Section~\ref{sec:wyner_K} under the requirement that $W$ renders all $K$
sources conditionally independent. Again, the main quantity is the mutual information, and the same three conditions of Section~\ref{sec:continuous_ci} must hold for $C(X_1; \ldots; X_K)$ to carry over to the continuous setting. Xu, Liu, and Chen developed the lossy interpretation for two sources and noted that the results extend to $K$ sources in a straightforward manner~\cite{xu2016lossy, xu2013generalizations}. According to the authors, both the plateau and its edge carry over, with the common rate holding at $C(X_1; \ldots; X_K)$ over a range of distortions. 

Furthermore, as in the two-view case, $C(X_1; \ldots; X_K)$ is finite for zero-mean unit-variance jointly Gaussian sources with common pairwise correlation $\rho \in (0,1)$. For any $K\geq 2$, these sources admit the decomposition
\begin{equation}
  X_k = \sqrt{\rho}\, W + \sqrt{1-\rho}\, N_k,
  \qquad k = 1, \ldots, K,
  \label{eq:gauss_decomp_K}
\end{equation}
with $W, N_1, \ldots, N_K$ independent standard Gaussians, and their common
information evaluates to~(\cite{xu2016lossy}, Theorem~8)
\begin{equation}
  C(X_1; \ldots; X_K)
  = \frac{1}{2}\log\left(1 + \frac{K\rho}{1-\rho}\right).
  \label{eq:gauss_K}
\end{equation}
For jointly Gaussian sources the intersective Gács--Körner common information
vanishes, while the Wyner quantity stays finite and informative. This asymmetry
is one reason why Wyner's common information has attracted more attention for
continuous variables.

A separate comment concerns the algebraic tools of Section~\ref{sec:wyner_K}. The factorization-based analysis applies directly to discrete sources, where the joint PMF is a finite array. In that setting, the conditional-independence constraint becomes a nonnegative matrix or tensor factorization, and identifiability follows from conditions on the nonnegative rank and the conditional PMFs, such as Kruskal's. This machinery is not tied to the discrete case. A smooth density with compact support can be approximated by a truncated multivariate expansion, whose coefficients form a finite tensor. Amiridi et~al.~\cite{amiridi2022} use a Fourier expansion and decompose this characteristic tensor to learn an approximation with a naive Bayes structure. A latent discrete variable $W$ renders the observed variables conditionally independent, so the
density factors as
\begin{equation}
  f_{X_1 \cdots X_K}(x_1, \ldots, x_K)
  = \sum_{w=1}^{|\mathcal{W}|} P_W(w)
    \prod_{k=1}^{K} f_{X_k \mid W}(x_k \mid w).
  \label{eq:naive_bayes}
\end{equation}

However, two issues stand in the way. First, not every density with a finite characteristic tensor is a finite sum of tensor products of univariate densities with the same harmonics, as we discuss below. Hence, the naive Bayes form is a modeling restriction rather than a universal representation. Dropping the requirement that the factors be densities removes this restriction for densities with a finite characteristic tensor, since every finite tensor admits a CP decomposition of high enough rank. Second, the decomposition of Amiridi et~al.~\cite{amiridi2022} does not constrain the univariate factors to be densities, so the learned model need not be a valid density.


Karakasis and Sidiropoulos~\cite{karakasis2024fejerriesz} instead enforce validity by construction. Their univariate spectral factorization based on the Fejér--Riesz theorem models each conditional $f_{X_k \mid W}$ as a bona fide PDF. They also consider a softmax-based representation for $P_W$. Hence, their implementation of \eqref{eq:naive_bayes} corresponds to a parametrization of a valid density. The price is a constrained decomposition. Due to their spectral factorization, their model considers only structured factors for the univariate densities. This raises the question of whether every finite characteristic tensor admits a CP decomposition with structured factor matrices at some finite rank. The model nevertheless approximates any density with compact support in the $\ell_1$ sense at high enough rank and number of Fourier coefficients. Identifiability carries over here as well. The Kruskal-type conditions of Section~\ref{sec:wyner_K} can still be applied to these structured factor matrices in order to infer the uniqueness of the learned representation. When these conditions hold, the learned $P_W$ and conditional densities form the unique naive Bayes representation of the learned density with at most $|\mathcal{W}|$ components, up to relabeling of $W$.

One distinction is worth stressing. These methods fit the density and carry no mutual-information objective. The latent $W$ they recover renders the views conditionally independent under the fitted density, so $W$ is feasible for Wyner's constraint on that density, and on the true one only when the fit is exact. However, the learned $W$ is not necessarily the mutual-information minimizer that defines Wyner's common information. Moreover, such a $W$ lies in a restricted part of that feasible set, since the recovered $W$ is discrete, whereas Wyner's optimal $W$ can be continuous, as we show in the jointly Gaussian case. The connection is therefore structural. Both rest on the same conditional-independence factorization, and density estimation supplies a common latent from data rather than Wyner's optimal~one.

\subsubsection{Applications}

Common information arises in several problems of information theory. The Gray--Wyner network~\cite{grawyner1974} is the original setting, where Wyner's common information is the smallest shared rate at minimum total rate for lossless reconstruction~\cite{wyner1975} and, under suitable conditions, for lossy reconstruction~\cite{xu2016lossy}. Viswanatha, Akyol, and Rose used this network to model a database that stores two correlated sources and serves requests for each source separately~\cite{viswanatha2011}. Storing the sources jointly minimizes the storage rate but makes retrieval costly, whereas storing them separately does the opposite. They showed that the optimal trade-off between storage and retrieval connects the operating points of Wyner and Gács--Körner, which correspond to these two extremes. Sula and Gastpar~\cite{sula2022} solved two related problems explicitly for Gaussian sources. They showed that Gaussian codebooks are optimal for the Gray--Wyner network with a sum-rate constraint on the private links, which gives the rate region in closed form. Moreover, they considered Wyner's common information, together with a relaxation that only bounds the conditional mutual information, for scalar and vector sources.

Another field where common information appears is information-theoretic security. In secure source coding, a transmitter describes a source to a legitimate receiver and uses a secret key shared with that receiver to limit what an eavesdropper can infer. Satpathy and Cuff studied secure source coding with Gaussian sources~\cite{satpathy2015}. Characterizing the trade-off between the communication rate and the key rate requires finding an optimal auxiliary variable, similar to the one in Wyner's common information. Their analysis also gives Wyner's common information of jointly Gaussian vectors in closed form, generalizing the scalar result of Xu, Liu, and Chen~\cite{xu2011continuous}. On the Gács--Körner side, the common information is the rate of shared randomness that two parties can extract from correlated observations without communicating. Allowing public communication leads to the common-randomness and secret-key capacities studied in~\cite{ahlswede1993, ahlswede1998, maurer1993, csiszar2000}. 

Common information also appears in caching, where part of the data is delivered to a user before the user's request is known. Wang, Lim, and Gastpar considered a user who requests one of two files with equal probability. Once the cache is large enough, the optimal cache content is a description that renders the two files conditionally independent, and the smallest cache rate at which this is possible is Wyner's common information~\cite{wang2016caching}. When the files need only be reconstructed within a given distortion, lossy versions of the Gács--Körner and Wyner common information also play an important role in caching~\cite{timo2018caching}.

{\color{black}In machine learning, common information has also been learned directly from data. Kleinman et~al.~\cite{kleinman2023} define a notion of common information that can be estimated from samples and recovers the Gács--Körner common information as a special case. They use this notion to train a variational autoencoder that separates the information shared by two views from the information unique to each. On the Wyner side, Ryu et~al.~\cite{ryu2022wyner} learn a compact common representation of two modalities for joint and conditional generation. Their motivation is channel synthesis, where Wyner's common information is the smallest rate needed to build the correlation, as discussed above. Huang and El Gamal~\cite{huang2024wyner} develop efficient solvers for two relaxations of Wyner's problem and apply them to multimodal clustering. Related information-theoretic objectives have been used to improve adversarial robustness~\cite{yu2023adversarial} and to measure the dependence among several variables without estimating their distributions~\cite{yu2021measuring}.}

\subsection{How Many Kinds of Common Information Are There After All?}

Thus far, we have covered several distinct quantities under the same name, and it is
worth stepping back to see how they relate. In particular, this section introduced the Gács--Körner common information $K_{GK}$, Wyner's common information $C$, its
$K$-variable extension, the common entropy $G_{CE}$, and the Tandon--Sankar--Poor quantity $C_{\text{TSP}}$. Maximal correlation entered as the test for whether $K_{GK}$ is positive. A few others were mentioned only in passing, such as the exact common information and the Rényi common information. This raises a natural question: \emph{How many kinds of common information are there, and what separates them?}

Despite the variety, they fall into two families that pull in opposite directions. The first family is intersective. Quantities in this family ask what the sources hold in common, in the sense that each source can produce the shared part on its own. Gács--Körner is the prototype here. The Gács--Körner common information seeks the richest variable that every source determines, so the quantity maximizes an entropy. The Tandon--Sankar--Poor quantity is its stochastic, privacy-driven cousin, since $C_{\text{TSP}}$ also captures only information that is available from every source individually. Maximal correlation belongs here as well, as the criterion for whether any such shared variable exists at all. The Gács--Körner and Tandon--Sankar--Poor quantities do not increase as more sources are added, because agreement across a larger set can only get harder to reach.

The second family is generative. Quantities in this family ask what must be shared to explain or reproduce the dependence, rather than what each source can recover alone. Wyner's common information is the prototype. Wyner's common information seeks the most parsimonious variable that explains the dependence by rendering the sources conditionally independent, so the quantity minimizes a rate. The $K$-variable extension $C(X_1;\ldots;X_K)$, the common entropy $G_{CE}$, and the exact common information are of the same kind. The extension $C(X_1;\ldots;X_K)$ does not decrease as more sources are added, because a single variable must now account for potentially more dependence. The two prototypes therefore move in opposite directions with the number of sources. The intersective quantities never increase, and the $K$-variable Wyner quantity never decreases.

The two families coincide only in special cases. For two discrete variables with finite alphabets, the prototypes satisfy $K_{GK}(X;Y) \le I(X;Y) \le C(X;Y)$, and both inequalities hold with equality exactly when the joint PMF, after a relabeling of the symbols, is block-diagonal with product blocks. In general, the intersective and generative values differ, and jointly Gaussian variables make the gap extreme. For a correlation $0 < |\rho| < 1$, the Gács--Körner common information vanishes, while Wyner's common information is finite and positive.

Partial information decomposition (PID)~\cite{williams2010} offers a third and complementary lens. PID splits the information that two sources $X$ and $Y$ jointly carry about a third variable $T$ into four nonnegative parts,
\begin{linenomath}
\begin{equation}
  I(X, Y; T) = R(X, Y; T) + U_X(X; T) + U_Y(Y; T) + S(X, Y; T),
  \label{eq:pid}
\end{equation}
\end{linenomath}
where $R$ is the redundant information that each source carries about $T$, $U_X$ and $U_Y$ are the unique information that only one source carries, and $S$ is the synergistic information that only the two together reveal. The redundant part $R$ belongs with the intersective family, since the information about $T$ that $R$ measures is available from each source alone, in the spirit of Gács--Körner. The synergistic part $S$ is information that the two sources hold only jointly, since $S$ appears only when the two sources are combined. Synergy has no counterpart in either family of common information. If $X$ and $Y$ are independent uniform bits and $T = X \oplus Y$, then $K_{GK}(X;Y) = C(X;Y) = 0$, yet the full bit of $I(X,Y;T)$ is synergistic, since neither source alone carries any information about $T$. PID is therefore complementary to our story rather than a refinement of the story. The main difference is that PID rests on a target variable $T$, whereas the common-information notions we discussed describe the dependence between the sources themselves, aiming to discover or engineer shared structure and then exploit that structure.

The synergistic term $S(X, Y; T)$ in~\eqref{eq:pid} names an older idea, namely that information can appear only in combination. McGill~\cite{mcgill1954} studied three variables and defined their interaction as the change in the mutual information between two of them when the third is held fixed. Hu~\cite{hu1962} later extended mutual information to any number of variables through an inclusion--exclusion formula on entropies. For two variables, mutual information is $I(X_1; X_2) = H(X_1) + H(X_2) - H(X_1, X_2)$, and for three variables the formula of Hu reads
\begin{linenomath}
\begin{equation}
\begin{split}
  I(X_1; X_2; X_3) &= H(X_1) + H(X_2) + H(X_3) + H(X_1, X_2, X_3) \\
  &\quad - H(X_1, X_2) - H(X_1, X_3) - H(X_2, X_3),
\end{split}
\end{equation}
\end{linenomath}
which equals $I(X_1; X_2) - I(X_1; X_2 \mid X_3)$. The interaction of McGill is the same quantity with the opposite sign. Unlike mutual information, $I(X_1; X_2; X_3)$ can be negative, as both authors observed, and a negative value is the footprint of synergy. In the XOR example above, $I(X;Y) = 0$ while $I(X;Y \mid T) = 1$ bit, so $I(X;Y;T) = -1$ bit. More generally, the decomposition~\eqref{eq:pid} gives $I(X;Y;T) = R(X,Y;T) - S(X,Y;T)$, so the three-variable quantity turns negative exactly when synergy outweighs redundancy. The common-information notions, by contrast, keep only the shared part and stay nonnegative. This line of work is a rich topic in its own right, but lies outside the scope of this tutorial.

\subsection{Summary}

This section developed the information-theoretic foundation for the rest of the
tutorial. Throughout, the objects of study were jointly distributed random variables, and the question was how to quantify what they share. Before moving on, it is worth highlighting the main ideas.

\begin{itemize}
\item \textbf{Two families of definitions.} The formalizations of common information
split into two groups that answer different questions. The \emph{intersective} family asks what the sources hold in common, in the sense that each source can produce the shared part on its own. Its prototype is the Gács--Körner quantity $K_{GK}(X;Y)$ of~\eqref{eq:gk}, the maximum entropy of a variable $W$ that is simultaneously a deterministic function of $X$ and of $Y$. The Tandon--Sankar--Poor quantity $C_{\mathrm{TSP}}$ of~\eqref{eq:poor_ci} is its stochastic counterpart, arising from a rate--privacy formulation rather than from deterministic recoverability. The \emph{generative} family asks what must be supplied to account for the dependence. Its prototype is Wyner's $C(X;Y)$ of~\eqref{eq:wyner}, the minimum of $I(X,Y;W)$ over all $W$ rendering $X$ and $Y$ conditionally independent. The $K$-variable extension $C(X_1;\dots;X_K)$ and the exact common information belong to the same family. The common entropy $G_{CE}$, which minimizes $H(W)$ instead of $I(X,Y;W)$ under the same constraint, belongs to this family as well.

\item \textbf{The two families behave oppositely to the number of sources.} The intersective quantities $K_{GK}$ and $C_{\mathrm{TSP}}$ are non-increasing in the number of sources, since agreement across a larger collection is a stronger requirement. The $K$-variable Wyner quantity is non-decreasing, since a single variable must account for more dependence as sources are added. For two discrete variables with finite alphabets, the prototypes are ordered as $K_{GK}(X;Y) \leq I(X;Y) \leq C(X;Y)$, with equality throughout exactly when the joint PMF, after a relabeling of the symbols, is block-diagonal with product blocks. Jointly Gaussian variables separate them sharply, since for $0 < |\rho| < 1$ the Gács--Körner common information vanishes while $C(X;Y)$ remains finite and positive.

\item \textbf{Positivity of $K_{GK}$ is decided by maximal correlation.} Witsenhausen's criterion states that $K_{GK}(X;Y) > 0$ if and only if the Hirschfeld--Gebelein--Rényi maximal correlation satisfies $S(X,Y) = 1$. The criterion certifies that a nontrivial common variable exists but does not identify an optimal one, since the maximizing functions yield in general only a coarsening of the optimal common variable.

\item \textbf{Common information admits an operational characterization.} In the
Gray--Wyner network, $C(X;Y)$ is the minimum rate on the link shared by the two
receivers, subject to lossless reconstruction at total rate $H(X,Y)$. For continuous
sources lossless reconstruction is unattainable at a finite rate, and the statement is
recovered in lossy form. Under full support and a normal distortion measure, the minimum common rate $C_3(\Delta)$ equals $C(X;Y)$ over a range of distortions, whose edge is $\gamma = R_{XY}^{-1}(C(X;Y))$ for successively refinable sources. For $K$ sources, $C_{\mathrm{TSP}}$ simultaneously characterizes the largest common rate compatible with maximal privacy and the largest common rate compatible with describing each source at exactly its own entropy.

\item \textbf{Conditional independence carries an algebraic form.} Writing Wyner's constraint $X \perp Y \mid W$ on a discrete joint PMF yields a nonnegative matrix factorization with $|\mathcal{W}|$ components, and the $K$-variable form of the constraint yields a nonnegative tensor factorization. For three or more variables, identifiability of the tensor factorization follows from uniqueness conditions, of which Kruskal's is the best known, whereas the two-variable matrix factorization is not unique in general.

\item \textbf{Partial information decomposition offers a complementary account.} PID splits $I(X,Y;T)$ into redundant, unique, and synergistic parts relative to a target $T$. The redundant part aligns with the intersective family, while the synergistic part lies outside every notion discussed here, since synergy is held only jointly and appears only when both sources are combined.
\end{itemize}

The definitions of this section specify what is to be recovered, namely a latent variable capturing what the sources hold in common, characterized either by deterministic recoverability from each source or by conditional independence of the sources given the latent variable. The definitions do not specify how to recover the latent variable from data, and they assume the joint distribution to be known. Both restrictions are lifted in what follows, where, following the multiview literature, we call the sources \emph{views}. Section~\ref{sec:cca} replaces scalar random variables with random vectors and the known distribution with a finite sample, and asks which linear combinations within each view are maximally correlated. The formal link between the two developments is the Hirschfeld--Gebelein--Rényi maximal correlation introduced above, since canonical correlation analysis solves the same correlation-maximization problem with the admissible functions restricted to be linear. This restriction is what makes the problem tractable and well posed from samples, and Sections~\ref{sec:cca}{--}
\ref{sec:nonlinear} are concerned with what can be recovered under the linearity restriction, and with what must be assumed to relax that restriction.

\section{Canonical Correlation Analysis}
\label{sec:cca}

\subsection{Three Problems to Set the Scene}
Before the mathematics, we start this section by introducing three problems from different fields, each asking the same underlying question.

For the first problem, consider a wireless sensor network (WSN). Each node carries several sensors and produces a bundle of recordings. Bandwidth is scarce, so we would like to find pairs of nodes whose recordings are strongly correlated and encode such pairs jointly rather than separately. The question here is which pairs of nodes genuinely share structure, and how we can quantify their correlation in order to sort all pairs and pick the pairs we can afford to encode jointly.

For the second problem, we switch to functional magnetic resonance imaging (fMRI). There, the brain volume is imaged as a grid of small volume elements, or voxels. Each voxel carries a time series of its blood-oxygen-level-dependent (BOLD) signal, a proxy for local neural activity. Neighboring voxels tend to fluctuate together, partly through shared blood-supply factors and partly because the single-voxel signal is spatially smoothed, to suppress noise, which correlates adjacent voxels. Hence, it is natural to treat a neighborhood of voxels as a single unit rather than a lone voxel. We then ask whether two different nonoverlapping neighborhoods exhibit correlated activity, as a step toward mapping how functional brain networks are organized.

For the third problem, consider an urban cellular network, where each multi-antenna base station forms a cell and serves the users within the cell. A user near the boundary of two adjacent cells is typically received weakly by both stations. Each station alone may be incapable of decoding that user, but since both stations receive a weak copy of the same transmission, the stations can collaborate to detect such users and decode their messages. The questions that arise are how two base stations can discover that a faint common signal is present in both received signals, and how the stations can cooperate to decode a message that neither could recover alone.

The three problems differ in every surface detail, yet are instances of the same
abstract form. In each of them, a view is described not by a single scalar random variable, as before, but by a vector of channels or features collecting a \emph{set} of related signals, gathered at a multi-sensor node, over a voxel neighborhood, or at a base station. The question is whether two such multivariate views share hidden common structure, and how strongly. The concepts of common information developed in Section~\ref{sec:common_info} can still be applied here in many different ways. For example, one could measure the common information between scalar-valued signals from different views and then report an average. However, two nodes may share a weak source that no single pair of their sensors reveals, because noise and interference bury the common signal in each single-channel measurement.

Most of the time, the natural grouping of the channels that make up a view carries a mathematical structure that can be exploited. When the same source signals appear mixed throughout the different channels of a view, processing all channels together can approximately extract a wide range of these signals, from stronger ones to weaker ones. Processing every channel on its own does not provide these degrees of freedom. As a result, in detecting multiview correlations, jointly processing all channels within a view can be significantly more effective, as we will see.

Mathematically speaking, the joint processing of the channels of each view will be
performed here by taking linear combinations of the channel signals and looking for
the combinations whose outputs are most correlated. The question of this section
then becomes whether some combination of the first set agrees with some combination
of the second, which treats the two sets as wholes rather than in parts. An ideal answer has to provide two things. The first is a number reporting how strongly the two multivariate views are related, which allows the pairs of sensor nodes above to be sorted and the strongly related ones encoded jointly. The second is the
resulting combinations themselves, since they carry the shared structure. This is
exactly the question canonical correlation analysis was built to answer, as we see
next.

\subsection{Hotelling's Original Formulation}

Canonical correlation analysis (CCA) was introduced by Harold Hotelling
in 1936~\cite{hotelling1936}, nearly four decades before the
common-information notions of Gács--Körner~\cite{gacs1973} and
Wyner~\cite{wyner1975}. Hotelling was a statistician and economic
theorist whose motivations were explicitly empirical. In a study by Truman L. Kelley, the measurements on each schoolchild fell into two groups. The reading group held reading speed and reading power, and the arithmetic group held arithmetic speed and arithmetic power. Hotelling was interested in finding what reading and arithmetic abilities have in common. In a second, economic setting, the variables were the prices and the quantities of a group of related commodities such as wheat, potatoes, rye, and barley, each commodity contributing its price to one set and its quantity to the other. There, the question was how prices and quantities move together across the whole group. In each case, the variables split into two sets, and the interest lies not in either set alone but in what the two sets share. In the scholastic example, Hotelling concluded that reading and arithmetic involve one common mental factor, and in the economic one, the prices and quantities of related commodities must be studied as a group rather than one commodity at a time. He argued that the classical one-variable-at-a-time approach could not reveal such shared structure, and his idea was to reduce each set to a single linear combination, one from each set, chosen so that the two are as correlated with each other as possible. That maximal correlation measures how strongly the two sets share their common factor.

Despite its power, CCA largely remained confined to statistics, with a few exceptions. As late as 2001, Borga~\cite{borga2001} noted that although canonical correlation was by then standard in economics, medical studies, meteorology, and
even the classification of malt whisky, it was ``surprisingly unknown in
the fields of learning and signal processing.'' That is no longer an accurate description of either field. In machine learning, CCA entered through its kernel and deep formulations, and deep CCA together with its autoencoding and identifiable variants was published at the main machine learning conferences~\cite{andrew2013, wang2015, benton2017, lyu2022}. The agreement objectives of current self-supervised learning rest on the same criterion~\cite{zbontar2021barlow, bardes2022vicreg}. In signal processing, CCA underlies the geometry of linear estimation, filtering, and Gaussian channel capacity~\cite{scharf2000}, blind detection of unknown users in cellular uplinks~\cite{ibrahim2020, ibrahim2024mmse}, underlay spectrum sharing~\cite{ibrahim2022underlay, razavi2024fm}, signal alignment in IoT networks~\cite{peppas2023harnessing, peppas2025signal}, scalable formulations for many views~\cite{fu2017, sorensen2021}, and multisubject fMRI analysis~\cite{karakasis2020multi, karakasis2021multi, karakasis2022}. Tutorials aimed at these audiences followed in learning, computing, and neuroscience venues rather than statistics ones~\cite{hardoon2004, uurtio2017, wang2018}. Next, we proceed with the rigorous mathematical formulation of CCA.

\subsubsection{Problem Formulation and Lagrangian Derivation}
\label{sec:cca_formulation}

Let $\bx_1 \in \R^{M_1}$ and $\bx_2 \in \R^{M_2}$ be two zero-mean
random vectors, called \emph{views}. {{Note} that we can always subtract the mean, in practice its sample estimate, to meet this assumption.} Following Hotelling's original setup~\cite{hotelling1936}, CCA seeks a pair of directions $\bq_1 \in \R^{M_1}$ and $\bq_2 \in \R^{M_2}$ that maximize the Pearson correlation between the projections $z_1 = \bq_1^\top\bx_1$ and $z_2 = \bq_2^\top\bx_2$. This objective reads
\begin{linenomath}
\begin{equation}
  \max_{\bq_1,\bq_2}
  \;\rho(z_1,z_2)
  =
  \frac{\Cov(\bq_1^\top\bx_1,\,\bq_2^\top\bx_2)}
  {\sqrt{\Var[\bq_1^\top\bx_1]\cdot\Var[\bq_2^\top\bx_2]}},
  \label{eq:cca_corr}
\end{equation}
\end{linenomath}
where an optimal pair $(\bq_1^*, \bq_2^*)$ is called a pair of canonical
directions, and the projections the pair produces, $z_1^* := (\bq_1^*)^\top\bx_1$ and $z_2^* := (\bq_2^*)^\top\bx_2$, are called canonical variates. The denominator in (\ref{eq:cca_corr}) normalizes the covariance by the standard deviations of the two projections, so that the final ratio measures correlation rather than raw covariance. Its purpose is to make the objective invariant to the scaling of $\bq_1$ and $\bq_2$, since multiplying either direction by a nonzero constant leaves $\rho$ unchanged, in absolute value. As we can see, only the directions $\bq_1$ and $\bq_2$ matter, not their norms, and this scale invariance is what allows us to fix the two variances to one in the equivalent constrained form below.

Letting $\bSigma_{kk} = \Cov(\bx_k)$ and
$\bSigma_{12} = \Cov(\bx_1,\bx_2) = \bSigma_{21}^\top$, the problem
in~\eqref{eq:cca_corr} is equivalent~to
\begin{linenomath}
\begin{equation}
  \max_{\bq_1,\bq_2}
  \;\bq_1^\top\bSigma_{12}\bq_2
  \quad\text{s.t.}\quad
  \bq_1^\top\bSigma_{11}\bq_1 = 1,\quad
  \bq_2^\top\bSigma_{22}\bq_2 = 1,
  \label{eq:cca_opt}
\end{equation}
\end{linenomath}
where the optimal value, called the (first) canonical correlation, is $\rho(z_1^*,z_2^*) = (\bq_1^*)^\top\bSigma_{12}\bq_2^*$. The equivalence holds because the constraints fix the denominator
of~\eqref{eq:cca_corr} to one, and, by the scale invariance noted above,
imposing them loses no generality, leaving the optimal directions
unchanged up to scaling. What we have really done by passing to~\eqref{eq:cca_opt} is trade the ratio in~\eqref{eq:cca_corr} for a simpler bilinear objective. The price we pay is the introduction of two constraints, which at first glance
complicates things.

\paragraph{{Lagrangian} 
 Derivation}

The standard approach for optimizing a function under equality constraints is the use of Lagrange multipliers. Hotelling derived the solution using them, and his argument is worth following. The Lagrangian of problem~\eqref{eq:cca_opt} is given by
\begin{linenomath}
\begin{equation}
  \mathcal{L}(\bq_1,\bq_2,\lambda_1,\lambda_2)
  = \bq_1^\top\bSigma_{12}\bq_2
  - \frac{\lambda_1}{2}(\bq_1^\top\bSigma_{11}\bq_1 - 1)
  - \frac{\lambda_2}{2}(\bq_2^\top\bSigma_{22}\bq_2 - 1).
  \label{eq:lagrangian}
\end{equation}
\end{linenomath}
Setting the partial derivatives to zero gives
\begin{linenomath}
\begin{align}
  \bSigma_{12}\bq_2 &= \lambda_1\bSigma_{11}\bq_1,
  \label{eq:kkt1}\\
  \bSigma_{21}\bq_1 &= \lambda_2\bSigma_{22}\bq_2.
  \label{eq:kkt2}
\end{align}
\end{linenomath}
To identify the multipliers, multiply~\eqref{eq:kkt1} and \eqref{eq:kkt2} on the left by $\bq_1^\top$ and $\bq_2^\top$, respectively. After using the constraints
$\bq_1^\top\bSigma_{11}\bq_1 = \bq_2^\top\bSigma_{22}\bq_2 = 1$, we get
\begin{linenomath}
\begin{equation}
\begin{split}
  \lambda_1 &= \bq_1^\top\bSigma_{12}\bq_2, \\
  \lambda_2 &= \bq_2^\top\bSigma_{21}\bq_1.
\end{split}
\end{equation}
\end{linenomath}
A careful reader will notice that $\bq_2^\top\bSigma_{21}\bq_1$ is a
scalar, which is equal to its own transpose, $\bq_2^\top\bSigma_{21}\bq_1 = (\bq_2^\top\bSigma_{21}\bq_1)^\top = \bq_1^\top\bSigma_{21}^\top\bq_2 = \bq_1^\top\bSigma_{12}\bq_2$, where the last step uses $\bSigma_{21}^\top = \bSigma_{12}$. Hence $\lambda_1 = \lambda_2$, and under the unit-variance constraints this common value is exactly the
correlation being maximized,
\begin{linenomath}
\begin{equation}
  \lambda_1 = \lambda_2
  = \bq_1^\top\bSigma_{12}\bq_2 = \rho(z_1,z_2) =: \rho .
\end{equation}
\end{linenomath}

We now turn our attention to eliminating $\bq_2$. Assuming $\bSigma_{22}$ is invertible and $\rho\neq 0$, solving
~\eqref{eq:kkt2} for $\bq_2$ and using $\lambda_2 = \rho$ gives
\begin{linenomath}
\begin{equation}
  \bq_2 = \frac{1}{\rho}\,\bSigma_{22}^{-1}\bSigma_{21}\bq_1 .
\end{equation}
\end{linenomath}
Substituting into~\eqref{eq:kkt1} and using $\lambda_1 = \rho$ gives,
after multiplying by $\rho$, the \emph{generalized eigenvalue problem}
\begin{equation}
  \bSigma_{12}\bSigma_{22}^{-1}\bSigma_{21}\bq_1
  = \rho^2\,\bSigma_{11}\bq_1 .
\end{equation}
Assuming $\bSigma_{11}$ is invertible, this becomes the standard
eigenvalue problem
\begin{equation}
  \bSigma_{11}^{-1}\bSigma_{12}\bSigma_{22}^{-1}\bSigma_{21}\bq_1
  = \rho^2\,\bq_1.
  \label{eq:gen_eig}
\end{equation}
Therefore, an optimal canonical direction $\bq_1$ is an eigenvector of $\bSigma_{11}^{-1}\bSigma_{12}\bSigma_{22}^{-1}\bSigma_{21}$ associated with its largest eigenvalue, and the canonical correlation is the positive square root of that eigenvalue. Then, the matching canonical direction $\bq_2$ follows from~\eqref{eq:kkt2}.

Compared to Hotelling's derivations, we reached~\eqref{eq:gen_eig} by a shortcut, eliminating $\bq_2$ for pedagogical convenience. The resulting eigenvalue problem is nonetheless equivalent to Hotelling's original symmetric formulation, differing only by that elimination, so the two share the same canonical correlations and canonical variates. Hence, we are still able to invoke Hotelling's results directly. Building on the positive definiteness of the covariance forms, he proved that the roots are real and that the canonical correlations, being correlation coefficients, cannot exceed unity~\cite{hotelling1936}. Equivalently, the eigenvalues $\rho^2$ of~\eqref{eq:gen_eig} are real and lie in $[0,1]$. Their
nonnegative square roots
$1 \ge \sigma_1 \geq \sigma_2 \geq \cdots \geq \sigma_s \ge 0$, where
$s = \min(M_1,M_2)$, are the \emph{canonical correlations}, and the
corresponding solutions $(\bq_1^{(i)}, \bq_2^{(i)})$ form the $i$-th pair
of \emph{canonical directions}. The canonical variates
$z_k^{(i)} = (\bq_k^{(i)})^\top\bx_k$ obey a further property that
Hotelling also established. Canonical variates associated with different
canonical correlations are mutually uncorrelated, both within each view
and across the two. Therefore, each canonical variate $z_k^{(i)}$ is correlated only with the corresponding variate of the other view, with correlation $\sigma_i$, and is uncorrelated with every other canonical variate, across and within views. When canonical correlations repeat, the variates can always be chosen so that this holds.

We set out to find a single most-correlated pair, yet the eigenvalue
problem~\eqref{eq:gen_eig} handed us a whole family. The matrix of \eqref{eq:gen_eig} has at most $s = \min(M_1, M_2)$ nonzero eigenvalues, so at most $s$ eigenvectors carry correlation. The leading one, associated with the largest
eigenvalue $\sigma_1^2$, is the direction $\bq_1^{(1)}$ we were after,
and together with the matching direction $\bq_2^{(1)}$, $\bq_1^{(1)}$ forms the first canonical pair. The remaining eigenvectors are not spurious. Each subsequent pair $(\bq_1^{(i)}, \bq_2^{(i)})$ solves the same correlation-maximization
problem, restricted to directions whose canonical variates are
uncorrelated with all earlier ones. The eigendecomposition therefore
delivers a nested sequence of canonical pairs, ordered by how much
correlation they capture, rather than a single pair.

Recovering the first $D \le s$ pairs yields a $D$-dimensional common
representation of the two views. To find these $D$ pairs simultaneously,
we seek matrices $\bQ_k \in \R^{M_k \times D}$ whose columns are the
first $D$ canonical directions. In sample form, with $N$ joint observations of $\bx_1$ and $\bx_2$ stacked as rows in $\bX_1 \in \R^{N \times M_1}$ and $\bX_2 \in \R^{N \times M_2}$, the multi-pair problem reads
\begin{linenomath}
\begin{equation}
  \min_{\bQ_1,\bQ_2}
  \|\bX_1\bQ_1 - \bX_2\bQ_2\|_F^2
  \quad\text{s.t.}\quad
  \bQ_k^\top\bX_k^\top\bX_k\bQ_k = \bI_D,\; k=1,2.
  \label{eq:cca_multi}
\end{equation}
\end{linenomath}
Given an optimal pair $\bQ_1^*$ and $\bQ_2^*$, the resulting canonical representations of $\bX_1$ and $\bX_2$ are given by $\mathbf{Z}_1^* = \bX_1\bQ_1^*$ and $\mathbf{Z}_2^* = \bX_2\bQ_2^*$. Notice that optimal solutions are not unique, since after multiplying both $\bQ_1^*$ and $\bQ_2^*$ on the right by the same orthogonal matrix $\mathbf{O} \in \R^{D \times D}$, the resulting matrices $\bQ_1^*\mathbf{O}$ and $\bQ_2^*\mathbf{O}$ are also optimal. This has an intuitive explanation. The columns of $\mathbf{Z}_1^*$ and $\mathbf{Z}_2^*$ are orthonormal bases of two subspaces of $\R^N$, and the objective in~\eqref{eq:cca_multi} depends only on how well the two bases align column by column. Rotating both bases by the same $\mathbf{O}$ keeps each basis orthonormal and preserves this alignment, so the objective is unchanged, whereas rotating the two bases differently would in general break the alignment.

These canonical variates are the goal of the construction. However, as explained above, the solution of~\eqref{eq:cca_multi} is not unique. The objective depends only on the sum of the canonical correlations $\sigma_1, \ldots, \sigma_D$, not on representations whose columns are ordered from the most to the least correlated directions. By design, $\mathbf{Z}_1^*$ and $\mathbf{Z}_2^*$ are expressed in the same $D$-dimensional space, a common coordinate system in which the two views are aligned, but not in a way that reflects how strongly each coordinate is correlated across views. Ordered canonical representations follow from the singular value decomposition ${\mathbf{Z}_1^*}^\top\mathbf{Z}_2^* = \mathbf{U}^*\boldsymbol{\Sigma}^*{\mathbf{V}^*}^\top$, namely $\bar{\mathbf{Z}}_1^* = \mathbf{Z}_1^*\mathbf{U}^*$ and $\bar{\mathbf{Z}}_2^* = \mathbf{Z}_2^*\mathbf{V}^*$. After this transform, both $\bar{\mathbf{Z}}_1^*$ and $\bar{\mathbf{Z}}_2^*$ remain orthonormal bases, and ${\bar{\mathbf{Z}}_1^*}{}^\top\bar{\mathbf{Z}}_2^* = \boldsymbol{\Sigma}^*$ is diagonal, with \linebreak $\mathrm{trace}\big({\bar{\mathbf{Z}}_1^*}{}^\top\bar{\mathbf{Z}}_2^*\big) = \mathrm{trace}\big({\mathbf{Z}_1^*}^\top\mathbf{Z}_2^*\big) = \sum_{d=1}^D \sigma_d$, where the diagonal elements $\sigma_1 \ge \cdots \ge \sigma_D$ of $\boldsymbol{\Sigma}^*$ are the canonical correlations. Each column pair of $\bar{\mathbf{Z}}_1^*$ and $\bar{\mathbf{Z}}_2^*$ is then one shared coordinate, and the larger $\sigma_i$ is, the more closely the $i$-th coordinates of the two views agree. As we will see later, these aligned representations carry valuable information that can be exploited in a range of downstream tasks, and in specific applications they acquire natural~\mbox{interpretations}.

\subsection{Closed-Form Solution via SVD and the Geometry of Principal Angles}
The eigenvalue problem~\eqref{eq:gen_eig} already solves CCA, but the
matrix involved entangles the cross-covariance $\bSigma_{12}$ with the
inverse within-view covariances $\bSigma_{11}^{-1}$ and
$\bSigma_{22}^{-1}$ in an unintuitive way. A change of coordinates called
\emph{whitening} removes that entanglement and recasts the problem as the
singular value decomposition (SVD) of a simpler matrix. As we show here, whitening also opens a geometric picture, in which the canonical correlations are the cosines of the \textit{principal angles} between the subspaces spanned by the canonical variates $\bar{\mathbf{Z}}_1^*$ and $\bar{\mathbf{Z}}_2^*$~\cite{stewart1990, scharf2000}.

To whiten a zero-mean random vector $\bx$ with covariance $\bSigma$ is to
apply $\bSigma^{-1/2}$, producing $\tilde\bx = \bSigma^{-1/2}\bx$, whose
covariance is the identity matrix. The result has unit variance in every
direction and no correlation between its entries. The name comes from
white noise, whose covariance is likewise the identity matrix. Whitening each
view separately,
\begin{linenomath}
\begin{equation}
  \tilde\bx_1 = \bSigma_{11}^{-1/2}\bx_1, \qquad
  \tilde\bx_2 = \bSigma_{22}^{-1/2}\bx_2,
\end{equation}
\end{linenomath}
standardizes each view, stripping away its internal correlations and scale, leaving only how the two views relate to each other. Why this helps is immediate. In terms of the whitened directions $\tilde\bq_k := \bSigma_{kk}^{1/2}\bq_k$, the projections are unchanged, $\bq_k^\top\bx_k = \tilde\bq_k^\top\tilde\bx_k$, and the unit-variance constraints of~\eqref{eq:cca_opt} reduce to ordinary unit-norm constraints $\|\tilde\bq_k\| = 1$. {\color{black}Writing the whitened cross-covariance matrix, known as the \emph{coherence matrix}~\cite{scharf2000}, as}
\begin{linenomath}
\begin{equation}
  \bM := \Cov(\tilde\bx_1, \tilde\bx_2)
  = \bSigma_{11}^{-1/2}\bSigma_{12}\bSigma_{22}^{-1/2},
  \label{eq:whitened_cross}
\end{equation}
\end{linenomath}
we see that~\eqref{eq:gen_eig} collapses to $\bM\bM^\top\tilde\bq_1 = \rho^2\tilde\bq_1$. CCA has thus reduced to finding the unit directions along which the two whitened views covary most. 

{{For two scalar stationary processes, the analogue of the coherence matrix is the magnitude-squared coherence, the squared magnitude of the correlation coefficient of their Fourier components at each frequency. Carter et~al.}~\cite{carter1973} {studied the estimation of the magnitude-squared coherence. Like CCA, the magnitude-squared coherence yields a sequence of correlations, but one per frequency from a single pair of scalar processes, rather than from a reduction of two vector views to their canonical variates.}}

{The unit directions along which the whitened views covary most} are exactly what the SVD of $\bM$ returns, 
\begin{linenomath}
\begin{equation}
  \bM 
  = \bU\,\mathrm{diag}(\sigma_1,\ldots,\sigma_s)\,\bV^\top,
  \label{eq:cca_svd}
\end{equation}
\end{linenomath}
where $\bU \in \R^{M_1 \times s}$ and $\bV \in \R^{M_2 \times s}$ have orthonormal columns. The canonical correlations $\sigma_1 \geq \cdots \geq \sigma_s \geq 0$ are the singular values, and the singular vectors $\bu_i$ and $\bv_i$ are the canonical directions in whitened coordinates. Mapping back to the original data by undoing the
whitening gives $\bq_1^{(i)} = \bSigma_{11}^{-1/2}\bu_i$ and
$\bq_2^{(i)} = \bSigma_{22}^{-1/2}\bv_i$. This SVD route is preferred in practice, as it returns the $\sigma_i$ directly rather than as the squared eigenvalues $\rho^2$ of~\eqref{eq:gen_eig}, while the first $D$ columns of $\bU$ and $\bV$ can be used to directly compute $\bar{\mathbf{Z}}_1^*$ and $\bar{\mathbf{Z}}_2^*$.

A geometric reading of the canonical correlations, as cosines of the
angles between two subspaces, is already present in Hotelling's
paper~\cite{hotelling1936}. Representing $N$ observations of a feature as a
point in $N$-dimensional space, he noted that the columns of $\bX_1$ and $\bX_2$ 
span two flat subspaces meeting at the origin, and showed that the
canonical correlations $\sigma_1, \ldots, \sigma_s$ are exactly the
cosines of the angles between them, the largest being the cosine of the
minimum angle. These angles are known today as the \emph{principal
angles} between the~subspaces.

To see this in the multi-pair formulation, note that the constraints
in~\eqref{eq:cca_multi} require the columns of $\bX_k\bQ_k$ to form an
orthonormal basis of a $D$-dimensional subspace of $\col(\bX_k)$.
Expanding the objective gives
\begin{linenomath}
\begin{equation}
  \|\bX_1\bQ_1 - \bX_2\bQ_2\|_F^2
  = 2D - 2\,\Tr\!\left(\bQ_1^\top\bX_1^\top\bX_2\bQ_2\right),
\end{equation}
\end{linenomath}
so minimizing the distance maximizes the alignment
$\Tr(\bQ_1^\top\bX_1^\top\bX_2\bQ_2)$. At the optimum, this alignment equals $\sum_{i=1}^D \cos\theta_i$, where $\theta_1 \le \cdots \le \theta_D$ are the $D$ smallest principal angles between $\col(\bX_1)$ and $\col(\bX_2)$. As a result,
\begin{linenomath}
\begin{equation}
  \min_{\bQ_1,\bQ_2}\|\bX_1\bQ_1 - \bX_2\bQ_2\|_F^2
  = 2D - 2\sum_{i=1}^D \cos\theta_i,
  \qquad \text{where }~~\sigma_i = \cos\theta_i .
  \label{eq:principal_angles}
\end{equation}
\end{linenomath}
Minimizing the distance is therefore the same as finding the pair of
$D$-dimensional subspaces, one within $\col(\bX_1)$ and one within
$\col(\bX_2)$, that are most closely aligned. Figure~\ref{fig:principal_angles} shows the case $D = 1$, where a single angle is involved. In this subsection we showed how the SVD is used when the focus is computing the canonical correlations. Later we will see that the SVD can also be used in a different way, one that places more emphasis on recovering the common representation shared by the views. 

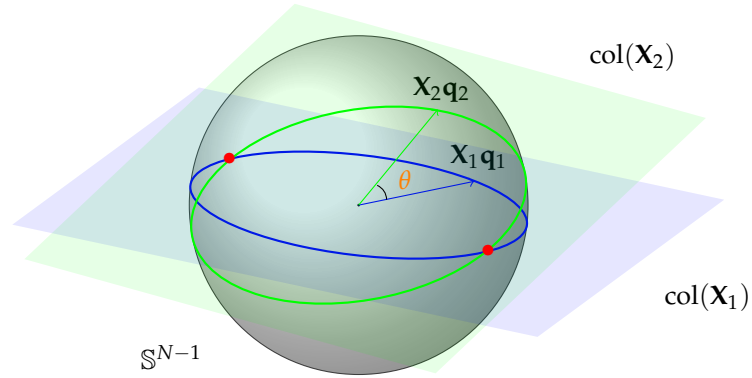
\begin{figure}[H]
\centering
\tdplotsetmaincoords{65}{110}
\pgfmathsetmacro{\rvec}{1.6}
\begin{tikzpicture}[scale=1.4,tdplot_main_coords]
\shadedraw[tdplot_screen_coords,ball color=white,opacity=0.7] (0,0) circle (\rvec);
\coordinate (O) at (0,0,0);
\node at (0,0,0)[circle,fill,inner sep=0.2pt]{};
\tdplotsetrotatedcoords{-20}{-10}{-70}
\draw[thick,color=blue,tdplot_rotated_coords] (\rvec,0,0) arc (0:360:\rvec);
\draw[->,color=blue,tdplot_rotated_coords] (0,0,0) -- (-0.65,-1.45,0);
\node[] at (-1.4,0.7,0) {$\bX_1\bq_1$};
\begin{scope}[tdplot_rotated_coords]
 \begin{scope}[canvas is xy plane at z=0]
  \fill[blue,fill opacity=0.1] (-2.8,-2.5) rectangle (2.5,2.8);
  \pgflowlevelsynccm
 \end{scope}
\end{scope}
\tdplotsetrotatedcoords{30}{-60}{120}
\draw[thick,color=green,tdplot_rotated_coords] (\rvec,0,0) arc (180:-180:-\rvec);
\draw[->,color=green,tdplot_rotated_coords] (0,0,0) -- (-0.2,-1.6,0);
\node[] at (-2.7,-0.15,0) {$\bX_2\bq_2$};
\begin{scope}[tdplot_rotated_coords]
 \begin{scope}[canvas is xy plane at z=0]
  \fill[green,fill opacity=0.1] (-2.5,-2.5) rectangle (2.3,2);
  \pgflowlevelsynccm
 \end{scope}
\end{scope}
\node[] at (0.8,3.8) {$\col\left(\bX_1\right)$};
\node[] at (-4,1.3) {$\col\left(\bX_2\right)$};
\node[] at (0.55,1.5) {$\color{red}\bullet$};
\node[] at (-0.55,-1.5) {$\color{red}\bullet$};
\node[] at (3.8,-0.5) {$\mathbb{S}^{N-1}$};
\node[] at (-0.68,0.22) {${\color{orange}\theta}$};
\draw[-,tdplot_rotated_coords] (0.2,-0.18,0) to[bend right] (0.75,0.55,0.8);
\end{tikzpicture}
\caption{{The} 
 single-component case, drawn for $\bX_1, \bX_2 \in \R^{3 \times 2}$. {Blue refers to view 1 and green to view 2, for the planes, the circles, and the vectors.}
The constraints confine each canonical variate to the unit sphere of $\R^{N}$, $\mathbb{S}^{N-1}$, and to the column space of its own view. Hence, the two feasible sets are the
circles where the sphere intersects those subspaces, and $\theta$ is the angle
between a feasible pair. Minimizing $\|\bX_1\bq_1 - \bX_2\bq_2\|_2^2$ closes
that angle as far as the subspaces allow. Here, two planes in $\R^3$ share a
line, so the circles intersect at the two antipodal points marked in red, each corresponding to an optimum, where $\theta = 0$ and $\sigma_1 = 1$. {For visualization purposes, the figure ignores centering. With $N=3$, centered columns lie in the two-dimensional subspace orthogonal to $\mathbf{1}$, so two centered data matrices of rank two would share the same column space.}
\label{fig:principal_angles}}
\end{figure}

This geometric view makes it clear why CCA recovers common structure. The canonical variates pick out the parts of each view's column space that point most strongly toward the other view. These aligned variates span estimates of a \emph{common subspace}. In particular, when the two column spaces intersect, the canonical variates with $\sigma_i = 1$ span exactly $\col(\bX_1)\cap\col(\bX_2)$, so in that case CCA performs subspace intersection. Beyond that intersection there is nothing the two views hold exactly in common, and the later canonical pairs settle at positive angles, with $\sigma_i$ below one. Figure~\ref{fig:intersection} shows both cases in one picture. This geometry is also the bridge promised in Section~\ref{sec:common_info}.

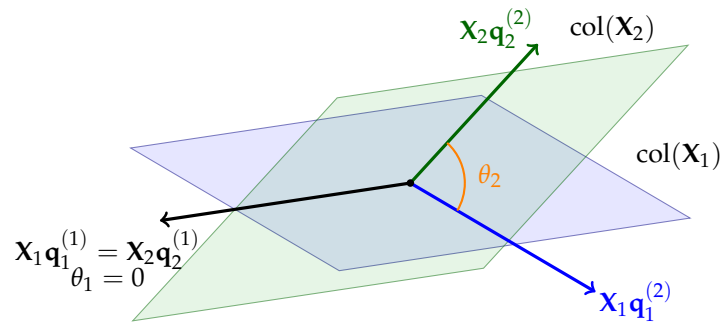
\begin{figure}[H]
\centering
\begin{tikzpicture}[scale=0.80]
\fill[blue,opacity=0.10]
  (4.624,-0.579) -- (1.176,1.449) -- (-4.624,0.579) -- (-1.176,-1.449) -- cycle;
\draw[blue!50!black,thin,opacity=0.5]
  (4.624,-0.579) -- (1.176,1.449) -- (-4.624,0.579) -- (-1.176,-1.449) -- cycle;
\fill[green!60!black,opacity=0.10]
  (4.590,2.278) -- (1.210,-1.408) -- (-4.590,-2.278) -- (-1.210,1.408) -- cycle;
\draw[green!40!black,thin,opacity=0.5]
  (4.590,2.278) -- (1.210,-1.408) -- (-4.590,-2.278) -- (-1.210,1.408) -- cycle;
\draw[very thick,->] (0,0) -- (-4.140,-0.621);
\draw[very thick,->,green!40!black] (0,0) -- (2.108,2.299);
\draw[very thick,->,blue] (0,0) -- (3.060,-1.800);
\draw[orange,thick] (0.90,0) arc (0:47.5:0.90);
\draw[orange,thick] (0.90,0) arc (0:-30.5:0.90);
\node[orange] at (1.32,0.12) {$\theta_2$};
\node at (3.35,2.55) {$\col\left(\bX_2\right)$};
\node at (4.45,0.45) {$\col\left(\bX_1\right)$};
\node[align=center] at (-5.00,-1.28)
  {$\bX_1\bq_1^{(1)} = \bX_2\bq_2^{(1)}$\\[-2pt] $\theta_1 = 0$};
\node[green!40!black] at (1.43,2.56) {$\bX_2\bq_2^{(2)}$};
\node[blue] at (3.74,-1.95) {$\bX_1\bq_1^{(2)}$};
\node at (0,0)[circle,fill,inner sep=1pt]{};
\end{tikzpicture}
\caption{Both canonical pairs at once, drawn for
$\bX_1, \bX_2 \in \R^{3 \times 2}$. Along the line the two subspaces share, the
canonical variates coincide, so $\theta_1 = 0$ and $\sigma_1 = 1$, and this
direction spans $\col(\bX_1) \cap \col(\bX_2)$. The second pair cannot be
brought into agreement, since each variate is confined to its own subspace and
must be uncorrelated with the first, so it settles at the angle $\theta_2$ with
$\sigma_2 = \cos\theta_2 < 1$. What CCA reports as common is the part of the
spectrum at $\sigma_i = 1$. {As in Figure}~\ref{fig:principal_angles}{, the figure ignores centering.} 
\label{fig:intersection}}
\end{figure}

\begin{Connection}[{\color{black}One-sided implication}: Perfect correlation and Gács--Körner]
\label{con:witsenhausen}
A canonical correlation of one is a linear certificate of Gács--Körner common information for vector-valued views. If a population canonical correlation equals one, then $\mathbf{a}^\top\bx_1 = \mathbf{b}^\top\bx_2$ almost surely for some nonzero $\mathbf{a}$ and $\mathbf{b}$. Then, the two views have deterministic mappings agreeing on a nontrivial common random variable and the maximal correlation $S(\bx_1,\bx_2)$ of Section~\ref{sec:gk} equals one. Hence, the Gács--Körner common information is positive and common information exists. Witsenhausen~\cite{witsenhausen1975} showed that a nontrivial common random variable exists exactly when two subspaces of $L^2$ share a line. The first subspace contains the zero-mean functions of $\bx_1$, and the second contains the zero-mean functions of $\bx_2$. CCA restricts its search to linear functions, so the two subspaces of $L^2$ shrink from infinite-dimensional to the finite-dimensional spans of the centered coordinates of $\bx_1$ and of $\bx_2$. In the finite-sample case, these spans become the column spaces of the centered data matrices $\mathbf{X}_1$ and $\mathbf{X}_2$. A sample canonical correlation then equals one exactly when $\col(\bX_1)\cap\col(\bX_2)\neq\{\mathbf{0}\}$. CCA therefore searches for a Gács--Körner common random variable among linear functions only. The converse of the certificate fails, as the views may share a nonlinear common random variable while every canonical correlation stays below one.
\end{Connection}

Drawing quantitative connections with Gács--Körner and Wyner's common information requires a model for the joint distribution of the views, under which these information quantities can be computed. As we discuss next, the jointly Gaussian model is the natural choice for CCA.

\subsubsection*{Invariance of Canonical Correlations}

One last thing that we would like to note is that the canonical correlations are invariant under independent invertible linear transformations of the two views. Replacing $\bx_1$ by $\bA_1\bx_1$ and $\bx_2$ by $\bA_2\bx_2$ with $\bA_1, \bA_2$ invertible leaves $\sigma_1, \ldots, \sigma_s$ unchanged. This is because the {\color{black}coherence (whitened cross-covariance) matrix} $\bM$ becomes $\mathbf{O}_1\bM\mathbf{O}_2^\top$ for some orthogonal matrices $\mathbf{O}_1$ and $\mathbf{O}_2$, and multiplication by orthogonal matrices leaves singular values unchanged. The canonical correlations therefore capture something intrinsic to the relationship between the two views, independent of the basis and scaling chosen within each~view.

\subsection{Connection to Mutual Information and Wyner's Common Information}
\label{sec:borga}

{\color{black}In signal processing, Scharf and Mullis~\cite{scharf2000} expressed the mutual information between jointly Gaussian random vectors in terms of canonical correlations. More specifically, they showed that, when a message and its measurement are jointly Gaussian, the information rate and the concentration ellipse of the filtering errors are functions of the canonical correlations. On this basis, they argued that the canonical variates of the message and the measurement form the natural coordinate system for both second-order estimation and Gaussian communication.} Borga~\cite{borga2001} later presented the same link in his tutorial on CCA. For jointly Gaussian random vectors $\bx_1$ and $\bx_2$, the mutual information is determined entirely by the canonical correlations,
\begin{linenomath}
\begin{equation}
  I(\bx_1;\bx_2) = -\frac{1}{2}\sum_{i=1}^s \log(1-\sigma_i^2).
  \label{eq:mi_cca}
\end{equation}
\end{linenomath}
Borga noted that information is additive over statistically independent variables and that the canonical variates are uncorrelated. For jointly Gaussian views, uncorrelated variates are independent, so the mutual information decomposes as a sum over the canonical variate pairs, each contributing $-\tfrac{1}{2}\log(1-\sigma_i^2)$.

The canonical correlations can also determine Wyner's common information of
\linebreak Section~\ref{sec:common_info}. For jointly Gaussian views, we get~\cite{satpathy2015}
\begin{linenomath}
\begin{equation}
  C(\bx_1;\bx_2) = \frac{1}{2}\sum_{i=1}^s \log\frac{1+\sigma_i}{1-\sigma_i},
  \label{eq:wyner_gaussian}
\end{equation}
\end{linenomath}
which reduces, for a scalar pair with correlation $\rho$, to the value
$\tfrac{1}{2}\log\frac{1+\rho}{1-\rho}$ found in Section~\ref{sec:common_info}.
Comparing~\eqref{eq:wyner_gaussian} with~\eqref{eq:mi_cca} term by term gives $C(\bx_1;\bx_2) \ge I(\bx_1;\bx_2)$, as we explained in the previous section.

\begin{Connection}[{\color{black}Special-case consequence}: Canonical correlations determine both information quantities]
\label{con:gaussian}
For jointly Gaussian views the mutual information is~\eqref{eq:mi_cca} and Wyner's
common information is~\eqref{eq:wyner_gaussian}. Both are functions of the canonical
correlations alone. Hence CCA computes both in the Gaussian case.
\end{Connection}

{\color{black}Huang, Wornell, and Zheng~\cite{huang2018gaussian} reach the coherence matrix $\mathbf{M}$ of \eqref{eq:whitened_cross} from a learning perspective. They consider two jointly Gaussian views $\mathbf{x}_1$ and $\mathbf{x}_2$, for which the canonical correlations are small, and ask which low-dimensional features of one view should be extracted before the inference task is known. The task is to estimate an unknown attribute $\mathbf{u}_1$ of $\mathbf{x}_1$ from $\mathbf{x}_2$, or an unknown attribute $\mathbf{u}_2$ of $\mathbf{x}_2$ from $\mathbf{x}_1$, under the Markov chain $\mathbf{u}_1 \leftrightarrow \mathbf{x}_1 \leftrightarrow \mathbf{x}_2 \leftrightarrow \mathbf{u}_2$. In several cases, they show the optimal $k$ features are the first $k$ canonical variates of the observed view, obtained from the SVD of $\mathbf{M}$. For the same weakly correlated views, they also show that Wyner's common information equals the sum of the canonical correlations up to higher-order terms, and that an optimal latent variable of Wyner's problem is built from the canonical directions scaled by the square roots of the canonical correlations. Then, the sum of the canonical variates of the two views is a sufficient statistic for that latent variable.}

Sula and Gastpar~\cite{sula2021} also consider the feature-extraction problem, by revisiting Wyner's formulation for random vectors. Wyner's constraint requires the latent vector $\mathbf{w}$ to explain the entire dependence between the views, and for jointly Gaussian views this forces the optimal $\mathbf{w}$ to carry every canonical component with a nonzero canonical correlation. Sula and Gastpar therefore relax the constraint to $I(\mathbf{x}_1;\mathbf{x}_2 \mid \mathbf{w}) \le \gamma$, which lets $\mathbf{w}$ explain only the strongest part of the dependence and leave at most $\gamma$ bits unexplained. A larger $\gamma$ thus yields a coarser and lower-dimensional $\mathbf{w}$. Their method, called common information components analysis, computes the optimal $\mathbf{w}$ at level $\gamma$ and projects $\mathbf{w}$ back onto each view to obtain the extracted features. For Gaussian statistics, the method reduces exactly to CCA, and $\gamma$ sets how many canonical components are extracted.

So far, CCA has given us the canonical correlations and the aligned representations $\mathbf{Z}_1$ and $\mathbf{Z}_2$. It is worth pausing to ask a more basic, almost philosophical question. Why should two views be correlated at all, and what ties them together? As we see next, Bach and Jordan~\cite{bach2005} proposed a rigorous answer. Two views can be correlated because they share something, a common latent factor that drives both. Each observation is then a noisy, linearly transformed glimpse of the same hidden cause, and the correlation we measure is the trace this factor leaves on the two views at once. Once the factor is known, nothing further links the views, so they become conditionally independent given the factor, which is exactly Wyner's feasibility condition. This is a generative stance rather than a descriptive one. Instead of asking how aligned the two views are, a generative model posits the mechanism that makes them aligned, and turns CCA into the estimation of that~mechanism.

\subsection{Probabilistic {Interpretation of CCA}} 

\label{sec:cca_prob}

Bach and Jordan~\cite{bach2005} turned this intuition into a precise model. They posit a single $D$-dimensional latent factor that generates both views and prove that the maximum-likelihood estimates of the model parameters are expressed through the first $D$ canonical directions and canonical correlations. 
\begin{Definition}[Probabilistic CCA model~\cite{bach2005}]
Let $\bz \in \R^D$, with $D \le \min(M_1,M_2)$, be a latent random vector, and
let the two views be generated as
\begin{linenomath}
\begin{equation}
  \bz \sim \mathcal{N}(\mathbf{0},\bI_D), \qquad
  \bx_k \mid \bz \sim \mathcal{N}(\bW_k\bz,\bPsi_k), \quad k = 1,2,
  \label{eq:prob_cca}
\end{equation}
\end{linenomath}
where $\bW_k \in \R^{M_k \times D}$ are loading matrices and
$\bPsi_k \succ 0$ are noise covariances. The joint distribution
factorizes as
\begin{linenomath}
\begin{equation}
  p(\bx_1,\bx_2,\bz) = p(\bz)\,p(\bx_1\mid\bz)\,p(\bx_2\mid\bz),
  \label{eq:cond_ind_prob_cca}
\end{equation}
\end{linenomath}
so the two views are \emph{conditionally independent given} $\bz$.
\end{Definition}

Under the model above, the latent $\bz$ is the common factor, the loadings $\bW_k$ carry it into each view, and the noise covariances $\bPsi_k$ absorb the private variation that $\bz$ does not explain. Bach and Jordan showed that maximum-likelihood estimation of this model matches CCA. The maximum-likelihood estimates of the loading matrices are $\widehat\bW_k = \bSigma_{kk}\bQ_k\bM_k$, where the columns of $\bQ_k$ are the first $D$ canonical directions and $\bM_1,\bM_2 \in \R^{D\times D}$ are any matrices with $\bM_1\bM_2^\top = \bP$ and spectral norms smaller than one. Here, $\bP$ is the diagonal matrix of the first $D$ canonical correlations. The estimates of the residual covariances are $\widehat\bPsi_k = \bSigma_{kk} - \widehat\bW_k\widehat\bW_k^\top$. The posterior means are $\E[\bz \mid \bx_k] = \bM_k^\top\bQ_k^\top\bx_k$. Hence, the shared latent is estimated, up to an invertible linear map, by the canonical variates $\bQ_k^\top\bx_k$. Since $\bM_1$ and $\bM_2$ are not unique, the model fixes a representation of $\bz$ only up to this family of solutions. Among these solutions, the one that maximizes the conditional entropy of the views given $\bz$ takes $\bM_1 = \bM_2 = \bP^{1/2}\bR$ for a rotation matrix $\bR$.

This probabilistic reading also links CCA back to the common-information ideas of Section~\ref{sec:common_info}. The conditional independence in~\eqref{eq:cond_ind_prob_cca} is Wyner's feasibility condition, as $\bz$ renders $\bx_1 \perp \bx_2 \mid \bz$, and the model decomposes each view into a shared part $\bW_k\bz$ and a private part with covariance $\bPsi_k$. The three viewpoints we have developed now become facets of one object. Hotelling's canonical correlations measure the correlation strength of the shared structure, the mutual information formula~\eqref{eq:mi_cca} of Scharf and Mullis and of Borga measures its information content, and the Bach--Jordan latent $\bz$ captures its generative source. Yet the tie between the Bach--Jordan latent $\bz$ and Wyner's common information is sharper than feasibility alone. Every member of the family $\bM_1\bM_2^\top = \bP$ satisfies Wyner's conditional independence constraint, while Wyner's objective asks for the feasible $\bz$ that minimizes $I(\bx_1,\bx_2;\bz)$. Satpathy and Cuff~(\cite{satpathy2015}, Section~IV-B) show that Wyner's common information of jointly Gaussian views~\eqref{eq:wyner_gaussian} is attained by the latent $\mathbf{v}$ in the model
\begin{linenomath}
\begin{equation}
  \bQ_1^\top\bx_1 = \bP^{1/2}\mathbf{v} + (\bI-\bP)^{1/2}\mathbf{n}_1, \qquad
  \bQ_2^\top\bx_2 = \bP^{1/2}\mathbf{v} + (\bI-\bP)^{1/2}\mathbf{n}_2,
  \label{eq:satpathy_cuff}
\end{equation}
\end{linenomath}
where $\mathbf{v}$, $\mathbf{n}_1$ and $\mathbf{n}_2$ are independent $\mathcal{N}(\mathbf{0},\bI)$ vectors. Each pair of canonical variates thus shares $v_i$ with weight $\sqrt{\sigma_i}$ and carries independent noise of variance $1-\sigma_i$. This construction is a Bach--Jordan model with $\bM_1 = \bM_2 = \bP^{1/2}$, where $\bP$ holds the true canonical correlations. The optimal latent is therefore the population counterpart of the maximum-entropy solution singled out above.

\begin{Connection}[{\color{black}Special-case equivalence}: The balanced Bach--Jordan latent attains Wyner's common information]
\label{con:bachjordan}
In the Bach--Jordan model, the true latent $\bz$ is always feasible for Wyner's problem, because the model makes the views conditionally independent given $\bz$. However, for jointly Gaussian views, as in the case of the Bach--Jordan model, the Wyner-optimal latent of Satpathy and Cuff corresponds to the balanced Bach--Jordan model with $\bM_1=\bM_2=\bP^{1/2}$ in~\eqref{eq:satpathy_cuff}, built from the true canonical directions and correlations. The balanced maximum-likelihood solution of Bach and Jordan results in the same form, with sample-based estimates in place of the true ones, and is therefore the plug-in estimate of this optimal model. CCA thus provides every parameter of the optimal model, and the canonical correlations alone give its value $C(\bx_1;\bx_2)$ in~\eqref{eq:wyner_gaussian}.
\end{Connection}

Bach and Jordan's results show what CCA estimates once their probabilistic model is assumed to hold. Earlier we saw a second, algebraic reading. Given the sample matrices $\bX_1$ and $\bX_2$, CCA tries to match their column spaces through the principal angles and returns the two aligned representations $\bar\bZ_1^*$ and $\bar\bZ_2^*$. When the two column spaces share a nontrivial subspace, the components with unit canonical correlation coincide, and $\bar\bZ_1^*$ and $\bar\bZ_2^*$ agree on a basis for the shared subspace. As a result, CCA can provide an estimate of the intersection $\col(\bX_1)\cap\col(\bX_2)$. However, it is not clear when the subspace that CCA recovers is a \emph{meaningful} one. This is the \emph{identifiability} question. Suppose the data genuinely consist of a shared factor and private, view-specific ones, so that the common structure lives in a low-dimensional subspace. Under what conditions does CCA recover a unique nontrivial common subspace? For the two-view linear case, Ibrahim and Sidiropoulos~\cite{ibrahim2019} gave an explicit answer.

\subsection{Identifiability: When Does CCA Recover the Common Subspace?}
\label{sec:cca_id}

Consider a generative model in which for each view, the sample matrix is a mixture of a shared component and a private one,
\begin{linenomath}
\begin{equation}
  \bX_k = \bG\bS_k^\top + \bC_k\bH_k^\top + \bE_k, \quad k=1,2,
  \label{eq:gen_two_view}
\end{equation}
\end{linenomath}
where $\bG \in \R^{N \times D}$ spans the \emph{common subspace} shared by both views, $\bC_k \in \R^{N \times L_k}$ spans the \emph{private subspace} of view $k$, the matrices $\bS_k$ and $\bH_k$ hold the common and private mixing coefficients, and $\bE_k$ is noise. As we will see later, depending on the application,
$\col(\bG)$ may correspond to the subspace of cell-edge signals, or the task-related brain activation shared across brain areas or even subjects. Stacking the two blocks, we obtain $\bX_k = [\bG,\ \bC_k]\,[\bS_k,\ \bH_k]^\top + \bE_k$. Based on the algebraic generative model in (\ref{eq:gen_two_view}), when does CCA recover the common subspace $\col(\bG)$?
\begin{Theorem}[\cite{ibrahim2019}]
\label{thm:cca_id}
  Suppose $\bX_1$ and $\bX_2$ are generated by~\eqref{eq:gen_two_view}, and that $[\bG,\bC_1,\bC_2]$ and $[\bS_k, \bH_k]$, for $k = 1, 2$, are
  full column rank. Then, in the noiseless case, the CCA solution of~\eqref{eq:cca_multi} satisfies
  $\col(\bX_1\bQ_1^*) = \col(\bX_2\bQ_2^*) = \col(\bG)$.
\end{Theorem}
The full column rank requirement on $[\bG,\bC_1,\bC_2]$ is just the statement that the three subspaces are linearly independent. It is worth stressing that this asks for less than orthogonality. The subspaces are free to lean toward one another at oblique angles, and the only thing ruled out is that any one of them shares a direction with the span of the other two. The cleanest way to see what this buys us is a uniqueness statement. Any vector in their combined span splits into a common part in $\col(\bG)$, a view-1 private part in $\col(\bC_1)$, and a view-2 private part in $\col(\bC_2)$, and independence guarantees that this split is unique. Were it not, some direction could be explained as either common or private, and the word ``common'' would lose its meaning. Independence is exactly what keeps a private direction from being rebuilt out of the common one or the other view's private part, so nothing private can quietly pass for common. Independence also implies a size requirement. The three subspaces can only stay independent if the ambient space has room for all of them at once, so their dimensions must fit, $N \ge D + L_1 + L_2$, and when they do their combined span has dimension exactly $D + L_1 + L_2$.

The second requirement, full column rank of $[\bS_k,\bH_k]$, keeps the
mixing from collapsing any common or private factor within a view. In the noiseless case, this requirement ensures $\col(\bX_k) = \col([\bG,\bC_k])$, so view $k$ keeps all $D + L_k$ of its common and private directions rather than losing some to a degenerate mix. A consequence of this requirement is that each view must be at least as wide as the number of factors, $M_k \ge D + L_k$. Later, we will see how these requirements translate into concrete system-design constraints, such as how many samples or antennas a deployment must provide for CCA to succeed.

\begin{Remark} The conditions of Theorem~\ref{thm:cca_id} are sufficient, and not all of them are necessary. When $[\bS_k,\bH_k]$ has full column rank for $k=1,2$, full column rank of $[\bG,\bC_1,\bC_2]$ is necessary for two views~\cite{sorensen2021}. In particular, $\col(\bC_1)\cap\col(\bC_2) = \{\mathbf 0\}$ is required, since a shared private direction would enter the intersection and be recovered as if it were common. Full column rank of $[\bS_k,\bH_k]$, on the other hand, is not generally necessary~\cite{sorensen2021}. A view can lose private directions to a degenerate private mixing $\bH_k$, and CCA still recovers $\col(\bG)$, because losing directions cannot enlarge the intersection of the two column spaces. The common mixing $\bS_k$ must be full column rank, since otherwise part of $\col(\bG)$ is missing from $\col(\bX_k)$. These facts follow from a single observation. In the noiseless case, CCA returns the intersection of the two column spaces exactly, so recovery succeeds precisely when that intersection is the modeled common subspace,
\begin{linenomath}
\begin{equation}
  \text{CCA identifies } \col(\bG)
  \iff \col(\bX_1)\cap\col(\bX_2) = \col(\bG).
\end{equation}
\end{linenomath}
The hypotheses of the theorem are a checkable way to guarantee the
right-hand side.
\end{Remark}

\begin{Connection}[{\color{black}Analogy}: Linear independence and Wyner's conditional independence]
\label{con:independence}
In the noiseless case, the model~\eqref{eq:gen_two_view} is deterministic and carries no distribution, so Wyner's conditional independence of Section~\ref{sec:common_info} is not literally defined here. What plays its role is the linear independence of the subspaces. Requiring the common and private subspaces to be jointly linearly independent, as Theorem~\ref{thm:cca_id} does, is the algebraic analogue of the private parts being conditionally independent given the common factor. Under this requirement, what CCA reports as common is genuinely common.
\end{Connection}

\begin{Connection}[{\color{black}Analogy}: One spectrum, two readings]
\label{con:two_readings}
A single set of canonical correlations carries both notions of
Section~\ref{sec:common_info}. The intersective reading, of Gács--Körner, lives at $\sigma_i = 1$, where the two canonical variates coincide and each view reconstructs the common component from its own data alone. Wyner's generative reading uses every nonzero $\sigma_i$, since one common variable must account for the whole dependence. For jointly Gaussian views, Wyner's quantity~\eqref{eq:wyner_gaussian} is finite only when every $\sigma_i < 1$, so the directions that Gács--Körner counts are exactly those that make Wyner's quantity infinite.
\end{Connection}

\subsection{Generalized CCA: From Two Views to $K$ Views}
\label{sec:gcca}

Thus far, we have studied the two-view setup, and we now understand that setup well. We know what CCA estimates, as well as models and conditions under which CCA recovers the true common latent subspace. The problems that opened this section, though, may involve more than two jointly observed random vectors, or views. A sensor network usually has many nodes, an fMRI study has many subjects, and a cellular uplink is heard by several base stations. \emph{Generalized canonical correlation analysis} (GCCA) conceptually extends CCA to $K \ge 2$ views. Rather than aligning a single pair, GCCA aims to align all views at once, producing reduced views $\bX_k\bQ_k$ that are close to each other. Unlike the two-view case, though, GCCA is not a single method. The quality of a $K$-view alignment can be measured in several inequivalent ways, and for $K > 2$ these give genuinely different formulations, and different solutions. In this tutorial, we focus on the five classical formulations collected by Kettenring~\cite{kettenring1971}, two of which he introduced. For two views, all five reduce to Hotelling's classical CCA~\cite{kettenring1971}. 

\subsubsection{Kettenring's Five Formulations}

Kettenring~\cite{kettenring1971} laid out five criteria for generalized
CCA. Each is a functional of the \linebreak $K\times K$ correlation matrix $\bPhi\left(\mathbf{q}_1,\ldots,\mathbf{q}_K\right)$ of the projected views. Its entries are \linebreak $\Phi\left(\mathbf{q}_1,\ldots,\mathbf{q}_K\right)_{kl} = \bq_k^\top\bSigma_{kl}\bq_l$, where $\bSigma_{kl} = \Cov(\bx_k,\bx_l)$ is the cross-covariance between views $k$ and $l$, while the unit-variance constraints
$\bq_k^\top\bSigma_{kk}\bq_k = 1$ fix the diagonal to $\Phi\left(\mathbf{q}_1,\ldots,\mathbf{q}_K\right)_{kk}=1$. When the projected views are mutually uncorrelated, $\bPhi\left(\mathbf{q}_1,\ldots,\mathbf{q}_K\right)$ is the identity, since its diagonal is always one and its off-diagonals vanish. Alignment pushes $\bPhi\left(\mathbf{q}_1,\ldots,\mathbf{q}_K\right)$ away from this baseline, spreading its eigenvalues away from one, and each of the five criteria measures that departure in a different way. As this may not be the first framing that comes to mind, we will see that each of the five also admits a more direct and intuitive problem formulation. Next, all five are stated here for a single component, one direction $\bq_k$ per view. Each extends to a $D$-dimensional representation by deflation, extracting further components in turn under a restriction that keeps them from repeating earlier ones. A simple choice requires the new variate of each view to be uncorrelated with the earlier variates of that view, and Kettenring discusses several alternatives, which can lead to different solutions~\cite{kettenring1971}.

\paragraph{SUMCOR (Sum of Correlations)}
The direct goal is to make the projected views, $z_k = \mathbf{q}_k^\top\mathbf{x}_k$, as correlated as possible on average, by maximizing the total pairwise correlation across the views,
\begin{linenomath}
\begin{equation}
  \max_{\{\bq_k\}_{k=1}^K} \; \sum_{k=1}^{K}\sum_{l=1}^K\mathrm{corr}(z_k,z_l)=\sum_{k=1}^{K}\sum_{l=1}^K \bq_k^\top\bSigma_{kl}\bq_l 
  \quad\text{s.t.}\quad \bq_k^\top\bSigma_{kk}\bq_k = 1, \;\; k=1,\ldots,K.
\end{equation}
\end{linenomath}
Since each term is
$\mathrm{corr}(z_k,z_l) = \bq_k^\top\bSigma_{kl}\bq_l = \Phi\left(\mathbf{q}_1,\ldots,\mathbf{q}_K\right)_{kl}$, the objective reduces to the sum of all entries of $\bPhi\left(\mathbf{q}_1,\ldots,\mathbf{q}_K\right)$,
$\mathbf{1}^\top\bPhi\left(\mathbf{q}_1,\ldots,\mathbf{q}_K\right)\,\mathbf{1}$. SUMCOR is the only formulation of the five whose value changes under sign flips of the $\bq_k$, and the formulation is due to Horst~\cite{horst1961b}.

\paragraph{MAXVAR (Maximum Variance)}
The goal of this formulation is to find a single consensus variable with
which all $K$ projected views are as correlated as possible on average.
Introducing an auxiliary consensus random variable $g$ of unit variance, MAXVAR maximizes the total squared correlation of the projected views with $g$,
\begin{linenomath}
\begin{equation}
  \max_{g,\{\bq_k\}_{k=1}^K} \; \sum_{k=1}^{K}\mathrm{corr}^2(z_k, g)
  \quad\text{s.t.}\quad \bq_k^\top\bSigma_{kk}\bq_k = 1, \;\; \Var(g)=1.
  \label{maxvar}
\end{equation}
\end{linenomath}
It can be shown that for fixed vectors $\bq_k$, the maximum over $g$ equals the largest eigenvalue $\lambda_{\max}(\bPhi(\mathbf{q}_1,\ldots,\mathbf{q}_K))$, so MAXVAR maximizes $\lambda_{\max}(\bPhi(\mathbf{q}_1,\ldots,\mathbf{q}_K))$ over the $\bq_k$. The largest value of $\lambda_{\max}$ is attained when $\bPhi\left(\mathbf{q}_1,\ldots,\mathbf{q}_K\right)$ becomes rank one, indicating that all views exactly agree on the single consensus $g$, up to sign. To see where the name comes from, consider a combination $w = \sum_{k} a_k z_k$ with $\|\mathbf{a}\|_2 = 1$. The variance of $w$ is $\mathbf{a}^\top\bPhi(\mathbf{q}_1,\ldots,\mathbf{q}_K)\,\mathbf{a}$, which is at most $\lambda_{\max}$, with equality when $\mathbf{a}$ is the top eigenvector. MAXVAR therefore maximizes the variance of the best such combination, hence the name. This consensus formulation originates with Horst~\cite{horst1961b} and Carroll~\cite{carroll1968}, and this formulation is the one we develop further below. The main reason is that, as we will see later, MAXVAR yields a direct eigendecomposition rather than the iterative solution SUMCOR requires.

\paragraph{SSQCOR (Sum of Squared Correlations)}
Here, the goal is again to make the projected views as correlated as
possible, but now measuring correlation without regard to sign, by
maximizing the total squared pairwise correlation across the views,
\begin{linenomath}
\begin{equation}
  \max_{\{\bq_k\}_{k=1}^K} \; \sum_{k=1}^{K}\sum_{l=1}^{K}\mathrm{corr}^2(z_k,z_l)
  \quad\text{s.t.}\quad \bq_k^\top\bSigma_{kk}\bq_k = 1, \;\; k=1,\ldots,K.
\end{equation}
\end{linenomath}
Since $\mathrm{corr}(z_k,z_l) = \bq_k^\top\bSigma_{kl}\bq_l = \bPhi\left(\mathbf{q}_1,\ldots,\mathbf{q}_K\right)_{kl}$
under the constraints, the objective reduces to the sum of squared entries of $\bPhi\left(\mathbf{q}_1,\ldots,\mathbf{q}_K\right)$, $\|\bPhi\left(\mathbf{q}_1,\ldots,\mathbf{q}_K\right)\|_F^2 = \sum_i \lambda_i(\bPhi\left(\mathbf{q}_1,\ldots,\mathbf{q}_K\right))^2$, and hence the sum of its squared eigenvalues. SSQCOR sits between SUMCOR and MAXVAR. SSQCOR is insensitive to sign like MAXVAR, and its objective is quadratic in the number of views like SUMCOR, but quartic in the directions $\bq_k$. SSQCOR was originally proposed by Kettenring~\cite{kettenring1971}.

\paragraph{MINVAR (Minimum Variance)}
Contrary to MAXVAR, where all views try to reach a consensus, MINVAR seeks a direction along which the projected views nearly cancel, that is, a weighted combination of them with as little variance as possible. Introducing a vector of weights $\bw\in\mathbb{R}^K$ of unit norm, MINVAR minimizes the variance of that combination,
\begin{linenomath}
\begin{equation}
  \min_{\{\bq_k\}_{k=1}^K,\; \|\bw\|=1} \; \Var\Big(\sum_{k=1}^K w_k z_k\Big)
  \quad\text{s.t.}\quad \bq_k^\top\bSigma_{kk}\bq_k = 1, \;\; k=1,\ldots,K.
\end{equation}
\end{linenomath}
Since $\Var(\sum_k w_k z_k) = \bw^\top\bPhi(\bq_1,\ldots,\bq_K)\bw$, for fixed vectors $\bq_1,\ldots,\bq_K$, the inner minimization over $\bw$ returns the smallest eigenvalue. As a result, MINVAR minimizes $\lambda_{\min}(\bPhi(\bq_1,\ldots,\bq_K))$. Driving $\lambda_{\min}$ toward zero pushes $\bPhi(\bq_1,\ldots,\bq_K)$ toward singularity, so the projected views come as close as possible to a linear dependence. Such singularity can emerge when any two views become collinear, since two rows equal up to sign make $\bPhi(\bq_1,\ldots,\bq_K)$ singular, regardless of the other $K-2$ views. Although this is an extreme form of singularity, Kettenring considered MINVAR a heuristic for finding the most correlated pair of views. Compared with the other four formulations, MINVAR is the least direct at aligning all views equally, since its value keys on a single collinear pair rather than on the whole set. MINVAR was also originally proposed by Kettenring~\cite{kettenring1971}.

\paragraph{GENVAR (Generalized Variance)}

The last formulation of GCCA is closest to MINVAR in terms of philosophy.
GENVAR is also trying to drive the projected views toward dependence,
now measured by their \emph{generalized variance}, the determinant of their
correlation matrix, i.e.,
\vspace{-6pt}
\begin{adjustwidth}{-\extralength}{0cm}
\begin{linenomath}
\begin{equation}
  \min_{\{\bq_k\}_{k=1}^K} \; \det\bPhi(\bq_1,\ldots,\bq_K)
  = \prod_i \lambda_i(\bPhi(\bq_1,\ldots,\bq_K))
  \quad\text{s.t.}\quad \bq_k^\top\bSigma_{kk}\bq_k = 1, \;\; k=1,\ldots,K.
\end{equation}
\end{linenomath}
\end{adjustwidth}
The determinant measures the volume spanned by the projected views, equal to one when they are uncorrelated ($\bPhi = \bI$) and vanishing as they approach linear dependence. For jointly Gaussian views, the determinant has a clean information-theoretic reading, since $-\tfrac{1}{2}\log\det\bPhi$ is their total correlation. In that case, GENVAR maximizes the total statistical dependence among the $K$ projected views. Because the determinant is the product of all eigenvalues, the determinant vanishes whenever any one of them does. However, unlike MINVAR, which watches only the smallest eigenvalue, GENVAR penalizes the whole spectrum, a distinction that matters only when perfect alignment, and hence a vanishing smallest eigenvalue, cannot occur. GENVAR is due to Steel~\cite{steel1951}.

The five agree in spirit and differ in emphasis. At the fully aligned
extreme $\bPhi = \mathbf{1}\mathbf{1}^\top$, SUMCOR and SSQCOR attain $K^2$, MAXVAR
attains $K$, and both $\det(\bPhi)$ and $\lambda_{\min}$ vanish.
Computationally, MAXVAR and MINVAR reduce to a single eigenvalue problem,
while SUMCOR, SSQCOR, and GENVAR require iterative
procedures~\cite{kettenring1971}. Beyond these five, the literature offers
many further generalizations. Nielsen~\cite{nielsen2002} pairs each
criterion with several constraint types. V\'ia, Santamar\'ia, and
P\'erez~\cite{via2007} give adaptive and least-squares formulations and
show the least-squares version coincides with MAXVAR. Regularized and
sparse variants extend GCCA to high-dimensional and structured
data~\cite{tenenhaus2011, witten2009}, and scalable algorithms handle very large problems~\cite{fu2017}. Nonlinear and deep extensions are deferred to Section~\ref{sec:nonlinear}.

\subsubsection{Revisiting the MAXVAR Formulation}
\label{sec:maxvar}

Among the five, MAXVAR is the one we focus on next, because MAXVAR admits a clean and scalable solution. Before analyzing it, we note that the formulation also admits a friendlier form than~\eqref{maxvar}, one that directly handles the extraction of $D > 1$ maximally correlated components. This form seeks a common representation $\bG \in \R^{N \times D}$ with orthonormal columns and a set of projections $\bQ_k$ that bring every view as close to $\bG$ as possible,
\begin{linenomath}
\begin{equation}
  \min_{\bG,\,\{\bQ_k\}_{k=1}^K}
  \sum_{k=1}^K \|\bX_k\bQ_k - \bG\|_F^2
  \quad\text{s.t.}\quad \bG^\top\bG = \bI_D.
  \label{eq:maxvar2}
\end{equation}
\end{linenomath}
This consensus formulation is due to Carroll~\cite{carroll1968}. For $D = 1$, it is equivalent to~\eqref{maxvar}. For $D > 1$, the constraint $\bG^\top\bG = \bI_D$ makes the consensus variables of successive components uncorrelated, which is one of the deflation restrictions discussed by Kettenring~\cite{kettenring1971}.

As we stated earlier, MAXVAR admits a closed-form solution. To see that, first, fix the common representation $\bG$ and minimize over the projections. Each term $\|\bX_k\bQ_k - \bG\|_F^2$ is a least-squares fit, minimized at\vspace{12pt}
\begin{linenomath}
\begin{equation}
  \bQ_k^* = \bX_k^\dagger\bG,
\end{equation}
\end{linenomath}
where $\dagger$ is the Moore--Penrose pseudoinverse. The fitted view is then $\bX_k\bQ_k^* = \bP_k\bG$, with $\bP_k = \bX_k\bX_k^\dagger =
\bU_k\bU_k^\top$ being the orthogonal projector onto $\col(\bX_k)$ and $\bU_k$ an orthonormal basis of $\col(\bX_k)$. Substituting back and using
$\|\bP_k\bG - \bG\|_F^2 = D - \Tr(\bG^\top\bP_k\bG)$, the objective becomes,
up to a constant, an eigenvalue problem in $\bG$,
\begin{linenomath}
\begin{equation}
  \max_{\bG^\top\bG = \bI_D} \;
  \Tr\!\Big(\bG^\top \Big(\textstyle\sum_{k=1}^K \bU_k\bU_k^\top\Big)\bG\Big).
  \label{eq:maxvar_reduced}
\end{equation}
\end{linenomath}
The solution of~\eqref{eq:maxvar_reduced} is given by the top $D$ eigenvectors of $\sum_k \bU_k\bU_k^\top$. A full eigendecomposition of $\sum_k \bU_k\bU_k^\top$  costs $\mathcal{O}(N^3)$ operations, which makes computing the solution challenging.

The authors of \cite{peppas2023harnessing} and \cite{peppas2025signal} noticed that this complexity can be reduced as the same solution is given equivalently by the top $D$ left singular vectors of the concatenated basis matrix $[\bU_1,\ldots,\bU_K] \in \R^{N \times \sum_k M_k}$ ({{here} 
we assume that each view has full column rank; otherwise, the basis $\bU_k$ of a rank-deficient view has fewer than $M_k$ columns}). So $\bG$ can be obtained from one thin SVD per view, to get the $\bU_k$, followed by a single SVD of the stacked bases. The computation costs $\mathcal{O}\big(N(\sum_k M_k)^2\big)$, linear in the number of samples $N$, which is what lets MAXVAR scale to large problems in terms of samples. Which route is preferable depends on the two dimensions. When the total view dimension $\sum_k M_k$ is smaller than the number of samples $N$, as in most large-sample settings, the second route wins, since it never forms or decomposes the $N \times N$ matrix. When instead $N < \sum_k M_k$, both routes cost $\mathcal{O}\big(N^2\sum_k M_k\big)$, since forming the $N \times N$ matrix already costs as much as the SVD of the stacked bases.

\subsubsection{Identifiability for $K$ Views}
\label{sec:gcca_id}

The two-view identifiability of the previous subsection extends to $K$
views, and the extension is favorable. More views impose more agreement, so the conditions for recovering the common subspace relax as $K$ grows.
S{\o}rensen et~al.~\cite{sorensen2021} made this
precise by reading GCCA as subspace intersection.

\begin{Theorem}[{\cite{sorensen2021}, Theorem~4.2}]
\label{thm:gcca_id}
  Let $\bX_k = \bG\bS_k^\top + \bC_k\bH_k^\top$ for $k = 1,\ldots,K$, with
  $\bG \in \R^{N \times D}$ spanning the common subspace and
  $\bC_k \in \R^{N \times L_k}$ the private subspace of view $k$. If
  \begin{linenomath}
  \begin{equation*}
    \boldsymbol{\Gamma}_K =
    \begin{bmatrix}
      \bC_1 & -\bG & -\bC_2 &        &       &          \\
      \bC_1 & &   &-\bG & -\bC_3 &        &                \\
      \vdots &    & &      &      \ddots    & \ddots  &     \\
      \bC_1 &     & &        &    &      -\bG  &  -\bC_K
    \end{bmatrix}
    \in \R^{(K-1)N \times \left((K-1)D + \sum_{k=1}^K L_k\right)}
  \end{equation*}
  \end{linenomath}
  has full column rank and $[\bS_k,\bH_k]$ has full column rank for $k = 1,\ldots,K$, then
  \linebreak $\dim\big(\bigcap_{k=1}^K \col(\bX_k)\big) = D$ and
  $\bigcap_{k=1}^K \col(\bX_k) = \col(\bG)$.
\end{Theorem}

The conditions above can be understood as follows. In the case of two views, full column rank of $\boldsymbol{\Gamma}_K$ boils down to the first identifiability condition for the two-view case of Ibrahim and Sidiropoulos~\cite{ibrahim2019}. 
Based on the assumed model, the intersection of the $K$ column spaces is $\col(\bG)$ enlarged by any directions that the private subspaces share outside $\col(\bG)$~\cite{sorensen2021}. Intuitively, $\boldsymbol{\Gamma}_K$ takes view 1 as a reference and asks whether any of its directions outside $\col(\bG)$ can be reproduced by every other view at the same time. Full column rank of $\boldsymbol{\Gamma}_K$ says that no such direction exists, since any one would yield a nonzero vector in the null space of $\boldsymbol{\Gamma}_K$, so $\col(\bG)$ is the only part shared by all views. Pairs of views may still share private directions, as long as the directions are not shared by every view. The rank condition is cumbersome to check directly, but it requires a dimension budget, $D + \tfrac{1}{K-1}\sum_k L_k \le N$. The same holds for the rank condition on $[\bS_k,\bH_k]$, which requires $D + L_k \le M_k$ for all $k$. Numerical experiments suggest that satisfying these budgets is also generically sufficient~\cite{sorensen2021}. The factor $\tfrac{1}{K-1}$ relaxes the budget as views are added. Hence, a common subspace that two views alone cannot separate can become identifiable once more views are available. For instance, with $N = 1000$, a common subspace of dimension $D = 10$, and private subspaces $L_k = 800$ in every view, the budget $D + \tfrac{1}{K-1}\sum_k L_k$ equals $1010$ at $K = 5$, exceeding $N$, so $\col(\bG)$ cannot be identified. At $K = 6$ the budget drops to $970$, so a sixth view makes identification possible. At $K = 2$ the budget reads $D + L_1 + L_2 \le N$, which is necessary for $[\bG,\bC_1,\bC_2]$ to have full column rank.

What recovers $\col(\bG)$ is the subspace intersection itself, not MAXVAR as such.
MAXVAR returns $\col(\bG)$ when the views are effectively full rank. In the experiments of S{\o}rensen et~al., MAXVAR fails to identify $\col(\bG)$ when the views are approximately low rank, whereas their subspace-intersection algorithm, which prescribes the signal rank of each view, still recovers $\col(\bG)$~\cite{sorensen2021}. The identifiability above is a property of the intersection, and any method that computes the intersection inherits the guarantee. As in the two-view case, the model is deterministic, and $\col(\bG)$ plays the role of the Gács--Körner common part of Section~\ref{sec:common_info}, the part that every view can reproduce on its own.

\subsection{Summary}
This section developed CCA as the second-order, sample-based counterpart of the common information of Section~\ref{sec:common_info}. Here, the protagonists were no longer scalar random variables but random vectors. Then, the question shifted to how many correlated variables can be formed by taking linear combinations within each view, and how to estimate them from data. Before moving on, it is worth collecting the main ideas.
\begin{itemize}
  \item \textbf{CCA finds the most correlated directions between two views.}
    Hotelling's problem is solved in closed form through an eigenvalue decomposition. The canonical correlations
    are the cosines of the principal angles between the two view subspaces, so CCA is at heart a statement about the relative geometry of those subspaces.
  \item \textbf{For the Gaussian case CCA computes Wyner's common information.} The canonical correlations determine both the mutual information, as shown by Scharf and Mullis and later by Borga, and Wyner's common information, as shown by Satpathy and Cuff. The balanced latent of the Bach--Jordan probabilistic model attains Wyner's common information. By analogy, the directions with canonical correlation one span the common subspace and play the role of the Gács--Körner common part of Section~\ref{sec:common_info}. Wyner's common information is finite and positive when some canonical correlations are nonzero and all are below one, that is, when the views are correlated but share no common subspace.
  \item \textbf{CCA identifies the common subspace under a generative model.} 
  When each view is a common component plus a private one, simple rank
conditions say when the shared subspace is recoverable from noiseless data, separating what CCA can and cannot~\mbox{disentangle}.
  \item \textbf{The framework extends to $K$ views.} The five formulations collected by Kettenring coincide for two views and diverge beyond. MAXVAR yields $D$ consensus variables, the columns of an orthonormal basis $\bG$, through an eigenvalue problem. Finally, $K$-view identifiability, cast as an intersection of subspaces, generically improves as more views are added.
\end{itemize}


The three problems that opened this section can now be addressed. For the sensor
network, the pairs of nodes are sorted by their canonical correlations, and the
strongly correlated pairs can now be encoded jointly rather than separately. For the
two voxel neighborhoods, the same quantities report whether the activity is
correlated, while the canonical variates carry the shared activity itself and not
only its strength. For the detection and decoding of cell-edge users, neighboring
base stations can forward their received measurements to a central unit. That unit
computes the canonical correlations and uses them as a detection criterion, i.e., a
high canonical correlation indicates the presence of one or more weak users. The
extracted variates can then be processed further to decode those users.
Section~\ref{sec:cca_id} gives the conditions under which the recovered subspace is
the one spanned by the cell-edge users. Generalized CCA handles the many nodes of a real network, and the many subjects of an fMRI study, in one step rather than pair by pair.

Our vantage point in this section was mostly statistical. Canonical correlation analysis and its generalizations were estimators of a shared structure from second-order statistics, and the probabilistic reading tied that structure back to the common-information notions of Section~\ref{sec:common_info}. Signal processing reaches the same object from a different starting point. There the common subspace arises from a concrete detection problem, where several sources are received across an array and the task is to tell them apart. The tools are the ones we have already used, common subspaces and the singular value decomposition, but there the guarantee that matters is operational, namely how little a receiver must know in advance to recover the sources. The next section follows that thread. That section starts with a single array and adds a second, at which point the common subspace of CCA reappears.

\section{The Signal Processing Thread, from Subspace Methods to CCA}
\label{sec:sp}

We now leave the statistical perspective of the previous section and take up
the one of signal processing. Subspace methods are a cornerstone of signal processing, and their starting point is a single observation about multichannel data. A multichannel receiver records $M$ channels per snapshot, which may be the elements of an antenna array or, as we see below, the chips of a spreading code.
Gathering a whole recording into a matrix, we let each row hold one snapshot across all channels and each column hold the entire history of one channel. As a result, $M$ channels observed over $N$ snapshots give a matrix $\bX \in \R^{N \times M}$. {{All} quantities in this section are real valued. The real convention costs nothing, since a complex baseband record splits into its real and imaginary
parts, which stack into a real matrix with twice as many channels.} 
What the receiver records is a superposition of a few sources, such as communication users, radar echoes or acoustic sources. When the sources are few and the channels many, the noiseless data matrix has low rank.

The row space of the noiseless data matrix is a \emph{signal subspace} and everything orthogonal to that subspace is a \emph{noise subspace} ({{array} processing usually arranges the snapshots as columns of the data matrix, in which case the signal subspace is the column space; with the snapshots as rows, as here, the same subspace is the row space}). The signal subspace matters because every source contributes to the data only through its direction across the $M$ channels, and all these directions lie in the signal subspace. Restricting attention to the signal subspace therefore keeps everything the sources contribute, discards the noise outside that subspace, and reduces the problem from $M$ dimensions to the number of sources. As long as the sources are strong enough, the two subspaces can be read off the eigendecomposition of a sample covariance matrix, or equivalently the SVD of $\bX$, without knowing anything about the sources themselves. This subspace estimate is the first step of a broad family of methods, from high-resolution direction-of-arrival estimation to interference suppression, blind channel estimation, and the multiuser detection we develop~here.

The signal subspace holds the contributions of all sources together, while a receiver is often interested in only some of them, such as the users that the receiver has to decode. {\color{black} To recover those users, the receiver must suppress the interference of all other users, who can be quite strong. Suppressing that interference amounts to finding the subspace that the desired users occupy, after which the receiver decodes the individual users within that subspace. Both steps rely on prior information, and this section asks how little a receiver must know in advance to carry them out. The question becomes most challenging for the weak users of interest, whose low received power makes their detection hard, especially in the presence of strong interference. We progress from receivers that rely on prior information about the users of interest to multiview receivers in which a second recording of the same transmission takes the place of most of that information. The second recording is provided either by a second receiver or by the transmitter itself, which sends its data twice. }

We can already see from this discussion how this problem connects to
Sections~\ref{sec:common_info} and~\ref{sec:cca}. The cell-edge users of Section~\ref{subsec:crossing-disciplines} are hard to decode, because every station receives them at low power. However, two stations that both receive a cell-edge user record the same transmission of that user. Similarly, when a user repeats a block, the two halves of the received block carry the same transmission. These recordings are natural candidates for views. Estimating the subspace occupied by the weak users then becomes the problem of estimating a
common subspace that the views share. The problem of interest here therefore
has an intersective character, aligned with the Gács--Körner notion of common
information from Section~\ref{sec:common_info}. At the same time, the several channels of each recording form a random vector, which serves as a view. A limited
number of snapshots for learning the joint distribution, together with an
object of interest that is itself a subspace, then point to algebraic methods
for computing the subspace intersection with CCA. The sufficient conditions for
recovering the common subspace, established abstractly in
Section~\ref{sec:cca_id}, therefore apply here as well, with the dimensions
counted in channels and block lengths.

The rest of this section is organized as follows. Section~\ref{sec:sp:model} states the signal model and the problem. Section~\ref{sec:sp:wangpoor} begins with a single-antenna receiver and code-division multiple access~(CDMA), where each user spreads its symbols with a signature known to the receiver, and presents the subspace approach of Wang and Poor, in which only the signature of the user of interest must be known. Section~\ref{sec:sp:pilots} turns to a receiver with multiple antennas, where known pilot symbols take the place of the signature. When only the pilots of the user of interest are known, however, strong interference degrades the estimates they provide. Section~\ref{sec:sp:leap} shows that a second station receiving the same users provides a multiview setup, in which CCA discards the interference that is not shared between the stations, without prior information. Section~\ref{sec:sp:celledge} applies this idea to recovering cell-edge users. Section~\ref{sec:sp:mmse} creates the second view at the transmitter, by repetition, and shows that CCA then yields the linear minimum mean-square-error~(MMSE) detector without channel estimation. Finally, Section~\ref{sec:sp:synthesis} compares what each receiver must know and what that knowledge costs, and discusses two assumptions whose relaxation leads to the rest of this~tutorial.

\subsection{The Signal Model and the Problem}
   \label{sec:sp:model}

Let $\bX \in \R^{N \times M}$ be the matrix of noisy multichannel recordings of
$R \ll M$ sources,
\begin{linenomath}
\begin{equation}
  \bX = \bB\bH^\top + \bE \in \R^{N \times M}.
  \label{eq:signal_model}
\end{equation}
\end{linenomath}
Here, $\bB \in \R^{N \times R}$ is the \emph{source matrix}, whose columns are
the source waveforms across $N$ snapshots, $\bH \in \R^{M \times R}$ holds the
channel responses or mixing coefficients, and $\bE$ is noise. When $\bB$ and
$\bH$ have full column rank, as assumed here, the noiseless data matrix $\bB\bH^\top$ has rank $R$. This low-rank structure is what the subspace decomposition exploits. In the noiseless case, the SVD
$\bX = \bU\,\mathrm{diag}(s_1, s_2, \ldots)\,\bV^\top$, with
$s_1 \ge s_2 \ge \cdots \ge 0$, exposes two subspaces at once. The first $R$ columns of $\bU$ span $\col(\bB) \subset \R^{N}$, the subspace in which the
source waveforms live, while the first $R$ columns of $\bV$ span $\col(\bH) \subset \R^{M}$, the subspace spanned by the channel vectors. The span $\col(\bH)$ is what array processing conventionally calls the signal subspace, and many classical detectors operate on $\col(\bH)$. The span $\col(\bB)$ is its counterpart in snapshot space. 

For a base station, for example, the sources are the users in its cell, whose transmitted waveforms span $\col(\bB)$ and whose channel vectors span $\col(\bH)$. Two base stations that both receive the cell-edge users record two different matrices $\bX_1$ and $\bX_2$. Since both recordings capture the cell-edge users, the transmissions of those users lie in both $\col(\bB_1)$ and $\col(\bB_2)$, and hence in their intersection $\col(\bB_1)\cap\col(\bB_2)$. This is why the snapshot space becomes the object of interest from Section~\ref{sec:sp:leap} onward.

The two subspaces $\col(\bH)$ and $\col(\bB)$ are the leading eigenspaces of two matrices formed from $\bX$. The $M \times M$ sample covariance $\tfrac{1}{N}\bX^\top\bX$ exposes $\col(\bH)$, and the $N \times N$ Gram matrix $\tfrac{1}{M}\bX\bX^\top$ exposes $\col(\bB)$. The decomposition is exact when $\bE = \mathbf 0$ and stays reliable as long as the signal-to-noise ratio keeps the $R$ signal singular values above the noise floor. {{Recovering} the source matrix $\bB$
itself from $\bX$ is the problem of blind source separation, and a prominent
method for it is independent component analysis (ICA)~\cite{hyvarinen2001}.
ICA assumes the sources to be statistically independent and non-Gaussian.
Borga's generalized-eigenvalue framework~\cite{borga2001} places CCA, PCA, PLS
and multilinear regression in a single family built on second-order
statistics, and ICA lies outside that family, relying instead on higher-order~statistics.}

The problem we follow through this section is how a receiver can recover some of
the columns of $\bB$, namely the waveforms of the users of interest, from $\bX$,
while the remaining columns act as interference. How this can be done depends on
what the receiver knows in advance about $\bB$ and $\bH$. The receiver may know the directions of the columns of $\bH$, or a few entries of $\bB$, or only the values that these entries can take. Alternatively, the transmitter may repeat the entries of a column of $\bB$ by design. The subsections that follow differ precisely in what is built into the transmission and known to the receiver.

\subsection{Multiuser Detection in CDMA: The Wang--Poor Approach}
\label{sec:sp:wangpoor}

Subspace methods did not begin in communications. They were originally
developed for sensor array processing and spectral estimation, and they were
carried into multiuser detection later~\cite{wangpoor1998}. As a first step
toward recovering the symbol streams of some users, we assume structure that
makes the problem easier. Specifically, we assume that the subspace
$\col(\bH)$ is fixed by design, and that the receiver knows the column of
$\bH$ that corresponds to a user of interest, up to a scalar. This assumption
is realistic in the classical setting of multiuser detection, code-division
multiple access~(CDMA), where the prior information of the receiver takes a
concrete form, the signature and timing of each user.

Under CDMA, several users transmit simultaneously over the same frequency band,
and each is assigned a distinct signature waveform, or spreading code, that the
receiver uses to tell them apart. A signature is a sequence of $M$ chips, where a chip corresponds to the elementary pulse from which the code is built, and $M$ is called the processing gain. Here, we assume a single-antenna receiver and synchronous users. User $r$ is assigned a signature $\bd_r \in \R^{M}$, normalized to unit norm, and transmits each of its symbols $b_r[n]$ by multiplying the symbol with $\bd_r$, so that each symbol is spread over $M$ chips. Each user therefore sends one symbol per $M$ chips, and its rate drops to $1/M$ of what an unspread
transmission would achieve. The channel scales the transmission of user $r$ by a gain $h_r$, and the receiver records the sum over all users. Sampling the received signal once per chip during symbol interval $n$ gives the $M$-dimensional snapshot\vspace{12pt}
\begin{linenomath}
\begin{equation}
  \bx[n] = \sum_{r=1}^{R} h_r\, b_r[n]\, \bd_r + \be[n],
  \label{eq:cdma_snapshot}
\end{equation}
\end{linenomath}
where $\be[n]$ is noise. This is the synchronous model of~\cite{wangpoor1998}. Every snapshot is a linear combination of the same $R$ vectors $h_r\bd_r$, and only the coefficients $b_r[n]$ change from one symbol to the next.

Stacking the snapshots $\bx[n]^\top$ over $N$ symbol intervals as the rows of $\bX$ gives exactly~\eqref{eq:signal_model}, where the columns of $\bB \in \R^{N \times R}$ are now the users' symbol streams, which we take to be uncorrelated with unit power, so that each row of $\bB$ has identity covariance. With $\bD \in \R^{M \times R}$ holding the signatures $\bd_r$ in its columns and $\bh \in \R^{R}$ the gains $h_r$, the mixing matrix is
\begin{linenomath}
\begin{equation}
  \bH = \bD\,\mathrm{diag}(\bh).
  \label{eq:effective_signatures}
\end{equation}
\end{linenomath}
The noiseless data matrix has rank $R$ as long as $R \le M$ and the
signatures are linearly independent, which is what the processing gain buys,
since $M$ chips leave room for up to $M$ independent signatures. The columns of $\bH$ are \emph{effective} signatures, in which the assigned code fixes the direction and the channel fixes only the norm and sign. As a result, $\col(\bH) = \col(\bD)$ whenever every $h_r \neq 0$, so the signal subspace is fixed by design. 

The factorization~\eqref{eq:effective_signatures} separates two kinds of prior
information. The factor $\bD$ is \emph{designed}, fixed by the system before
any transmission, so a receiver may know $\bD$ exactly. The factor
$\mathrm{diag}(\bh)$ is \emph{propagated}, set by the physical channel from
each user to the receiver, so a receiver can only estimate $\mathrm{diag}(\bh)$
from the data, and the estimate is poorest for the weak users. The task is to recover a given user's symbols while suppressing the others, whose transmissions act as multiple-access interference. The receivers we consider next are linear filters $\bw \in \R^M$, known as \emph{detectors}, applied as $\bX\bw$.

The classical linear detectors set the baseline. The decorrelating detector of Lupas and Verd\'u~\cite{lupasverdu1989}, also called the zero-forcing detector because it drives the interference to zero, applies the inverse of the
cross-correlation matrix $\bD^\top\bD$ of the assigned signatures to the
matched-filter outputs, which cancels the multiple-access interference entirely
at the price of amplifying the noise. Note that $\bh$ does not appear in
$\bD^\top\bD$, so the detector cannot depend on how strong the interfering users
are. This is the property known as \emph{near--far resistance}, and it will
matter again in Section~\ref{sec:sp:leap}. The linear MMSE receiver of Madhow and Honig~\cite{madhowhonig1994} instead minimizes the MSE between its output and the transmitted symbol, balancing interference suppression against noise enhancement. To strike that balance, the MMSE receiver inverts the
received covariance
\begin{linenomath}
\begin{equation}
  \bSigma = \bD\,\mathrm{diag}(\bh)^2\,\bD^\top + \sigma^2\bI,
  \label{eq:cdma_cov}
\end{equation}
\end{linenomath}
which follows from~\eqref{eq:signal_model} and~\eqref{eq:effective_signatures}
for white noise of variance $\sigma^2$. The MMSE receiver is near--far
resistant as well~\cite{wangpoor1998}, since $\bSigma^{-1}$ suppresses each
interferer more strongly the stronger that interferer is received. The interferers need not use CDMA either, since the linear MMSE detector also
suppresses a rich variety of narrowband interference from coexisting
systems~\cite{poorwang1997part1, buzzi2002overlay}.

Both detectors are treated in the standard reference on the subject~\cite{verdu1998}. Their cost in prior information is substantial and asymmetric. The zero-forcing detector needs the signatures of \emph{all} users. The MMSE receiver needs those signatures too, together with the users' \emph{received amplitudes}, i.e., the entries of $\bh$, and the \emph{noise variance} $\sigma^2$, since $\bh$ sits between the assigned signatures in~\eqref{eq:cdma_cov} and cannot be factored out. The received amplitudes depend on the transmit power and the path loss, so they are propagated quantities rather than designed ones. For the desired user, whose signature is the $r$-th column $\bd_r$ of $\bD$, the
two classical detectors are\vspace{12pt}
\begin{linenomath}
\begin{equation}
  \bw_r^{\mathrm{zf}} \;\propto\; \bD(\bD^\top\bD)^{-1}\mathbf{e}_r,
  \qquad
  \bw_r^{\mathrm{mmse}} \;\propto\; \bSigma^{-1}\bd_r,
  \label{eq:classical_detectors}
\end{equation}
\end{linenomath}
where $\mathbf{e}_r$ is the $r$-th standard basis vector. For a row $\bx^\top$
of $\bX$, the output $\bx^\top\bw_r^{\mathrm{zf}}$ is the $r$-th entry of
$(\bD^\top\bD)^{-1}\bD^\top\bx$, which is the decorrelating operation described
above applied to the matched-filter outputs $\bD^\top\bx$.

\subsubsection*{{What} 
 the Signal Subspace Buys}

Wang and Poor~\cite{wangpoor1998} showed that both detectors can be built from
the \emph{signal subspace} alone. Take the $M \times M$ sample covariance
$\hat\bSigma = \tfrac{1}{N}\bX^\top\bX$ of Section~\ref{sec:sp:model}. Its
population form is the one in~\eqref{eq:cdma_cov}, and the eigendecomposition of that form separates into a signal part and a noise part,
\begin{linenomath}
\begin{equation}
  \bSigma = \bU_s\bLambda_s\bU_s^\top + \sigma^2\bU_n\bU_n^\top,
  \label{eq:wp_evd}
\end{equation}
\end{linenomath}
where $\col(\bU_s) = \col(\bD)$, assuming $\mathrm{diag}(\bh)$
in~\eqref{eq:effective_signatures} is invertible. What the eigendecomposition
exposes is the span of the assigned signatures, which the channel gains do not
change, and the leading eigenvectors of $\hat\bSigma$ estimate that span.

In terms of the signal subspace, the two detectors
of~\eqref{eq:classical_detectors} can be written as
\begin{linenomath}
\begin{equation}
  \bw_r^{\mathrm{zf}} \;\propto\;
    \bU_s(\bLambda_s - \sigma^2\bI)^{-1}\bU_s^\top\bd_r,
  \qquad
  \bw_r^{\mathrm{mmse}} \;\propto\;
    \bU_s\bLambda_s^{-1}\bU_s^\top\bd_r,
  \label{eq:wp_detectors}
\end{equation}
\end{linenomath}
each up to a positive scale factor. Everything on the right-hand side except $\bd_r$ is estimated from the received data. The interfering users' signatures are never needed, and neither are the received amplitudes, because comparing~\eqref{eq:wp_evd}
with~\eqref{eq:cdma_cov} gives
\begin{linenomath}
\begin{equation}
  \bU_s(\bLambda_s - \sigma^2\bI)\bU_s^\top
    = \bD\,\mathrm{diag}(\bh)^2\,\bD^\top,
  \label{eq:powers_from_eigs}
\end{equation}
\end{linenomath}
so the signal eigenpairs provide exactly the amplitude information that the
classical MMSE receiver had to be told. This is \emph{blind multiuser detection}, blind in the sense that it requires neither training symbols nor any knowledge of the interfering users, needing only the signature and timing of the user of
interest.

Two consequences of~\eqref{eq:wp_detectors} matter for what follows. First,
the two detectors of~\eqref{eq:wp_detectors} coincide as $\sigma \to 0$, so at high signal-to-noise ratio the distinction between zero-forcing and MMSE disappears.
Second, and most important here, the receiver must still be handed the
signature $\bd$. Given the signature, the timing can itself be estimated from
the data. The standard sliding correlator fails when the received powers are
very unequal~\cite{strom1996}, but subspace-based and maximum-likelihood delay
estimators remain near--far r\mbox{esistant~\cite{strom1996, bensleyaazhang1996}}.
Without the signature, however, nothing in a single covariance singles out a
weak user from the interferers that share the signal subspace, for
synchronization or for detection.

The subspace approach also survives departures from the model above. Under
multipath, the effective signature is no longer a scaled copy of the assigned
one, but the effective signature can still be estimated from the same subspace,
given the assigned one~\cite{wangpoor1998}. The number of active users need not
be known either, since information-theoretic criteria estimate the rank of the
signal subspace from the eigenvalues of $\hat\bSigma$~\cite{waxkailath1985}.
Wang and Poor's monograph~\cite{wangpoor2003book} covers the approach in full.
In all of these variants, however, the receiver must be handed the assigned
signature of the user of interest. That signature is the one piece of prior
information that the signal subspace cannot replace, and for a weak user, such
as a user at the cell edge, the base station may not know the signature at all.
The rest of this section drops the signatures and turns to receivers with
multiple antennas. Such a receiver is first handed some of the symbols of the
user of interest, and later a second view of the same transmission takes the
place of most of this prior information.

\subsection{Multiple Antennas and Pilot Symbols}
\label{sec:sp:pilots}

The receiver of Section~\ref{sec:sp:wangpoor} was aided by knowing the
signature of the user of interest, which is fixed by design. Next, instead of
spreading the transmissions over chips with a code, we consider a receiver with
multiple antennas, and keep $M$ for the number of antennas. Each user sends one
symbol per snapshot, and the snapshot equation becomes
\begin{linenomath}
\begin{equation}
  \bx[n] = \sum_{r=1}^{R} b_r[n]\, \bh_r + \be[n],
  \label{eq:ma_snapshot}
\end{equation}
\end{linenomath}
where $\be[n]$ is noise. The model remains~\eqref{eq:signal_model}, but the
$r$th column $\bh_r$ of $\bH$ is now the array response of user $r$, set
entirely by propagation, so the receiver knows neither any column of $\bH$ nor $\col(\bH)$ in advance. On the other hand, under the assumptions of Section~\ref{sec:sp:wangpoor}, the two classical linear detectors for user $r$ are
\begin{linenomath}
\begin{equation}
  \bw_r^{\mathrm{zf}} = \bH(\bH^\top\bH)^{-1}\mathbf{e}_r,
  \qquad
  \bw_r^{\mathrm{mmse}} = \bSigma^{-1}\bh_r,
  \qquad
  \bSigma = \bH\bH^\top + \sigma^2\bI.
  \label{eq:pilot_detectors}
\end{equation}
\end{linenomath}
Since $\bH(\bH^\top\bH)^{-1}\mathbf{e}_r = (\bH\bH^\top)^{+}\bh_r$ with
$\bH\bH^\top = \bSigma - \sigma^2\bI$, both can be estimated, exactly as
in~\eqref{eq:wp_detectors}, from the signal subspace of $\hat\bSigma$ together
with the single direction $\bh_r$. The designed signature $\bd_r$ of the CDMA
receiver is thus replaced by the propagated response $\bh_r$, which is unknown
to the receiver.

The covariance $\bSigma$ and $\sigma^2$ can be easily estimated from the received data, so what the receiver needs is the direction $\bh_r$. Designed knowledge makes the estimation of that direction possible through \emph{pilot} symbols, which
precede the data symbols and are known to the receiver. As a result, the
receiver knows a few rows of $\bB$ rather than a column of $\bH$. There are two
options for using the pilots, which differ in whose pilots the receiver must~know.

In the first option, every user sends $N_p$ pilot symbols over the same interval.
Collecting the pilot snapshots in $\bX_p \in \R^{N_p \times M}$ and the pilot
symbols in $\bB_p \in \R^{N_p \times R}$, the least-squares estimate of all the
channels is
\begin{linenomath}
\begin{equation}
  \hat\bH^\top = (\bB_p^\top\bB_p)^{-1}\bB_p^\top\bX_p,
  \label{eq:pilot_ls}
\end{equation}
\end{linenomath}
which requires $N_p \ge R$ and is simplest for orthogonal pilots, with
$\bB_p^\top\bB_p = N_p\bI$. Either detector in~\eqref{eq:pilot_detectors}
then follows by substitution.

In the second option, only the pilots of the user of interest, collected in
$\bb_{r,p} \in \R^{N_p}$, are known to the receiver. The least-squares fit of
$\bX_p$ by $\bb_{r,p}\bh^\top$ gives, for unit-power pilots, the estimate
$\hat\bh_r = \tfrac{1}{N_p}\bX_p^\top\bb_{r,p}$, which treats the data of the
other users as part of the noise. Either detector
in~\eqref{eq:pilot_detectors} then follows by substituting $\hat\bh_r$
together with $\hat\bSigma$ computed from all $N$ snapshots. For the MMSE
detector, a second estimator is available, which fits $\bX_p\bw$ to
$\bb_{r,p}$ in the least-squares sense and gives
$(\bX_p^\top\bX_p)^{-1}\bX_p^\top\bb_{r,p}$. This estimator implements the same
detector $\bSigma^{-1}\bh_r$, but estimates $\bSigma$ from the $N_p$ pilot
snapshots alone, so it uses less data and requires $N_p \ge M$.

The two options trade coordination against accuracy. The first requires the
receiver to know the pilots of every user, sent over an aligned interval, which
is a strong assumption, especially when the users belong to different cells. In return, the receiver can separate all users during the pilots, so only the noise
corrupts the estimate of $\bh_r$. The second requires only the pilots of the
user of interest, but the data of the other users then act as interference on
the estimate, and strong interferers corrupt the estimate, especially that of a weak user's response. In both cases the estimate improves as $N_p$ grows, so a weak user needs long pilots, particularly in the presence of strong interferers. 
Pilots also serve synchronization, since the receiver locates the pilot block by correlating the received signal with the known pilots, and in the second option this correlation suffers from the same interference. However, pilot snapshots carry no data, and they must be repeated whenever the channel changes. 

In the rest of this section, a second view of the same transmission takes over
the role of the pilots in suppressing the interference, so that at most a few
known symbols, or only the values that the symbols can take, are needed
afterwards.

\subsection{The Multi-Array Leap from Signal Subspace to Common Subspace}
\label{sec:sp:leap}
 
Consider again the cell-edge problem of Section~\ref{subsec:crossing-disciplines}, illustrated in Figure~\ref{fig:celledge}. Users near the shared boundary of two cells are attenuated by path loss, so their signals arrive at each base station at low signal-to-interference-plus-noise ratio. At the same time, they are received by
\emph{both} stations at comparable power, whereas a user close to one station
is received strongly there and lies near the noise floor at the other. This is
the near--far problem of Section~\ref{sec:sp:wangpoor}, with the cell-edge users in the far position at both stations at once. The receivers of the previous subsections handle that problem only when they are handed the signature of the user of interest or long pilots, which are exactly what may not be known for a cell-edge user. Without them, the cell-edge transmissions cannot be extracted from the signal subspace that they span together with the strong cell-center users. Even when the pilots of a cell-edge user are known, the strong cell-center users corrupt both the timing and the estimated response of that user unless the pilots are sufficiently long, as discussed in Sections~\ref{sec:sp:wangpoor} and~\ref{sec:sp:pilots}.
 
Ibrahim and Sidiropoulos~\cite{ibrahim2019, ibrahim2020} proposed to combine the recordings of the two stations in order to recover the cell-edge users, and they named their strategy base-station interferometry. They assume that both stations
receive the same cell-edge transmissions, so, as noted in
Section~\ref{sec:sp:model}, the cell-edge waveforms lie in
$\col(\bB_1)\cap\col(\bB_2)$. If each cell-center user is received essentially
by one station only, the cell-edge transmissions are the only ones that appear
in $\col(\bB_1)\cap\col(\bB_2)$. A single station has nothing that singles out
the cell-edge users, but the two stations together do, since the cell-edge
users are exactly the ones they share. Finding the subspace shared by two data
matrices is the problem that CCA solves in Section~\ref{sec:cca}. 
Under the identifiability conditions of Section~\ref{sec:cca_id}, CCA recovers
$\col(\bB_{ce}) = \col(\bB_1)\cap\col(\bB_2)$, where the columns of
$\bB_{ce}$ are the cell-edge transmissions, without a signature or pilots, and
regardless of the power of the cell-center users at either~station.

The indifference to interference power follows from the invariance property of
\linebreak Section~\ref{sec:cca}, by which canonical correlations are unchanged under any
invertible linear transformation of either view. The power of each user enters
a view through its mixing matrix, which is such a transformation, so a
cell-center user is discarded for being private to one station, however
strongly that user is received. Noise, on the other hand, is added to a view
rather than applied to it as a transformation, so noise does affect the
canonical correlations. Hence, the cell-edge users must still be received above the noise at each station for a common subspace to emerge. The second station
therefore removes the effect of interference power, but not the effect of
noise. Next, we present the work of Ibrahim and
Sidiropoulos~\cite{ibrahim2019, ibrahim2020} in detail.

\subsection{Cell-Edge Detection}
\label{sec:sp:celledge}

Two base stations, each with $M_k$ antennas, jointly observe $D$ cell-edge
users, which sit near a shared boundary and are received by both. At the same time,
$L_k$ cell-center users are received only by station $k$. Which of the two
stations administratively serves a given cell-edge user plays no role here,
since what matters is only that both receive it. Over $N$ snapshots, station $k$
records
\begin{linenomath}
\begin{equation}
  \bX_k
  = \bG\,\bS_k^\top
  + \bC_k\,\bH_k^\top
  + \bE_k
  \in \R^{N \times M_k},
  \qquad k = 1,2,
  \label{eq:cell_edge_model}
\end{equation}
\end{linenomath}
and both of them forward their recordings to a central processing unit.
The matrix \linebreak $\bG \in \R^{N \times D}$ carries the $D$ cell-edge symbol streams,
which are \emph{common} because both stations receive these users, and
$\bC_k \in \R^{N \times L_k}$ carries the cell-center streams that are
\emph{private} to station $k$. The matrices $\bS_k$ and $\bH_k$ hold
the corresponding channels, and $\bE_k$ is noise plus residual
interference. In terms of~\eqref{eq:signal_model}, station $k$ has source
matrix $\bB_k = [\bG,\ \bC_k]$ and mixing matrix $[\bS_k,\ \bH_k]$,
so the cell-edge subspace $\col(\bB_{ce})$ of Section~\ref{sec:sp:leap} is
$\col(\bG)$. Returning to Figure~\ref{fig:celledge}, $\bG$ collects the
transmissions of the users at the boundary and $\bC_k$ those of the users
well inside cell $k$.
 
Equation~\eqref{eq:cell_edge_model} is not merely analogous to the generative
model~\eqref{eq:gen_two_view} of Section~\ref{sec:cca_id}. It is that model,
symbol for symbol, with the cell-edge symbol matrix $\bG$ in the role of the
common component and the cell-center symbol matrices $\bC_1, \bC_2$ in the role
of the private ones. The identifiability theory therefore transfers without
modification.
 
\begin{Theorem}[\cite{ibrahim2019, ibrahim2020}]
\label{thm:celledge_id}
In the noiseless case, if $[\bG, \bC_1, \bC_2] \in \R^{N \times (D+L_1+L_2)}$ has
full column rank and $[\bS_k, \bH_k]$ has full column rank for
$k = 1,2$, then the MAXVAR solution of the two-view CCA
problem~\eqref{eq:maxvar2} recovers $\col(\bG)$. Explicitly, the optimal common
representation is $\bG\bP$ for a nonsingular $D \times D$ matrix $\bP$.
\end{Theorem}
 
The hypotheses are exactly those of Theorem~\ref{thm:cca_id}, and they read here
as concrete system-design constraints, which is the payoff promised in
Section~\ref{sec:cca_id}. Full column rank of $[\bG,\bC_1,\bC_2]$ requires
$N \ge D + L_1 + L_2$ and the users' transmitted sequences to be linearly
independent, which is automatic with high probability for finite-alphabet signals and modest $N$, since the users transmit independently. Full column rank of $[\bS_k,\bH_k]$ relates antennas to users, in that $M_k \ge D + L_k$, so a station must have at least as many antennas as it has users to receive. Beyond that dimensionality budget, the condition holds with probability one for channels drawn from a jointly continuous distribution. {The cell-edge papers}~\cite{ibrahim2019,ibrahim2020} {work with complex baseband records and stack their real and imaginary parts into a real matrix with twice as many channels, which is where the factor of two in the antenna condition of}~\cite{ibrahim2020} {comes from.} {The results we present extend to the complex case in a straightforward way, as shown in those two papers.} 
 
It is worth pausing on what the first stage has achieved. The cell-edge subspace comes for free, from the two recordings alone. Prior knowledge could also handle the cell-edge users at a single station, but only with coordination between the two cells in advance. Spreading the cell-edge users with signatures would be one option, at the rate cost of Section~\ref{sec:sp:wangpoor}. Although Section~\ref{sec:sp:wangpoor} focused on a single-antenna receiver, extensions of CDMA to multi-antenna receivers are possible~\cite{sidiropoulos2000parafac}, but they lie beyond the focus of this tutorial. If instead a station knows orthogonal pilots of every user it receives, the channel estimates of Section~\ref{sec:sp:pilots} have an accuracy that does not depend on the interference power.

A weaker form of coordination, in which only the pilots of the cell-edge users are made known to the stations, would help little, since the strong cell-center
interference corrupts both the correlation-based synchronization and the channel
estimate of Section~\ref{sec:sp:pilots}. Once the common subspace is available,
however, this weaker form becomes sufficient, because the cell-center
interference does not enter $\col(\bG)$. Still, going from the subspace to the individual cell-edge streams requires additional structure, because the recovered $\bG\bP$ is a mixture of these streams and unmixing it calls for a second stage that exploits prior information. The known pilots of the cell-edge users are one such prior, and they now act on a mixture free of cell-center interference. In some cases, even these pilots are not needed for the unmixing.

One such case is the constant modulus of the transmitted symbols, which for the
binary alphabet $\{\pm 1\}$ used in~\cite{ibrahim2019, ibrahim2020} makes every
entry of $\bG$ equal in magnitude. Unmixing then amounts to factoring the
recovered matrix into a matrix of $\pm 1$'s and a nonsingular $D \times D$ matrix, a
problem that is NP-hard even when the latter is known~\cite{ibrahim2019,
ibrahim2020}. Van der Veen's Real Analytical Constant Modulus Algorithm
(RACMA)~\cite{vanderveen1997} does not attempt an optimal solution. Assuming the
noise is small, it reduces the factorization to a generalized eigenvalue problem
and returns the streams up to permutation and sign. Each user then carries a
unique identification sequence, and correlating the recovered streams with these
sequences resolves both ambiguities~\cite{ibrahim2019, ibrahim2020}.

\subsubsection{When Is a User Common?}

Theorem~\ref{thm:celledge_id} is noiseless, and in practice every user is
overheard, to some degree, at both stations. The appendix of~\cite{ibrahim2020}
analyzes this case and shows that, for nearly orthogonal sequences and large
arrays, the canonical correlation of each user is the product of two factors,
one per station. Each factor lies between zero and one and increases with the
signal-to-noise ratio at which that station receives the user. A high canonical
correlation therefore requires both factors to be high, and a single high factor
does not suffice. A cell-center user is received strongly by its serving station
but weakly by the other, so its canonical correlation stays low regardless of
how strong its signal is at the serving station. A cell-edge user is received by both stations, so its canonical correlation is limited only by its signal-to-noise
ratios. Strong interference is thus harmless to CCA precisely because it is
one-sided.

\subsubsection{The Same Computation Synchronizes the Two Stations}

So far, the two recordings were assumed aligned in time. If they are not, the
cell-edge streams in $\bX_1$ and $\bX_2$ no longer match row by row. Since
independent symbol streams shifted by even one symbol are already uncorrelated,
no common subspace remains and even the largest canonical correlation is small.
Computing the largest canonical correlation over a range of relative delays and
selecting the delay at which it peaks therefore recovers the
alignment~\cite{ibrahim2020}. Unlike the correlation-based synchronization of
Section~\ref{sec:sp:pilots}, which the strong cell-center interference corrupts,
the canonical-correlation-based approach requires no pilots. As a result, without any side information, CCA can (i) reveal whether cell-edge
users are present, (ii) synchronize the two stations when they are, and
(iii) estimate the cell-edge subspace, leaving potential prior information to be
used only for unmixing the cell-edge streams. 

{\color{black}Revealing the presence of
cell-edge users is a hypothesis testing problem. A closely related problem has been studied in the context of passive radar~\cite{santamaria2017passive}. There, two multi-antenna arrays, playing the role of our two stations, record signals, and the question is whether they share a common subspace signal of known dimension. The noise at each array may be arbitrarily correlated across its antennas but is uncorrelated with the noise at the other array. This setting is similar to the cell-edge one, as the cell-center interference can be seen as noise of this kind when its leakage to the other station is neglected. Modeling the common symbols and the noises as Gaussian, Santamar\'{\i}a et~al.~\cite{santamaria2017passive} derived a generalized likelihood ratio test that can be used for deciding whether a common signal is present by comparing a product of functions of the leading canonical
correlations to a threshold.}

\subsection{A Designed Second View and Blind MMSE Detection}
\label{sec:sp:mmse}

In Section~\ref{sec:sp:celledge}, the second view was provided by the network under specific circumstances. Users located at the shared boundary of two base
stations create an opportunity for a multiview approach to decoding them. Taking this opportunity, however, comes at the cost of infrastructure, namely the cooperation of two base stations and a central processing unit, while the CCA-based approach is a good fit only for the cell-edge users. In this subsection, the transmitter creates the second view by design. We return to the single multi-antenna receiver of Section~\ref{sec:sp:pilots}, with snapshots~\eqref{eq:ma_snapshot} and the MMSE detector $\bw_r^{\mathrm{mmse}} = \bSigma^{-1}\bh_r$ of~\eqref{eq:pilot_detectors}. There, the receiver estimates $\bh_r$ from pilots; 
 hence the detector is only as good as the channel estimate, and the pilots carry costs of their own. Long pilots cost spectral efficiency and, on time-varying channels, even degrade the channel estimate, while short ones average poorly and make pilot contamination more severe~\cite{ibrahim2024mmse}.

The scheme of Ibrahim, Karakasis and Sidiropoulos~\cite{ibrahim2024mmse} removes
the pilots used for channel estimation. Each user $r$ collects its $N_d$ data
symbols in $\bb_{c,r} \in \R^{N_d}$ and sends them twice, so that $N = 2N_d$,
permuting each copy with its own user-specific interleaver. The symbol stream of
user $r$ is thus
$\bb_r = [(\bPi_r^{(1)}\bb_{c,r})^\top\ (\bPi_r^{(2)}\bb_{c,r})^\top]^\top$, with
$\bPi_r^{(1)}, \bPi_r^{(2)}$ known to the receiver. When only one user repeats, the interleavers can be omitted, since the other users already transmit different symbols in the two halves. At the receiver, applying
$\bPi_r^{(i)\top}$ to the $i$th half of the recording reinstates the repetition
for user $r$ and leaves every other user scrambled differently in the two
halves. The two de-permuted halves form two views
\begin{linenomath}
\begin{equation}
  \bX_r^{(i)} = \bb_{c,r}\,\bh_r^{(i)\top}
  + \sum_{j \neq r} \bb_{j,r}^{(i)}\,\bh_j^{(i)\top} + \bE_r^{(i)}
  \in \R^{N_d \times M},
  \qquad i = 1,2,
  \label{eq:repetition_views}
\end{equation}
\end{linenomath}
where $\bh_j^{(i)}$ is the channel of user $j$ during the $i$th half and
$\bb_{j,r}^{(i)}$ is the $i$th copy of the symbols of user $j$ after
de-permutation with $\bPi_r^{(i)\top}$.

Since the channels may change between the two halves, the scheme also suits time-varying channels. The symbols $\bb_{c,r}$ of user $r$ are common to the two views, while the streams of all other users are rendered private by the interleaving. The two views thus follow the common-plus-private model~\eqref{eq:gen_two_view},
with the difference that the common component is now created by the transmitter's coding instead of the network deployment. Moreover, the common component is one-dimensional by design. According to~\cite{ibrahim2024mmse}, this allows more flexible waveforms, which may even be analog, and considerably reduces the complexity of the receiver. The other users need not follow the protocol or coordinate with user $r$ in any way, and their signals need not even be digital~\cite{ibrahim2024mmse}. Applying CCA to the pair of matrices $\bX_r^{(1)}$ and $\bX_r^{(2)}$ returns canonical directions $\bq_r^{(1)}, \bq_r^{(2)} \in \R^{M}$ and canonical variates $\bX_r^{(i)}\bq_r^{(i)}$, exactly as in Section~\ref{sec:cca}. Since the canonical directions act as linear detectors on the received snapshots, it is natural to ask how they relate to the classical detectors of~\eqref{eq:pilot_detectors}. To provide an answer, let
$\bH^{(i)} = [\bh_1^{(i)}, \ldots, \bh_R^{(i)}]$ collect the channels during the
$i$th half. Ibrahim, Karakasis and Sidiropoulos showed that the canonical directions are, up to scaling, the MMSE detectors of user $r$ for the two halves. {They state this result for complex baseband models, of which our real model is a special case.}
 
\begin{Theorem}[\cite{ibrahim2024mmse}]
\label{thm:blind_mmse}
Suppose the different user transmissions are uncorrelated, the symbol and
channel matrices of~\eqref{eq:repetition_views} have full column rank, and the
CCA problem is solved using ensemble correlation matrices. Then, the optimal
canonical vectors are scaled MMSE detectors,
\begin{linenomath}
\begin{equation}
  \bq_r^{(i)*}
  = \gamma_r\big(\bH^{(i)}\bH^{(i)\top} + \sigma^2\bI\big)^{-1}\bh_r^{(i)},
  \qquad \gamma_r \neq 0.
\end{equation}
\end{linenomath}
\end{Theorem}

In the noiseless case, the identifiability results of Section~\ref{sec:cca_id}
provide mild conditions under which CCA recovers the symbols of user $r$ up to a scalar. Theorem~\ref{thm:blind_mmse} shows that, in the presence of noise and
interference, CCA does more. Its canonical directions are the MMSE detectors themselves, and its canonical variates are therefore the linear MMSE estimates, up to a scaling that one pilot symbol resolves. It needs no training for channel estimation and no estimate of the noise power, and the result holds even when the channels change between the two views, since $\bH^{(1)}$ and $\bH^{(2)}$ are allowed to differ. When the channels are constant over both halves, $\bh_j^{(1)} = \bh_j^{(2)} = \bh_j$, the two canonical directions coincide and equal, up to scaling, the MMSE detector $\bw_r^{\mathrm{mmse}} = \bSigma^{-1}\bh_r$ of~\eqref{eq:pilot_detectors}~\cite{ibrahim2024mmse}. With sample rather than ensemble correlations, Theorem~\ref{thm:blind_mmse}
holds asymptotically in the block length~\cite{ibrahim2024mmse}. Empirically,
the estimate from a single view approaches the MMSE detector with perfect
channel knowledge as the block grows, and averaging the estimates from the two
views outperforms it in symbol error rate, owing to the combining gain of sending
the data twice~\cite{ibrahim2024mmse}.

In summary, repeating the data at the transmitter, and interleaving it when
several users repeat, provides the second view that the network provided in
Section~\ref{sec:sp:celledge}, now with a single multi-antenna receiver. On this
designed view, CCA computes the MMSE detector of each repeating user without channel
estimation, at the cost of sending each symbol twice. Next, we compare the
receivers of this section in terms of what each requires to be designed at the
transmitter and known at the receiver.

\subsection{Summary}
\label{sec:sp:synthesis}

Throughout this section, the problem was the factorization
$\bX = \bB\bH^\top + \bE$ of~\eqref{eq:signal_model}, and the goal was to
retrieve a subset of the columns of $\bB$, namely the symbol streams of the users
of interest. We began with the most favorable case, in which the directions of
the columns of $\bH$ are fixed by design. Such a case arises in CDMA with a
single antenna. There, all users follow the same protocol, and their spreading
signatures provide the observation dimensions that the missing antennas cannot,
at the cost of rate. For this case, the classical zero-forcing and MMSE
detectors are written in terms of the signatures of all users, since nulling an
interferer requires knowing its direction. From that point on, the receiver was
given less prior information, and the coordination among users became weaker.
Cell-edge detection motivated why retrieving only a subset of the columns of
$\bB$ is of interest, and it transferred the coordination from the users to the
infrastructure, namely two multi-antenna base stations and a central processing
unit. Section~\ref{sec:sp:mmse} carried the same idea to a single multi-antenna
receiver, where prior coordination is needed only between the receiver and the
users of interest, as the other users need not follow the same protocol. Apart
from this coordination, the price is a multi-antenna receiver and half the
transmission rate of the users of interest. In exchange, the receiver
asymptotically attains the linear MMSE detector knowing only the interleavers of
the desired user. Table~\ref{tab:sp_ladder} lays the sequence out.

The first two rows are the classical detectors for CDMA. In its classical form,
the zero-forcing detector needs the signatures of all users, and the MMSE
detector also needs the received amplitudes and the noise variance. Wang and
Poor~\cite{wangpoor1998} showed that the signal subspace, estimated from the
data, provides the directions of the interferers. As a result, both detectors
can be written in closed form from this subspace, with the desired user's
signature as the only prior knowledge, since the timing can then be estimated
from the data~\cite{strom1996, bensleyaazhang1996}. These receivers can detect
only users that spread their symbols with their own signatures, since the
detectors are built from these signatures. Each such user sends one symbol per
$M$ chips, so its rate drops to $1/M$. Other CDMA users must be assigned distinct signatures, but the CDMA users can also share the band with preexisting narrowband systems, since the linear MMSE detector suppresses a rich variety of narrowband
interference~\cite{poorwang1997part1, buzzi2002overlay}. With antennas in place of chips, the columns of $\bH$ are set by propagation rather than by design, so the designed knowledge must enter elsewhere. Pilots can take this role, as the designed knowledge from which the MMSE detector is built. However, when only the pilots of the desired user are known, strong interference degrades the channel estimate, and longer pilot sequences are required.

\begin{table}[H]
\caption{{Prior} 
 information required by each receiver, and its cost.
\label{tab:sp_ladder}}

\begin{adjustwidth}{-\extralength}{0cm}
\begin{tabularx}{\fulllength}{lCCC}
\toprule
\textbf{Receiver} & \textbf{Designed Knowledge} &
\textbf{Propagated Knowledge} & \textbf{Cost} \\
\midrule
\multicolumn{4}{l}{\emph{{CDMA, single antenna} 
}} \\\midrule
Zero-forcing~\cite{lupasverdu1989}
  & all signatures & none & detected users spread, rate $1/M$ \\\midrule
Linear MMSE, closed form~\cite{madhowhonig1994}
  & all signatures & received amplitudes, noise variance
  & detected users spread, rate $1/M$ \\\midrule
Subspace blind multiuser detection~\cite{wangpoor1998}
  & own signature & none & detected users spread, rate $1/M$ \\
\midrule
\multicolumn{4}{l}{\emph{{Multiple antennas}}} \\\midrule
Pilot-based MMSE (Section~\ref{sec:sp:pilots})
  & pilots & none & pilot overhead, rate $1 - N_p/N$ \\\midrule
Cell-edge CCA~\cite{ibrahim2020}
  & symbol alphabet for unmixing & none
  & two stations, central unit, full rate\\\midrule
Blind MMSE CCA~\cite{ibrahim2024mmse}
  & own interleavers if several users repeat & none
  & rate $1/2$ for the desired user only \\
\bottomrule
\end{tabularx}
\end{adjustwidth}
\end{table}
The last two rows of Table~\ref{tab:sp_ladder} weaken the pilot requirement, and
they pay for it differently. Cell-edge detection needs no designed knowledge to
find the common subspace, and requires in exchange a second station and a
central unit. However, going from the common subspace to the individual symbol streams still requires prior information. The pilots of the cell-edge users can enable this recovery, since the subspace is free of cell-center interference. In some cases, knowing the symbol alphabet alone suffices, up to the permutation and sign ambiguities that short identification sequences resolve.

CCA-based blind MMSE detection needs a single array, provided the desired user repeats its data and, if more than one user repeats, also interleaves it. The other users need not follow any protocol or coordinate with the desired user in any way, and their signals are not restricted to be narrowband, or even digital~\cite{ibrahim2024mmse}. The array only needs at least as many antennas as
users. The interleavers are the designed counterpart of the signature, as they enable the recovery of the desired user. Like the subspace detector of Wang and Poor, the CCA receiver computes the MMSE detector blindly, here from two views jointly, without having to estimate either a signal subspace or the noise power. 
The price is half the rate of the desired user and, when several users repeat,
interleavers known in advance to the receiver. In both cell-edge detection and blind MMSE detection, the second view is what replaces the pilots for channel estimation, provided once by the network and once by the code.

The two CCA receivers also come with different kinds of guarantees.
Theorem~\ref{thm:celledge_id} is an \emph{identifiability} statement. It ensures
that, in the absence of noise, CCA recovers the subspace of the cell-edge users.
Theorem~\ref{thm:blind_mmse} is an \emph{optimality} statement. It ensures that,
with ensemble correlations, the canonical directions coincide, up to scaling,
with the linear MMSE detector, which maximizes the output
signal-to-interference-plus-noise ratio among all linear detectors. Together,
the two results suggest that a multiview receiver with little prior knowledge
can do more than detect a potentially weak desired user. When the transmitter
creates a second view by design, such a receiver can also attain the best linear~\mbox{performance}.

Finally, the signal processing thread of this section arrives at the same
object that Section~\ref{sec:cca} reached from the statistical and
information-theoretic directions. There, CCA appeared as the maximum-likelihood
solution of a latent-factor model, its canonical correlations determined Wyner's
common information for jointly Gaussian views, and Section~\ref{sec:cca_id} stated its identifiability conditions for a general two-view generative model. Here,
CCA appears as a detector, reached from a low-rank data matrix and a question
about interference. As we saw, the same identifiability conditions re-emerge as
requirements on the numbers of antennas and the block length. The statistical,
information-theoretic, and signal processing threads thus meet at CCA, each
contributing its own reading of the common component.

Two assumptions ran through this section, and relaxing each leads to a problem
that the rest of this tutorial takes up. First, both routes to a second view
assume that the pairing of observations is known in advance. The two base
stations are fixed by the deployment, and the two halves of a repeated block are
fixed by the code. Ibrahim and Sidiropoulos name the case of more than two
stations as future work in the conclusion of~\cite{ibrahim2020}, and take it up
in~\cite{ibrahim2020selective}, where generalized CCA is applied to a selected
subset of stations. More views relax the identifiability conditions, as
Theorem~\ref{thm:gcca_id} shows. However, a realistic deployment has several groups of cell-edge users, each received by a different subset of stations, and some of
those stations may themselves be mobile, as with base stations mounted on
unmanned aerial vehicles. Which stations share which users is then unknown and
must be estimated jointly with the shared structure, a problem with close ties
to subspace clustering. Section~\ref{sec:open} briefly covers the work of Karakasis and Sidiropoulos~\cite{karakasis2026subspace} on unifying subspace clustering and generalized CCA.

Second, every model in this section, beginning with~\eqref{eq:signal_model},
takes the received signal to be a linear mixture of the transmitted symbols.
This is an idealization of the radio front ends. At each transmitter, the power
amplifier compresses as it approaches saturation~\cite{cripps_pa}, so a map $f$
acts on each entry of $\bB$. At the receiver, amplifiers saturate and
low-resolution converters clip and
quantize~\cite{mollen2017, schenk_rf_impairments}, so a map $g$ acts on each
entry of the noisy mixture, and~\eqref{eq:signal_model} becomes
\begin{linenomath}
\begin{equation}
  \bX = g\big(f(\bB)\,\bH^\top + \bE\big) \in \R^{N \times M}.
  \label{eq:signal_model_nl}
\end{equation}
\end{linenomath}
The nonlinearity of the transmitter alone keeps the mixture linear, now in the
distorted symbols $f(\bB)$, so a common subspace still exists. The
nonlinearity of the receiver, however, acts after mixing. The left singular vectors of $\bX$ then no longer span $\col(f(\bB))$, so what two views share is no longer a
common subspace of the received data. This motivates computing canonical
correlations and variates through nonlinear transformations instead of linear
ones. Section~\ref{sec:nonlinear} studies nonlinear CCA and asks what can still
be recovered, and under what conditions.
\section{Beyond Linear: Nonlinear and Deep CCA}
\label{sec:nonlinear}

To this point, we have restricted ourselves to linear transformations of the observed views. However, relaxing this assumption is essential in many real-world use cases, as illustrated next. 

We begin with the communication models of Section~\ref{sec:sp}. Every model there takes the received signal to be a linear mixture of the transmitted symbols, which assumes ideal radio front ends. As discussed in Section~\ref{sec:sp:synthesis}, practical front ends are nonlinear, which turns the model into~\eqref{eq:signal_model_nl}, and where the nonlinearity acts decides whether a common subspace survives. The transmitter map $f$ keeps the mixture linear in the distorted symbols $f(\bB)$. The receiver map $g$, which models a saturating amplifier or a low-resolution analog-to-digital converter that clips and quantizes, acts on each entry after the mixing. A linear mixture followed by such an entrywise map is known as a post-nonlinear mixture~\cite{taleb_pnl, yang2020nonlinear}. The common subspace that CCA recovered in Section~\ref{sec:sp} is then no longer a subspace of the received data, and linear CCA may no longer recover the cell-edge users.

A similar issue arises in neuroimaging. Consider two neighborhoods of voxels in the brain, each recorded over time, as in Section~\ref{sec:cca}, and suppose a common neural process drives both. These neighborhoods can be from the same person, or from two different subjects presented with the same stimulus. Here, the information recorded is not that process, but the blood-oxygen-level-dependent (BOLD) signal it induces. The map from neural activity to the BOLD signal passes through the hemodynamic response, i.e., the changes in blood flow and oxygenation that the activity triggers, and this map is nonlinear. The usual linear convolution model is only its first-order approximation~\cite{buxton1998, friston2000nonlinear}. Because these hemodynamic responses may vary from region to region~\cite{handwerker2004}, and even more across different subjects, the two neighborhoods project the shared process through two different nonlinear maps. As a result, each recording is a distinct nonlinear function of the common process together with activity private to that neighborhood. A linear analysis, such as the canonical correlations of Section~\ref{sec:cca}, may understate, or miss, the coupling that the shared process induces.

Finally, consider two views of one phenomenon from which we would like to extract the common factors. A natural instance arises in the study of speech production. One view can be a set of articulatory measurements, namely the positions of the tongue, the lips, and the soft palate as the vocal tract moves. Another can be a set of features extracted from a recording of the acoustic waveform. Figure~\ref{fig:speech} shows an example of the two views over the same time interval. Several factors are common to the two views and tie them together, among them the phoneme being produced and the speaker's identity and age. Each of these shapes both views. The phoneme is realized as a particular configuration of the vocal tract, and the speaker's anatomy, which varies with identity and age, sets the geometry of that vocal tract. Hence the phoneme, the identity, and the age are all imprinted on the articulatory measurements. The same factors shape the sound as well, since the phoneme and the speaker's vocal characteristics jointly determine the sound that is heard~\cite{petersonbarney1952}. Nevertheless, the map from the articulatory view to the acoustic view is nonlinear, and even many-to-one, since different vocal-tract shapes can produce nearly the same sound~\cite{atal1978}. Aligning the views through linear combinations may therefore recover little of what they share, and the common factors have to be extracted by nonlinear means.

\begin{figure}[H]
\begin{minipage}[c]{0.38\textwidth}
\centering
\includegraphics[width=\textwidth, trim={195pt 181pt 0 0}, clip]{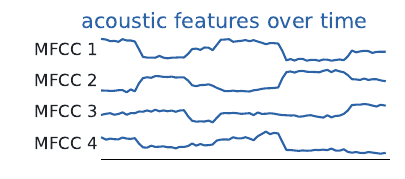}\\[-1mm]
{\footnotesize\hspace*{8mm}time}\\[1mm]
\text{(\textbf{a})}
\end{minipage}\hspace{1cm}%
\begin{minipage}[c]{0.38\textwidth}
\centering
\includegraphics[width=\textwidth, trim={195pt 22pt 0 110pt}, clip]{Figures/speech_setup.pdf}\\[-1mm]
{\footnotesize\hspace*{8mm}time}\\[1mm]
\text{(\textbf{b})}
\end{minipage}
\caption{The two views of the speech example over the same time interval. (\textbf{a}) The acoustic view, summarized by the first four mel-frequency cepstral coefficients (MFCCs), a standard description of the short-time spectrum of the sound. (\textbf{b}) The articulatory view, namely the trajectories of sensors on the tongue, the upper and lower lips, and the soft palate. The two views bear no obvious resemblance to each other, yet both change at the same moments, since the phoneme being produced shapes both.}
\label{fig:speech}
\end{figure}

There are two reasons for considering CCA with nonlinear transformations. The first is to measure how strongly the views are correlated. The canonical correlations answer this question, and the measure they provide is useful in its own right. In the fMRI example, a measure of correlation between brain regions can help improve our understanding of how the brain functions. The second is the recovery of the common latent factors. In the nonlinear communications example, CCA with nonlinear transformations, together with a few known symbols to resolve the remaining ambiguities, could recover the weak cell-edge users. More broadly, the recovered common factors can be used for a downstream task. The benefit is clearest in a particular regime, when samples from jointly observed high-dimensional views are plentiful, the low-dimensional common factors are known to carry the information a task needs, but the labeled data at hand are scarce. There, requiring the learned functions of the two views to agree on the unlabeled samples restricts the set of candidate pairs of functions, and with it the number of labeled examples needed to find a pair that performs the task well~\cite{sridharan2008}.

Both of these reasons rest on the same premise, which Section~\ref{sec:cca} made explicit for linear maps. Behind each collection of views lies a common factor that drives all of them, as in the probabilistic model of Bach and Jordan~\cite{bach2005} and the common-plus-private model~\eqref{eq:gen_two_view}. Here, each view is formed from the common factor and a private factor of its own through its own, possibly nonlinear, function. This is the modeling assumption we follow in the rest of this section. It is motivated by examples like those above, where in some cases a common factor is easy to name, though extracting it is not trivial. A more general problem is to recover both the common and the private factors, a task known as latent disentanglement. Although that is an active research topic, in this tutorial we focus mostly on what it takes to recover the common factor.  Under this model, the task of this section is to rigorously recover the common factor, or at least to reliably gauge how strongly the views express it.

The natural way to look for the common factor is to find view-wise transformations that make the views maximally correlated. When the transformations are linear, this is Hotelling's canonical correlation. When the views are single variables and the transformations are arbitrary functions, normalized to have zero mean and unit variance outputs, we obtain the Hirschfeld--Gebelein--R\'enyi (HGR) maximal correlation of Section~\ref{sec:common_info}. This section works at their meeting point, maximizing correlation over nonlinear transformations of multi-variable views.

Maximizing correlation is an objective, not a guarantee, and the gap between the two is what this section is really about. Two challenges define that gap. First,
correlation alone is too permissive. Once arbitrary transformations are allowed,
correlation can be driven high by degenerate or overfitted choices that agree on
nothing meaningful. Hence, the transformed views may agree without reflecting the common factor we assumed. The normalization constraints of HGR, zero-mean and unit-variance outputs, help but do not fully resolve this. Recall from Section~\ref{sec:common_info} that a pair of functions attaining $S(X,Y)=1$ certifies that the G\'acs--K\"orner common information is positive, but need not recover the common part $W^*$, only a possibly coarser function $c(W^*)$ of it. This observation brings us to the second challenge. Even when the transformations agree on a common variable, the variable may be any coarsening $c(W^*)$ of the common part, from a coarse grouping of its values up to $W^*$ itself, and these coarsenings carry different amounts of common information. It is not clear which of them the objective singles out. Finally, learning the transformations from finite samples sharpens both difficulties.

The rest of this section is organized as follows. Section~\ref{sec:optimal-scaling} presents optimal scaling, an early form of nonlinear CCA in which each variable of a view is transformed separately. Section~\ref{sec:kcca} presents kernel CCA, which maps each view into a space of functions where ordinary CCA applies, and Section~\ref{sec:kcca_reg} shows why that approach must be regularized. Sections~\ref{sec:bridge} and~\ref{sec:dcca} replace that fixed map with a neural network learned from data, known as deep CCA, and show that the networks can agree on the training samples while carrying nothing the views share. Section~\ref{sec:vcca} turns to formulations that fit a probabilistic model of the views rather than maximize a correlation. Section~\ref{sec:idgcca} gives two sets of conditions under which the common factor is provably recovered, and Section~\ref{sec:ssl} relates these methods to self-supervised learning, where the views are constructed from one sample rather than observed. Finally, Section~\ref{sec:nonlinear:summary} summarizes the section.

\subsection{Optimal Scaling: The Pre-Neural Nonlinear CCA}
\label{sec:optimal-scaling}

The first step beyond linear CCA was a half-step, taken through a method known as \emph{optimal scaling}. Optimal scaling is a general method for handling variables that are categorical or recorded on an unknown scale. Each observed variable of a view is replaced by a transformed version, where each transformation is drawn from a class fixed by the type of the variable and chosen so as to maximize the correlation. All the transformations are estimated jointly with the canonical weights. Three variable types are \mbox{distinguished~\cite{stevens1946, vanderburg1994}}. A numerical variable takes values on a numerical scale, such as a temperature. An ordinal variable takes categories with a natural order but no meaningful distances between them, such as low, medium and high. A nominal variable takes categories with no order at all, such as eye color. Accordingly, a numerical variable admits any affine transformation, an ordinal variable any monotone nondecreasing one, and a nominal variable any assignment of values to its categories. The idea goes back to Fisher, who assigned scores to the categories of a contingency table so that the two sets of scores were maximally correlated~\cite{fisher1940}, and psychometrics developed it into the general method described here~\cite{youngdeleeuwtakane1976}.

Placing the transformed variables inside the MAXVAR consensus objective
\eqref{eq:maxvar2} gives a nonlinear version of CCA, proposed for two views by
Young, de Leeuw and Takane as CORALS \cite{youngdeleeuwtakane1976} and
extended to $K$ views by van der Burg, de Leeuw and Verdegaal as OVERALS
\cite{vanderburg1988}. A variant called CANALS~\cite{vanderburgdeleeuw1983} applies optimal scaling to two views but maximizes the sum of the squared canonical correlations rather than their sum, and coincides with CORALS for a single component. A direct consequence of the admissible classes above is that optimal scaling extends CCA beyond the linear case only through ordinal and nominal variables, since with all variables numerical the method reduces to ordinary CCA \cite{vanderburg1988}.

Two features matter for what follows. First, each view-encoder is still
additive, in the~sense
\begin{equation}
\mathbf{f}_k(\mathbf{x}_k) \;=\; \sum_{t=1}^{M_k}
\mathbf{A}_k[t,:]^\top \, \varphi_{k,t}(x_{k,t}) ,
\label{eq:additive-encoder}
\end{equation}
where $x_{k,t}$ is the $t$-th coordinate of $\mathbf{x}_k$, the scalar map
$\varphi_{k,t}$ is the transformation assigned to that coordinate, and
$\mathbf{A}_k[t,:]^\top \in \mathbb{R}^{D}$ is the corresponding row of reduction matrix \linebreak $\mathbf{A}_k\in\mathbb{R}^{M_k \times D}$. Hence, each coordinate is distorted on its own and the mixing
that follows is linear. This is the reverse of the post-nonlinear mixtures
\cite{taleb_pnl, yang2020nonlinear} discussed in the introduction of this
section. Writing both as forward models,
\begin{equation}
\underbrace{\mathbf{y} = \boldsymbol{\theta}\big(\mathbf{H}\mathbf{s}\big)}_{\text{post-nonlinear}}
\qquad\text{versus}\qquad
\underbrace{\mathbf{y} = \mathbf{H}\,\boldsymbol{\varphi}(\mathbf{s})}_{\text{pre-nonlinear}} ,
\end{equation}
with $\boldsymbol{\theta}$ and $\boldsymbol{\varphi}$ being vector valued
functions acting entrywise. Therefore, the nonlinearity comes after the mixing
in the first and before it in the second. The front-end model~\eqref{eq:signal_model_nl} of Section~\ref{sec:sp:synthesis} contains both cases, since each row of $\bX$ is a snapshot of the form above. With the transmitter map $f$ taken as the identity, \eqref{eq:signal_model_nl} is a post-nonlinear mixture, and with the receiver map $g$ taken as the identity, it is a pre-nonlinear one. However, neither allows for general nonlinear transformations.

Second, the consensus objectives reward agreement between the transformed views, and the transformations are learned. Nothing in the problem formulations distinguishes agreement produced by the common factor from agreement produced by the transformations themselves. The admissible classes make the latter available, since a nominal transformation assigns one free value per category and a monotone one is required only to be nondecreasing, so both can merge distinct values. Young, de Leeuw and Takane warn that under extreme nonlinearities and high noise their algorithm may yield degenerate results~\cite{youngdeleeuwtakane1976}. This is the first appearance of a difficulty that concerns every method in this section.
Sections~\ref{sec:kcca} and~\ref{sec:kcca_reg} are where the transformation class becomes rich enough for the degeneracy to be written down in closed form, and Section~\ref{sec:dcca} is where it is met with no restriction at all on the learned transformations.
 

\subsection{Kernel CCA}
\label{sec:kcca}

Optimal scaling relaxed linearity under two restrictions that are worth separating. Its encoders in~\eqref{eq:additive-encoder} nonlinearly transform each coordinate of a view separately and then mix the transformed coordinates linearly. Hence, no nonlinear transformation may act on several coordinates at once. The second restriction is that each coordinate transformation, although estimated from data jointly with the linear reduction, is confined to a restricted class fixed in
advance by the type of the variable. When every variable is numerical the
admissible class is affine and optimal scaling reduces to ordinary
CCA~\cite{vanderburg1988}. Hence, nothing is gained on precisely the continuous
views that motivated this section, the voxel time series, the received signals, and the acoustic features.

Kernel CCA lifts the first restriction and relaxes the second. For each view, a single transformation is learned and applied to the view as a whole rather than coordinate by coordinate. The view is first mapped into a space of very high, possibly infinite, dimension, in which the transformation becomes a linear operation. The class of candidate transformations is then a reproducing kernel Hilbert space (RKHS), fixed by the choice of a kernel for each view, and it is far larger than the classes of optimal scaling. The price is that, once the class is this rich, agreement between the transformed views can be produced by the transformations alone. This is the difficulty already met in Section~\ref{sec:optimal-scaling}, where it was an empirical hazard under extreme nonlinearity and high noise. In the kernel setting it can be detected, as we will see in Section~\ref{sec:kcca_reg}, and for the kernels in common use it is the generic outcome rather than an edge case.

The construction begins with a kernel, a symmetric function $\kappa_k(\bx,\bx')=\kappa_k(\bx',\bx)$ that measures the similarity between two points of view $k$. The kernel is positive definite when, for any finite set of points $\bx^{(1)},\ldots,\bx^{(N)}$, the $N\times N$ matrix with entries $\kappa_k(\bx^{(m)},\bx^{(n)})$ is positive semidefinite. Fixing such a kernel defines a space of functions along with it, and fixing one of the two arguments of the kernel gives a function in that space. Both arguments take values in $\mathcal{X}_k$, the space in which view $k$ takes its values, which for the continuous views of this section can be $\R^{M_k}$. A Hilbert space $\mathcal{H}_k$ of real-valued functions, defined on $\mathcal{X}_k$, is a \emph{reproducing kernel Hilbert space} (RKHS) with kernel $\kappa_k$ when
$\kappa_k(\bx,\cdot) \in \mathcal{H}_k$ for every $\bx \in \mathcal{X}_k$, and
\begin{linenomath}
\begin{equation}
  f_k(\bx) \;=\; \big\langle f_k,\, \kappa_k(\bx,\cdot)
  \big\rangle_{\mathcal{H}_k}
  \qquad \text{for all } f_k \in \mathcal{H}_k,\; \bx \in \mathcal{X}_k,
  \label{eq:reproducing}
\end{equation}
\end{linenomath}
where $\big\langle \cdot, \cdot\big\rangle_{\mathcal{H}_k}$ is the inner product on $\mathcal{H}_k$. By the theorem of Moore and Aronszajn \cite{aronszajn1950}, such a space exists and is unique for every positive definite kernel.

Fixing one argument maps an arbitrary point $\bx'$ to the function $f_k(\cdot)=\kappa_k(\bx',\cdot)$ in $\mathcal{H}_k$, a representation of $\bx'$ that may be infinite dimensional. Evaluating this function at another point $\bx$ gives $f_k(\bx)=\kappa_k(\bx',\bx)$. At the same time, following~\eqref{eq:reproducing} we obtain that $f_k(\bx)$ is also equal to the inner product between the representations of $\bx'$ and $\bx$, i.e., \linebreak $f_k(\bx)=\langle \kappa_k(\bx',\cdot),\kappa_k(\bx,\cdot)\rangle_{\mathcal{H}_k}$.
Hence, every positive definite kernel $\kappa_k(\bx',\bx)$ equals an inner product between the images of its two arguments under the feature map \linebreak $\phi_k(\bx) = \kappa_k(\bx,\cdot)$, and computes that inner product without forming the images. Inner products between mapped points, and evaluations of functions built from them, thus reduce to kernel evaluations. This is what keeps the approach practical even when $\mathcal{H}_k$ is infinite dimensional. Note also that the mapped training samples sit in $\mathcal{H}_k$ alongside the candidate transformations themselves, so the learned transformation can be expanded in terms of the representations of the training samples in $\mathcal{H}_k$.

Kernel CCA now poses the following problem over $\mathcal{H}_1$ and
$\mathcal{H}_2$,
\begin{linenomath}
\begin{equation}
  \max_{f_1 \in \mathcal{H}_1,\, f_2 \in \mathcal{H}_2}\;
  \frac{\Cov\big(f_1(\bx_1),\, f_2(\bx_2)\big)}
       {\sqrt{\Var[f_1(\bx_1)]}\;\sqrt{\Var[f_2(\bx_2)]}}\, ,
  \label{eq:kcca_pop}
\end{equation}
\end{linenomath}
and the value it attains is the first kernel canonical correlation. No constraints appear, because the ratio is unchanged when either function is rescaled by a positive constant or shifted by a constant, covariance and variance being central moments. The same ratio has appeared twice before, and only the freedom given to $f_1$ and $f_2$ distinguishes the three problems. In Hotelling's
problem~\eqref{eq:cca_corr} the two functions are linear functionals of the views. In the Hirschfeld--Gebelein--R\'enyi maximal correlation of Section~\ref{sec:gk} they range over all measurable functions, a class larger than any RKHS, with the normalization to zero mean and unit variance stated as constraints rather than absorbed into a ratio. 

The kernel is what determines how much of that freedom is granted. A linear kernel makes each feature map $\phi_k$ the identity and reduces~\eqref{eq:kcca_pop}  to Hotelling's problem~\eqref{eq:cca_corr}~\cite{kussgraepel2003}. A larger RKHS, on the other hand, cannot decrease the attainable value, being a maximum over a larger
feasible set. For translation-invariant kernels, which depend on the two points only through their difference, such as the Gaussian kernel $\kappa_k(\bx,\bx') = \exp\big(-\|\bx-\bx'\|^2/(2\sigma^2)\big)$ with bandwidth $\sigma$, Bach and Jordan report that the first kernel canonical correlation increases as the bandwidth shrinks~\cite{bachjordan2002}. Whether a value
obtained this way can reliably say something about the two views is the question
Section~\ref{sec:kcca_reg} takes up.

Given $N$ paired samples, $\left\{\left(\bx_1^{(n)}, \bx_2^{(n)}\right)\right\}_{n=1}^N$, write $\tilde\bK_k\in\mathbb{R}^{N\times N}$ for the Gram matrix with
entries $[\tilde\bK_k]_{mn} = \kappa_k(\bx_k^{(m)}, \bx_k^{(n)})$ and
$\bK_k = \bJ\,\tilde\bK_k\,\bJ$ for its centered version, where
\linebreak $\bJ = \bI_N - \tfrac{1}{N}\mathbf{1}\mathbf{1}^\top$ and $\mathbf{1}$ is the $N$-dimensional vector of ones. Centering in feature space is what implements the zero-mean requirement, and it can be carried out on the Gram matrix alone, without access to $\phi_k$. From here on, a mapped sample is understood to be centered, so every Gram matrix below is $\bK_k$ rather than $\tilde\bK_k$.

Split any candidate transformation as $f_k = f_k^{\parallel} + f_k^{\perp}$, with
$f_k^{\parallel}$ in the span of the mapped training samples and $f_k^{\perp}$
orthogonal to it. By the reproducing property~\eqref{eq:reproducing}, we get
\begin{linenomath}
\begin{equation}
  f_k\big(\bx_k^{(n)}\big)
  = \big\langle f_k,\, \phi_k(\bx_k^{(n)})\big\rangle_{\mathcal{H}_k}
  = \big\langle f_k^{\parallel},\, \phi_k(\bx_k^{(n)})\big\rangle_{\mathcal{H}_k}
  = f_k^{\parallel}\big(\bx_k^{(n)}\big),
  \label{eq:orthogonal_invisible}
\end{equation}
\end{linenomath}
since $\phi_k(\bx_k^{(n)})$ lies in the span and $f_k^{\perp}$ is orthogonal
to it. The component $f_k^{\perp}$ is therefore invisible to every empirical quantity computed from the samples. Since the empirical objective value depends on
$f_k$ only through these evaluations, any candidate and its projection
$f_k^{\parallel}$ attain the same value, so nothing is lost by restricting our
attention to
\begin{linenomath}
\begin{equation}
  f_k(\cdot) \;=\; \sum_{n=1}^{N} \alpha_{k,n}\,
  \phi_k\big(\bx_k^{(n)}\big)(\cdot),
  \qquad \balpha_k \in \R^{N}.
  \label{eq:representer}
\end{equation}
\end{linenomath}
This is the expansion in terms of the training samples noted above, and it reduces the search over $\mathcal{H}_k$ to a search over $N$ coefficients. Since the objective cannot see $f_k^{\perp}$, nothing forces it to vanish, so the optimal $f_k$ is not unique in $\mathcal{H}_k$. Within the span, the coefficients $\balpha_k$ are not necessarily unique either, since different $\balpha_k$ give the same function whenever the Gram matrix is singular~\cite{kussgraepel2003}. The freedom in $f_k^{\perp}$ disappears once the squared norm $\|f_k\|^2_{\mathcal{H}_k}$ is added to the variance of $f_k$ in the denominator of~\eqref{eq:kcca_pop}, as is done in practice. An orthogonal component then increases the denominator while leaving the numerator unchanged. Hence, in the penalized case, every optimal $f_k$ has $f_k^{\perp} = 0$.

Three quantities are needed in what follows, the empirical counterparts of the two moments in~\eqref{eq:kcca_pop} and the RKHS norm used by the penalty introduced above. The first two are averages over the training samples and can be computed simply. Substituting~\eqref{eq:representer} into the reproducing property~\eqref{eq:reproducing}, and then using linearity of the inner product and the fact that the inner products of the centered mapped samples are the entries of $\bK_k$, gives the evaluation of a centered candidate transformation at the $n$-th training sample,
\begin{linenomath}
\begin{equation}
  {f}_k\big(\bx_k^{(n)}\big)
  = \Big\langle \sum_{m=1}^{N} \alpha_{k,m}\,{\phi}_k\big(\bx_k^{(m)}\big),\;
    {\phi}_k\big(\bx_k^{(n)}\big)\Big\rangle_{\mathcal{H}_k}
  = \sum_{m=1}^{N} \alpha_{k,m}\,
       \big[\bK_k\big]_{mn}
  = \big[{\bK}_k\balpha_k\big]_n.
  \label{eq:eval_vector}
\end{equation}
\end{linenomath}
Hence, the vector of evaluations on the training samples is ${\bK}_k\balpha_k$, and each empirical moment is an average of products of entries of two such vectors. For the cross-covariance,
\vspace{-6pt}
\begin{adjustwidth}{-\extralength}{0cm}
\begin{linenomath}
\begin{equation}
  \widehat{\Cov}\big({f}_1(\bx_1), {f}_2(\bx_2)\big)
  = \frac{1}{N}\sum_{n=1}^{N}
    {f}_1\big(\bx_1^{(n)}\big)\, {f}_2\big(\bx_2^{(n)}\big)
  = \frac{1}{N} \big(\bK_1\balpha_1\big)^\top \big(\bK_2\balpha_2\big)
  = \frac{1}{N}\,\balpha_1^\top \bK_1\bK_2\, \balpha_2 ,
  \label{eq:kcca_cov}
\end{equation}
\end{linenomath}
\end{adjustwidth}
where in the last step we used the symmetry of $\bK_1$. Similarly, for each view we can compute the empirical variance as,
\begin{linenomath}
\begin{equation}
  \widehat{\Var}\big(f_k\big(\bx_k\big)\big)
  = \frac{1}{N}\sum_{n=1}^{N}  f_k\big(\bx_k^{(n)}\big)^2
  = \frac{1}{N} \big(\bK_k\balpha_k\big)^\top \big(\bK_k\balpha_k\big)
  = \frac{1}{N}\,\balpha_k^\top \bK_k^2\, \balpha_k .
  \label{eq:kcca_var}
\end{equation}
\end{linenomath}

The RKHS norm is obtained differently, since it is not an average over the training samples but an inner product in $\mathcal{H}_k$. Expanding both arguments through~\eqref{eq:representer} and using that the inner products of the centered mapped samples are the entries of $\bK_k$ leaves the double sum,
\vspace{-6pt}
\begin{adjustwidth}{-\extralength}{0cm}
\begin{linenomath}
\begin{equation}
  \|f_k\|_{\mathcal{H}_k}^2
  = \Big\langle \sum_{m=1}^{N} \alpha_{k,m}\,\phi_k\big(\bx_k^{(m)}\big),\;
    \sum_{n=1}^{N} \alpha_{k,n}\,\phi_k\big(\bx_k^{(n)}\big)
    \Big\rangle_{\mathcal{H}_k}
   = \sum_{m=1}^{N}\sum_{n=1}^{N} \alpha_{k,m}\,\alpha_{k,n}\,
    \big[\bK_k\big]_{mn}
  = \balpha_k^\top \bK_k\, \balpha_k .
  \label{eq:kcca_norm}
\end{equation}
\end{linenomath}
\end{adjustwidth}
The difference between the two derivations is worth noting. Each empirical
moment applies the Gram matrix twice, once for each evaluation vector, while
the norm applies it once. Note also that
the moments carry a factor $1/N$ and the norm does not, since only the former
average over the training samples.

Substituting the first two into~\eqref{eq:kcca_pop} and using the scale invariance of the ratio, exactly as in going from~\eqref{eq:cca_corr} to~\eqref{eq:cca_opt}, gives the empirical problem
\begin{linenomath}
\begin{equation}
  \max_{\balpha_1, \balpha_2}\;
  \balpha_1^\top \bK_1\bK_2\, \balpha_2
  \quad \text{s.t.} \quad
  \balpha_k^\top \bK_k^2\, \balpha_k = 1, \;\; k = 1,2,
  \label{eq:kcca_plain}
\end{equation}
\end{linenomath}
where the factors of $1/N$ have been dropped throughout, since a positive
multiplicative constant on the objective does not move the maximizer and one
in the constraint only rescales $\balpha_k$, which the same scale invariance
absorbs. Problem~\eqref{eq:kcca_plain} is the direct kernelization, and it is what several of the early formulations of kernel CCA solve~\cite{melzer2001, kussgraepel2003}. It is also, as Section~\ref{sec:kcca_reg} shows, degenerate for commonly used kernels.

The common safeguard adopted across the literature is to penalize the RKHS norms
of $f_1$ and $f_2$ alongside their variances, which bounds the size of the
candidate functions in $\mathcal{H}_k$, and with it the class from which
agreement can be manufactured. Carrying that penalty, with weight $\lambda_k \geq 0$ for view $k$, through~\eqref{eq:kcca_norm} turns the constraints of~\eqref{eq:kcca_plain} into a ridge on the Gram matrices,
and norm-penalized kernel CCA reads\vspace{12pt} 
\begin{linenomath}
\begin{equation}
  \max_{\balpha_1, \balpha_2}\;
  \balpha_1^\top \bK_1\bK_2\, \balpha_2
  \quad \text{s.t.} \quad
  \balpha_k^\top \big(\bK_k^2 + \lambda_k \bK_k\big)\balpha_k = 1,
  \;\; k = 1,2.
  \label{eq:kcca_emp}
\end{equation}
\end{linenomath}
{{Implementations} differ in how the penalty is carried into the constraint. Equation~\eqref{eq:kcca_emp} follows Akaho~\cite{akaho2001} and Hardoon {et~al.}~\cite{hardoon2004} in penalizing $\|f_k\|_{\mathcal{H}_k}^2$, which by~\eqref{eq:kcca_norm} equals $\balpha_k^\top\bK_k\balpha_k$ and therefore adds $\lambda_k\bK_k$ to the constraint matrix. Akaho also notes that penalizing the coefficient norms $\|\balpha_k\|^2$ instead, which contributes $\lambda_k\bI$, carries the same complexity, and that he found no significant difference in performance in his experiments with synthetic data.} The stationarity conditions of~\eqref{eq:kcca_emp} give the generalized
eigenvalue problem
\begin{linenomath}
\begin{equation}
  \begin{bmatrix} \mathbf{0} & \bK_1\bK_2 \\ \bK_2\bK_1 & \mathbf{0} \end{bmatrix}
  \begin{bmatrix} \balpha_1 \\ \balpha_2 \end{bmatrix}
  = \sigma
  \begin{bmatrix} \bK_1^2 + \lambda_1\bK_1 & \mathbf{0} \\
                  \mathbf{0} & \bK_2^2 + \lambda_2\bK_2 \end{bmatrix}
  \begin{bmatrix} \balpha_1 \\ \balpha_2 \end{bmatrix} ,
  \label{eq:kcca_gev}
\end{equation}
\end{linenomath}
whose leading generalized eigenvalue is the first penalized kernel canonical correlation, and at $\lambda_1 = \lambda_2 = 0$ the first kernel canonical correlation of~\eqref{eq:kcca_plain}. Further components are obtained from the subsequent eigenvectors, as in Section~\ref{sec:cca}, and the same applies to the unpenalized problem~\eqref{eq:kcca_plain}.

\subsubsection*{Attribution}

Several groups reached kernel CCA at once~\cite{kussgraepel2003}, and the first version is not the one in use today. Lai and Fyfe~\cite{laifyfe2000} published the earliest of the versions discussed here, alongside a nonlinear neural implementation of CCA, but their construction is not the one above. They form the cross-block entrywise as $\kappa(\bx_1^{(i)}, \bx_2^{(j)})$, with one argument drawn from each view, rather than from the product $\bK_1\bK_2$ of two within-view Gram matrices. The kernel must therefore accept a point of either view, so the two views have to live in the same space and have the same number of features. Their
construction does not produce a feature-space cross-covariance, which is why
the numbers it returns need not be correlations. Finally, no penalized variants were considered in their work. Their own discussion records both consequences, i.e., that the method is prone to singularities and that it can return correlations exceeding one.

Akaho~\cite{akaho2001}, working independently, gave the formulation used here,
with a separate kernel for each view and the penalized constraint
of~\eqref{eq:kcca_emp}. Melzer, Reiter and Bischof~\cite{melzer2001} applied the unpenalized form~\eqref{eq:kcca_plain} to feature extraction for pose estimation.
Van Gestel {et~al.}~\cite{vangestel2001} consider the supervised case,
using CCA between the feature map of the inputs and a vector of targets. They
show that for a scalar target the regularized problem reduces to least-squares
support vector machine regression. Bach and Jordan~\cite{bachjordan2002} arrived at the same object from independent component analysis (ICA). Starting from two random variables, they use a feature map to expand each into a random vector, and then consider linear CCA in the feature space. The maximal correlation over the RKHS that this process returns is what they minimize as a measure of dependence, and is the basis of the kernel ICA contrast functions they propose.  Hardoon, Szedmak and Shawe-Taylor~\cite{hardoon2004} gave an overview of the method aimed at the learning community, together with an application to retrieving images from text queries.

Of the works discussed, those that carry a penalty~\cite{akaho2001,
vangestel2001, bachjordan2002, hardoon2004} bound the class and those that do
not~\cite{laifyfe2000, melzer2001} leave it unbounded. The next subsection is
about why that distinction decides whether the reported number means anything.


\subsection{Why Regularization Is Not Optional}
\label{sec:kcca_reg}

Norm penalization in~\eqref{eq:kcca_emp} ($\lambda_k > 0$) is not a numerical convenience. Setting $\lambda_k = 0$ returns the unpenalized problem~\eqref{eq:kcca_plain}, in which nothing bounds the size of the learned transformations. The correlations obtained can then reflect that freedom rather than the common factor. This is the degeneracy mentioned throughout the section. What the kernel setting adds is that here the degeneracy can be stated exactly, since only $N$ training samples are mapped. Hence, the span of the mapped training samples of each view, which Kuss and Graepel call the effective feature space~\cite{kussgraepel2003}, has dimension at most $N$, regardless of how large $\mathcal{H}_k$ is, and the failure can be read off from the dimensions of the two effective feature spaces.

\subsubsection{The Degeneracy Is a Dimension Count}
Kuss and Graepel~\cite{kussgraepel2003} give the cleanest account of why degeneracy can happen.
Collect the representations of all training samples as
$\bPhi_k = [\phi_k(\bx_k^{(1)}),\ldots,\phi_k(\bx_k^{(N)})]^\top$, so that the
$n$-th row of $\bPhi_k$ is the representation of the $n$-th training sample $\phi_k(\bx_k^{(n)})$.
Since these are centered, as agreed in Section~\ref{sec:kcca}, we have
$\bK_k = \bPhi_k\bPhi_k^\top$. Written this way, $\bPhi_k$ acts in two~directions.

The expansion~\eqref{eq:representer} reads $f_k = \bPhi_k^\top\balpha_k$, so
$\bPhi_k^\top$ maps a coefficient vector $\balpha_k \in \R^N$ to a candidate
transformation in $\mathcal{H}_k$. Evaluating that candidate on the $N$
training samples gives a vector of values, which by~\eqref{eq:eval_vector} is
\begin{linenomath}
\begin{equation}
  \bz_k \;=\; \bPhi_k f_k \;=\; \bPhi_k\bPhi_k^\top\balpha_k
  \;=\; \bK_k\balpha_k .
  \label{eq:variate_map}
\end{equation}
\end{linenomath}
Hence $\bPhi_k$ maps a transformation in $\mathcal{H}_k$ back to a vector in $\R^N$. Two subspaces follow, one at each end. The candidate transformations fill $\col(\bPhi_k^\top) \subseteq \mathcal{H}_k$, the effective feature space introduced above. Their variates fill $\col(\bK_k) \subseteq \R^N$. Both have dimension $\text{rank}(\bK_k)$, which is at most $N$, since $\bK_k$ is $N \times N$. Therefore, regardless of how large $\mathcal{H}_k$ is, the candidate transformations span at most $N$ dimensions, and Kuss and Graepel identify their span with $\col(\bK_k)$, the space of their variates~\cite{kussgraepel2003}. Note that $\mathcal{H}_1$ and $\mathcal{H}_2$ are different spaces, while the variates of both views live in the same $\R^N$. Comparing the two views is therefore a question about subspaces of $\R^N$.

The geometry of Section~\ref{sec:cca} applies to the unpenalized problem~\eqref{eq:kcca_plain}, with $\bX_k$ replaced by $\bK_k$. The canonical correlations $\sigma_i$ are the cosines of the principal angles between $\col(\bK_1)$ and $\col(\bK_2)$. The solution can also be obtained by applying linear CCA to two ordinary data matrices. Let $\bK_k = \bV_k\bLambda_k\bV_k^\top$ be the eigendecomposition of the Gram matrix, and keep the $d_k = \text{rank}(\bK_k)$ nonzero columns of $\bV_k\bLambda_k^{1/2}$ in $\bC_k \in \R^{N \times d_k}$. Row $n$ of $\bC_k$ holds $d_k$ numbers describing the $n$-th training sample of view $k$, and these are its kernel principal component scores. By construction $\bC_k\bC_k^\top = \bK_k$ and $\col(\bC_k) = \col(\bK_k)$, so applying the linear CCA of Section~\ref{sec:cca} to $\bC_1$ and $\bC_2$ returns the kernel canonical correlations. Note that $d_k$ is not a modeling choice. It equals $\dim \col(\bK_k)$, so the chosen kernels and the training samples determine it. A linear kernel gives $d_k = \text{rank}(\bX_k) \le \min(N, M_k)$. The Gaussian kernel is strictly positive definite~\cite{micchelli1986}, i.e., its Gram matrix on any set of distinct points is positive definite, so with distinct training samples the uncentered Gram matrix is nonsingular. Centering can be performed by multiplying $\tilde\bK_k$ by $\bJ$ on both sides, and $\bJ$ projects onto $\mathbf{1}^{\perp}$. Hence exactly one direction is lost, $\col(\bK_k)$ is the whole of $\mathbf{1}^{\perp}$, and $d_k = N-1$.

To see what can go wrong, suppose $\col(\bK_1)$ and $\col(\bK_2)$ intersect in a subspace of dimension $h$. Then, the first $h$ canonical variates coincide and
$\sigma_i = 1$ for $i \le h$. Together the two subspaces span
$\col([\bK_1\ \bK_2])$, whose dimension is $d_1 + d_2 - h$. That column space
sits in $\R^N$, so its dimension is at most $N$. Hence
$h \ge \max\{d_1 + d_2 - N, 0\}$, and the intersection contains a nonzero vector whenever
\begin{linenomath}
\begin{equation}
  d_1 + d_2 \;>\; N .
  \label{eq:kcca_collapse}
\end{equation}
\end{linenomath}

Condition~\eqref{eq:kcca_collapse} is easily met. A Gaussian kernel gives $d_1 = d_2 = N-1$, so it holds for every $N > 2$. In fact, $\col(\bK_1)$ and $\col(\bK_2)$ are then both the whole of $\mathbf{1}^{\perp}$, hence identical. Everything belongs in the intersection, and therefore every canonical correlation equals one regardless of the training set. A similar budget appeared in the linear case. Theorem~\ref{thm:cca_id} of Ibrahim and Sidiropoulos asked the common and private subspaces to fit inside $\R^N$, i.e., $N \ge D + L_1 + L_2$, which is what leaves the intersection of the two column spaces meaningful and equal to the modeled $\col(\bG)$. The feature map fills $\R^N$ in a way that generates fictitious subspace overlaps and hence spurious common subspaces. Kuss and Graepel conclude that the kernel canonical correlations must be interpreted with caution, and that kernel CCA is better regarded as a geometric algorithm for constructing highly correlated features~\cite{kussgraepel2003}.

As a last note, recall from Section~\ref{sec:kcca} that Bach and Jordan report the first kernel canonical correlation to increase as the bandwidth of the Gaussian kernel shrinks. That is a statement about the value attained, not about $d_k$. 
The uncentered Gram matrix $\tilde\bK_k$ is nonsingular for every $\sigma > 0$, so $d_k = N-1$ regardless of the chosen bandwidth, and~\eqref{eq:kcca_collapse} holds either way. What $\sigma$ changes is the spectrum of $\bK_k$, and that is what the penalty acts on. As $\sigma \to 0$, the off-diagonal entries of $\tilde\bK_k$ vanish while its diagonal entries stay at one, so $\tilde\bK_k \to \bI$ and $\bK_k = \bJ\tilde\bK_k\bJ \to \bJ$ for both views. Every vector $\bz \in \mathbf{1}^{\perp}$ is then a variate of both views, and with $\bK_k\balpha_k = \bz$ and $\lambda_1 = \lambda_2 = \lambda$, the penalized problem~\eqref{eq:kcca_emp} attains a penalized kernel canonical correlation of $1/(1+\lambda)$ for every such $\bz$, regardless of the data. 

As $\sigma \to \infty$, a first-order expansion of the exponential gives $\bK_k \approx \bar\bX_k\bar\bX_k^\top/\sigma^2$, with $\bar\bX_k$ the centered view, which is a scaled linear kernel. All $N-1$ directions remain, so~\eqref{eq:kcca_collapse} still holds, but only $\text{rank}(\bar\bX_k)$ of them carry non-negligible eigenvalues. Since the entries of $\bK_k$ are of order $1/\sigma^2$, for fixed $\lambda_k$ the penalty term $\lambda_k\bK_k$ dominates $\bK_k^2$ in the constraint of~\eqref{eq:kcca_emp}, and the problem approaches linear partial least squares, which maximizes the covariance rather than the correlation of linear projections of the views. Linear CCA is recovered only if $\lambda_k$ shrinks faster than $1/\sigma^2$. The two limits show that the value the method reports can depend on the bandwidth and the penalty as much as on the views. A reported correlation is therefore informative only when compared with the value that the same bandwidth and penalty return for unrelated views.

\subsubsection{The Population Target, and Consistency}
Bach and Jordan~\cite{bachjordan2002} reached the first kernel canonical
correlation of~\eqref{eq:kcca_pop} from independent component analysis (ICA).
They use it as a contrast function and call it the $\mathcal{F}$-correlation.
An ICA criterion must vanish exactly when the two variables are independent.
The $\mathcal{F}$-correlation does so when the RKHS is large enough. However, its empirical version does not, since, by~\eqref{eq:kcca_collapse}, it equals one regardless of the data. Hence, they penalize the RKHS norms, and the population object their penalized estimator targets is
\begin{linenomath}
\begin{equation}
  \sigma_{\mathcal{F}}^{\lambda}
  = \max_{f_1, f_2}
  \frac{\Cov\big(f_1(\bx_1), f_2(\bx_2)\big)}
  {\big(\Var[f_1(\bx_1)] + \lambda \|f_1\|_{\mathcal{H}_1}^2\big)^{1/2}
   \big(\Var[f_2(\bx_2)] + \lambda \|f_2\|_{\mathcal{H}_2}^2\big)^{1/2}},
  \label{eq:kcca_pop_reg}
\end{equation}
\end{linenomath}
where $\Cov$ and $\Var$ are taken under the true distribution rather than over
training samples. This is the population counterpart of~\eqref{eq:kcca_emp},
with a single constant $\lambda$ in place of the per-view $\lambda_k$. They state that the penalized empirical canonical correlation is a consistent estimator of
$\sigma_{\mathcal{F}}^{\lambda}$, and defer the proof. With a Gaussian kernel, the unpenalized value ($\lambda=0$) estimates no population quantity, since it equals one for every data set. Bach and Jordan's target $\sigma_{\mathcal{F}}^{\lambda}$ still depends on $\lambda$. Fukumizu, Bach and Gretton~\cite{fukumizu2007} reach the functions $f_1$ and $f_2$ that solve the unpenalized population kernel CCA in~\eqref{eq:kcca_pop} instead. They let $\lambda$ shrink with the sample size $N$ and prove, under regularity conditions on the joint distribution, that the estimated functions converge to the optimal ones in the $L^2$ sense. A sufficient condition on $\lambda$ is that it vanish more slowly than $N^{-1/3}$, that is, $\lambda \to 0$ with $\lambda N^{1/3} \to \infty$.

Whether a penalty is needed for consistency had in fact been settled once
before, in a scenario that predates the kernel literature. Leurgans, Moyeed and
Silverman~\cite{leurgans1993} studied CCA under the setting where each sample of a view is a curve rather than a vector. They show that without a roughness penalty, the method breaks down completely. With finitely many samples, every linear functional of the first view has a perfectly correlated partner in the second. All sample canonical correlations are therefore one, regardless of the true ones, which is the functional counterpart of the collapse that~\eqref{eq:kcca_collapse} forces. They restore consistency with a roughness penalty on the canonical directions, which are themselves curves in this setting, and choose a single smoothing parameter, which plays the role of a common $\lambda_1 = \lambda_2$ in~\eqref{eq:kcca_emp}, by cross-validation. Since their functions lie in a Sobolev space, which is itself an RKHS, their method can be read as a form of penalized kernel CCA, with the roughness penalty playing the role of the RKHS norm in~\eqref{eq:kcca_pop_reg}.

\subsubsection{The Unrestricted Problem Has Its Own Lineage}

Every method of this section so far draws the transformations from a restricted class, and Section~\ref{sec:kcca} noted that the Hirschfeld--Gebelein--R\'enyi maximal correlation ranges over all measurable functions instead. That problem was studied in the population setting by Lancaster~\cite{lancaster1958} and is solved algorithmically by alternating conditional expectation, the ACE procedure of Breiman and Friedman~\cite{breimanfriedman1985}, which alternates between conditional expectations of one transformed view given the other until the transformations stop changing. Michaeli, Wang and Livescu return to Lancaster's theory and show that the optimal transformations are given by the singular value decomposition of an operator built from the joint density of the views~\cite{michaeli2016ncca}. Replacing the joint density by an estimate from the training samples turns this into a finite eigenvalue problem. In their implementation, a Gaussian kernel density estimate restricted to a few nearest neighbors yields a sparse matrix normalized by a diagonal one, so that, unlike~\eqref{eq:kcca_gev}, no kernel matrix has to be inverted. They report that the resulting method outperforms kernel CCA and is competitive with methods that learn the transformations with neural networks, at a lower computational~cost.

Note where this line of research leaves the lesson of this subsection. Nothing bounds the transformations in the Lancaster-based approach, so the penalty of~\eqref{eq:kcca_emp} has no counterpart. What takes its place is the density estimate. The resolution of the estimated joint density is what limits how finely the transformations can distinguish points, so the safeguard against the degeneracy of~\eqref{eq:kcca_collapse} is not removed but relocated, from a norm penalty to the bandwidth or neighborhood size of a density estimator.

Two caveats attach to that relocation. The first is that estimating the joint
density is a harder problem than the one we are after, since the canonical
transformations are a functional of the density and not the density itself.
Vapnik's principle advises against exactly this, stating that ``when solving a given problem, try to avoid solving a more general problem as an intermediate step''~(\cite{vapnik2000}, Section~1.9). The second is quantitative. The sample size needed to estimate the joint density of the views grows exponentially in their total dimension $M_1 + M_2$, and the views that motivated this section, the voxel time series, the received signal blocks and the acoustic features, are all high-dimensional. Michaeli, Wang and Livescu work with a kernel density estimate over a small number of nearest neighbors, which is what keeps the method practical, and the recovered transformations rest on that estimate regardless.
 

\subsection{From Kernels to Networks}
\label{sec:bridge}

The consistency question opened by Bach and Jordan was answered by letting the penalty vanish slowly enough as the sample size grows~\cite{fukumizu2007}. Computationally, the $N \times N$ Gram matrices remained the binding constraint, since decomposing them costs $\mathcal{O}(N^3)$ if done directly, and a succession of approximations was developed to avoid it, among them incomplete Cholesky factorization~\cite{bachjordan2002}, partial Gram--Schmidt~\cite{hardoon2004}, the Nystr\"om method~\cite{williams2001}, randomized nonlinear CCA~\cite{lopezpaz2014}, and incremental SVD~\cite{brand2002}. Arora and Livescu~\cite{arora2012} brought the last of these to kernel CCA, reducing the decomposition to $\mathcal{O}(m^2N)$ time and $\mathcal{O}(mN)$ space for a truncation rank $m \ll N$, and applied it to the speech production setting discussed at the beginning of this section.

Deep CCA was proposed against that computational background. Andrew
{et~al.}~\cite{andrew2013} give two reasons for replacing the kernel by a
network. The kernel is chosen rather than learned, so it limits the
representation. Moreover, by~\eqref{eq:representer}, kernel CCA carries one
coefficient per training sample, so training scales badly with $N$ and every
new point has to be compared against the whole training set. On the other hand, a network has a fixed number of weights, set by the architecture, so it trains well on large data and can be evaluated on a new point without using the training set.
What is given up in the exchange is the closed-form geometry of
Section~\ref{sec:kcca_reg}, and with it the ability to detect degeneracy by
counting dimensions. Next, we turn to approaches that use neural networks.

\subsection{Deep CCA and the Trivial Solution Problem}
\label{sec:dcca}

Andrew et~al.~\cite{andrew2013} learn the transformations
$\boldsymbol{f}_k:\R^{M_k}\rightarrow\R^{D}$ with neural networks, drawn from a
class $\mathcal{C}$ fixed by the architecture, and maximize the sum of
canonical correlations between the two outputs. Written at the population
level, the problem is 
\begin{linenomath}
\begin{equation}
\begin{aligned}
  \max_{\boldsymbol{f}_1, \boldsymbol{f}_2 \in \mathcal{C}} \;\;
  & R^{(2)}(\boldsymbol{f}_1,\boldsymbol{f}_2) := \E\Big[\Tr\big(\boldsymbol{f}_1(\bx_1) \boldsymbol{f}_2(\bx_2)^\top\big)\Big] \\
  \text{s.t.} \;\;
  & \E\big[\boldsymbol{f}_k(\bx_k) \boldsymbol{f}_k(\bx_k)^\top\big] = \bI_D, \quad
    \E\big[\boldsymbol{f}_k(\bx_k)\big] = \mathbf{0}_D, \quad k = 1,2 .
\end{aligned}
  \label{eq:dcca}
\end{equation}
\end{linenomath}
\textls[-15]{The constraint on each view makes its $D$ variates centered, uncorrelated, and of unit variance.}

The empirical formulation replaces the expectations in~\eqref{eq:dcca} with
averages over the $N$ training samples. More specifically, the raw network outputs are first centered by subtracting the empirical mean, giving $\bar\bF_k \in \R^{N \times D}$. The within-view covariance matrices are estimated with a ridge,
$\widehat\bSigma_{kk} = \tfrac{1}{N-1}\bar\bF_k^\top\bar\bF_k + r_k\bI$ with
$r_k > 0$, and the cross-covariance matrix via $\widehat\bSigma_{12} = \tfrac{1}{N-1}\bar\bF_1^\top\bar\bF_2$. The linear CCA projections then follow in closed form to align the views, and the objective becomes the sum of the singular values of
$\widehat{\mathbf{T}} = \widehat\bSigma_{11}^{-1/2}\widehat\bSigma_{12}\widehat\bSigma_{22}^{-1/2}$,
that is, the sum of the canonical correlations of the two outputs computed from
these estimates. The authors note that the ridge both keeps the covariance
estimates nonsingular and reduces the detection of spurious correlations in the
training data.

Two features of the formulation matter later. The constraints couple all
training samples, so the objective does not decompose over data points and
stochastic methods that process one sample at a time do not apply directly.
The original work tried minibatches but obtained much better results with
full-batch L-BFGS, and later work finds that sufficiently large minibatches
suffice~\cite{wang2015unsupervised}. More importantly, $r_k$ regularizes the
covariance \emph{estimate} of the raw network outputs rather than the function class $\mathcal{C}$, which is why it does not prevent the degeneracy described below.

Benton et~al.~\cite{benton2017} extend deep CCA to $K$ views by introducing neural network-based transformations in the MAXVAR criterion of Section~\ref{sec:gcca}. They learn view networks $\boldsymbol{f}_k : \R^{M_k} \rightarrow \R^{D_k}$ and linear transformations $\bU_k \in \R^{D_k \times D}$ of the network outputs, such that all views agree on a consensus random vector $\bg \in \R^{D}$ as closely as possible. At the population level, the problem is
\begin{linenomath}
\begin{equation}
\begin{aligned}
  \min_{\{\boldsymbol{f}_k \in \mathcal{C}\}_{k=1}^K,\, \{\bU_k\}_{k=1}^K,\, \bg} \;\;
  &  \sum_{k=1}^{K}
    \E\Big[\big\|\bg - \bU_k^\top \boldsymbol{f}_k(\bx_k)\big\|^2\Big] \\
  \text{s.t.} \;\;
  & \E\big[\bg\bg^\top\big] = \bI_D, \quad \E\big[\bg\big] = \mathbf{0}_D .
\end{aligned}
  \label{eq:dgcca_U}
\end{equation}
\end{linenomath}
The empirical formulation collects the network outputs in
$\bF_k \in \R^{N \times D_k}$, whose $n$-th row is the transformed $n$-th
training sample of the $k$-th view, and the corresponding realizations of $\bg$
in $\bG \in \R^{N \times D}$. As in the two-view derivation, let
$\bar\bF_k \in \R^{N \times D_k}$ be the centered versions. With the networks
fixed, the conditionally optimal $\bG$ and the $\bU_k$ solve
$\min \sum_k \|\bG - \bar\bF_k\bU_k\|_F^2$ subject to $\bG^\top\bG = \bI_D$,
which is~\eqref{eq:maxvar2} with the fixed views replaced by nonlinearly transformed ones. The factor of $N$ separating this constraint from the population one, $\E[\bg\bg^\top] = \bI_D$, is a rescaling of $\bG$ that the $\bU_k$'s absorb. The solution meets both conditions on $\bG$. Orthonormality of its columns is
imposed directly, and they are orthogonal to $\mathbf{1}_N$ because the
$\bar\bF_k$ are centered.

For fixed networks, the problem of finding $\bG$ and the $\bU_k$ matrices has the
closed form of Section~\ref{sec:maxvar}. To obtain $\bG$, writing
$\bP_k = \bar\bF_k(\bar\bF_k^\top\bar\bF_k)^{-1}\bar\bF_k^\top$ for the
projector onto $\col(\bar\bF_k)$ and $\bM = \sum_k \bP_k$, the columns of $\bG$
are the top $D$ eigenvectors of $\bM$, as the objective equals
$DK - \Tr(\bG^\top\bM\bG)$ (see relation~\eqref{eq:maxvar_reduced}). Benton
et~al.\ therefore alternate between solving this eigenproblem for fixed
networks and, with $\bG$ fixed, backpropagating the gradient of
$\sum_k \|\bG - \bar\bF_k\bU_k\|_F^2$ through the networks. The only difference
from the linear case is that $\bM$ depends on the network weights, so the
eigenproblem is solved at every step rather than~once.

Since each $\bU_k$ is a linear map of the network output, it can be absorbed
into $\boldsymbol{f}_k$ for a sufficiently rich class $\mathcal{C}$. From now on we therefore consider functions $\boldsymbol{f}_k:\R^{M_k}\rightarrow\R^{D}$ and write the problem as \begin{linenomath}
\begin{equation}
\begin{aligned}
  \min_{\{\boldsymbol{f}_k \in \mathcal{C}\}_{k=1}^K,\, \bg} \;\;
  & R^{(K)}(\boldsymbol{f}_1,\ldots,\boldsymbol{f}_K,\bg) := \sum_{k=1}^{K}
    \E\Big[\big\|\bg - \boldsymbol{f}_k(\bx_k)\big\|^2\Big] \\
  \text{s.t.} \;\;
  & \E\big[\bg\bg^\top\big] = \bI_D, \quad \E\big[\bg\big] = \mathbf{0}_D .
\end{aligned}
  \label{eq:dgcca}
\end{equation}
\end{linenomath}

\subsubsection{Trivial Solutions Can Emerge}
In Section~\ref{sec:kcca_reg} we already met a degeneracy of this kind, forced
there by a dimension count. Here the mechanism is different, and the argument is
due to Karakasis and Sidiropoulos~\cite{karakasis2023}. With $N$ samples, the
empirical constraints of~\eqref{eq:dgcca} become $\bG^\top\bG = N\bI_D$ and
$\bG^\top\mathbf{1}_N = \mathbf{0}_D$. These can be satisfied when $N - D - 1$
rows of $\bG$ are equal to $\mathbf{0}_D$ and the remaining $D+1$ are the rows
of any matrix $\bV \in \R^{(D+1) \times D}$ satisfying $\bV^\top\bV = N\bI_D$ and
$\bV^\top\mathbf{1}_{D+1} = \mathbf{0}_D$, which exists since
$\mathbf{1}^{\perp}$ in $\R^{D+1}$ has dimension exactly $D$. At least $D+1$
rather than $D$ nonzero rows of $\bG$ are needed, since $\mathbf{1}^{\perp}$ in
$\R^{D}$ would leave only $D-1$ directions. More generally, for continuous view
distributions and a sufficiently rich $\mathcal{C}$, transformations mapping the
finitely many view training samples onto the rows of any feasible $\bG$ exist.
When the transformations of all $K$ views map onto the same rows, every term of
the empirical objective is zero, so each such $\bG$ is optimal. The argument is
therefore not specific to~\eqref{eq:dgcca}. It applies to any nonlinear
formulation carrying these constraints, and the MAXVAR form only makes it easy
to state, since the consensus matrix $\bG$ appears explicitly. The two
constraints do rule out a constant encoder, since a nonzero constant violates
the second and the zero constant violates the first. Nevertheless, any empirical formulation relying solely on consensus over finitely many training points admits trivial solutions, once the transformations are rich enough.

It is tempting to read this as the coarsening of Section~\ref{sec:gk}, but the
two differ in what feasibility means. Under Gács--Körner, the constraint
$W = f(X) = g(Y)$ holds almost surely, so the feasible set is fixed by the
joint law and every feasible $W$ is measurable with respect to the meet
$\sigma(X) \cap \sigma(Y)$. Here, agreement is imposed at $N$ points only.
For a rich enough $\mathcal{C}$ the attainable set of $\bG$ is therefore all of
$\{\bG : \bG^\top\bG = N\bI_D,\ \bG^\top\mathbf{1}_N = \mathbf{0}_D\}$, which
does not depend on $\bX_1,\ldots,\bX_K$, and the empirical version of $R^{(K)}$ vanishes on every point of it. The solution set of the empirical problem can thus be the entire feasible set, and it carries no information about the joint law. Returning to the population settings, Witsenhausen's criterion of Section~\ref{sec:gk} is relevant. Under mild conditions, it states that $K_{GK}(X;Y) > 0$ if and only if $S(X,Y) = 1$, which certifies that a nontrivial common part exists, $H(W^*) > 0$, without identifying $W^*$. As we will see later, assuming the existence of a latent common factor $\bg$, almost-sure agreement buys
identifiability of a potentially non-invertible function of it, $\psi(\bg)$.
The functions $\psi$ form the spectrum of common latent factors described at the
beginning of this section, and when $\bg$ is discrete, the Gács--Körner
formulation suggests looking among them for the $\psi$ maximizing the entropy of
$\psi(\bg)$. Section~\ref{sec:ssl} returns to this idea in the weaker form of a
variance floor.

\subsubsection{Reconstruction as a Remedy}

The trivial solution is not what Wang {et~al.}~\cite{wang2015} set out to
remove. The problem they study is learning multi-view representations for
downstream tasks. Specifically, they are interested in the setting where two or
more views are available for training and only one at test time. Hence, they
seek a transformation of the view available at test time that captures the
information the remaining views provide, but without using any labels from the
downstream tasks. Note that, as stated, the task is not the recovery of a latent
factor per se. Wang {et~al.} consider two families of neural objectives for
that problem. The two differ in what they ask a learned representation to
preserve. Reconstruction asks the representation to retain enough information to
reconstruct the inputs. Correlation asks the representations of the two views to
agree.

The first family is represented by the deep autoencoder of Ngiam {et~al.}~\cite{ngiam2011}, which Wang {et~al.} call the split autoencoder. A
single encoder is applied to the view available at test time. Its output must
reconstruct both views. Hence, the second view enters as the target of a decoder
rather than through an agreement term. The second family is represented by the deep CCA of~\eqref{eq:dcca}. Wang {et~al.} perform a comparison of such techniques across several tasks and report an advantage for the correlation-based
objectives. The comparison includes a variant of their own, called deep canonically
correlated autoencoders (DCCAE), which combines the two and gives the best results on most of their tasks.

Decoders $\boldsymbol{w}_k$ mapping $\R^{D}$ to $\R^{M_k}$ are introduced, with reconstruction penalty
\begin{linenomath}
\begin{equation}
  L^{(K)}(\boldsymbol{f}_1,\boldsymbol{w}_1,\ldots,\boldsymbol{f}_K,\boldsymbol{w}_K)
  = \sum_{k=1}^{K} \E\Big[\big\|\bx_k
    - \boldsymbol{w}_k\big(\boldsymbol{f}_k(\bx_k)\big)\big\|^2\Big],
  \label{eq:dccae_recon}
\end{equation}
\end{linenomath}
and, for two views, DCCAE solves the problem
\begin{linenomath}
\begin{equation}
\begin{aligned}
  \max_{\{\boldsymbol{f}_k, \boldsymbol{w}_k \in \mathcal{C}\}_{k=1,2}} \;\;
  & R^{(2)}(\boldsymbol{f}_1,\boldsymbol{f}_2) \;-\; \lambda\,
    L^{(2)}(\boldsymbol{f}_1,\boldsymbol{w}_1,\boldsymbol{f}_2,\boldsymbol{w}_2) \\
  \text{s.t.} \;\;
  & \E\big[\boldsymbol{f}_k(\bx_k) \boldsymbol{f}_k(\bx_k)^\top\big] = \bI_D,
    \quad
    \E\big[\boldsymbol{f}_k(\bx_k)\big] = \mathbf{0}_D, \quad k = 1,2,
\end{aligned}
  \label{eq:dccae}
\end{equation}
\end{linenomath}
where $\lambda>0$ sets the trade-off between the two terms, with larger values
favoring reconstruction. Note that in~\eqref{eq:dccae_recon} each view is 
reconstructed from its own encoding, unlike the split autoencoder, where one encoding reconstructs both.

The trade-off has an information-theoretic reading, which the authors give
explicitly. Taken alone, the first term maximizes the canonical correlations
between the two transformed views. Wang {et~al.} read this as maximizing
their mutual information, which by~\eqref{eq:mi_cca} increases with those
correlations for jointly Gaussian views. The reconstruction term alone maximizes
a lower bound on the mutual information between each view and its own encoding, which is the reading Vincent {et~al.}~\cite{vincent2010} give of any autoencoder. Hence, the objective mediates between what a view shares with its own encoding and what the two encodings share with each other. Wang {et~al.} read this as the trade-off of the information bottleneck~\cite{tishby1999}. 

The reconstruction term does rule out the trivial solutions discussed earlier. In
the construction given there, $N - D - 1$ training samples of view $k$ are mapped
to $\mathbf{0}_D$. A decoder is a function, so it returns one vector for all of
them, and~\eqref{eq:dccae_recon} therefore forces the encoders to keep distinct
training samples apart. However, we cannot guarantee recovery of anything else, since reconstruction is imposed at finitely many points as well. Assume the two transformed views agree exactly on the training samples. Because of the constraints and that perfect alignment, the transformed samples of both
views form the same matrix $\bG$, with orthogonal columns that are also orthogonal to $\mathbf{1}_N$. Suppose further that its rows are distinct. For any such $\bG$,
given continuous view distributions and a sufficiently rich function class
$\mathcal{C}$, there are encoders mapping $\bx_k^{(n)}$ to the $n$-th row of
$\bG$ and decoders mapping that row back to $\bx_k^{(n)}$. The correlation term
is then maximal and the reconstruction penalty vanishes, so the two are optimized
together and $\lambda$ plays no part. Therefore, the solution set is narrowed from every feasible matrix to every feasible matrix with distinct rows. A richer class $\mathcal{C}$ attains more of these matrices, so the richer it is, the less the solution set depends on the data.

Apart from that extreme, it is not clear if both terms of the problem formulation
can be met exactly. The encodings have to reflect both to some extent. Information that is not shared among the views, i.e., private information, may be needed for satisfying reconstruction. Following Karakasis and Sidiropoulos~\cite{karakasis2023}, we call the admission of such information into the encoding \emph{leakage}. At the same time, leakage makes the encoding alignment harder. Although it is unclear what exactly DCCAE learns, its formulation differs from those of Andrew {et~al.} \cite{andrew2013} and Benton {et~al.} \cite{benton2017} in one important respect. The encodings must also serve reconstruction, so they cannot be free of the views they come from. Neither~\eqref{eq:dcca} nor~\eqref{eq:dgcca} asks for that, and the correlation they report is between representations that do not necessarily reflect 
their own view.

 
\subsection{Probabilistic and Variational Formulations}
\label{sec:vcca}


Most of the nonlinear methods so far have learned transformations that make the
views correlate as much as possible, and have treated the transformed views as
estimates of the common latent factor. Section~\ref{sec:dcca} ended with the
formulation of DCCAE, which asks each encoding to reconstruct its own view,
through a deterministic decoder, while the two encodings achieve high
cross-correlation. Although this requirement makes the learned encodings carry
information about their own views, it also creates a potential conflict. This is
because the learned encodings are now required to serve two purposes that do not
necessarily align. A deterministic decoder is a function, so it assigns the
same reconstruction to all view samples mapped to the same encoding. To keep the
reconstruction error small, the encoding itself has to tell those samples apart.

A stochastic decoder, i.e., a conditional distribution of the view given the
encoding, handles those samples differently. It spreads its mass over them, so
the variation among them is represented by the conditional distribution,
evaluated at that encoding, rather than by the encoding itself. The problem that
emerges now is that reconstruction error is no longer a reasonable measure of
fit for learning an encoder paired with a stochastic decoder, since the decoder
returns a distribution rather than a single reconstruction. The natural measure
is instead the likelihood this distribution assigns to the observed view.

Learning under stochastic decoders therefore calls for likelihood-based
objectives. The starting point is the probabilistic model~\eqref{eq:prob_cca}
of Section~\ref{sec:cca_prob} rather than Hotelling's ratio, and no objective
below is defined in terms of correlation. The methods of this subsection reverse
the logic of the deterministic ones. They assume a probabilistic generative
model of the views that contains a common latent factor. Then, they fit that
model to the training samples by maximizing their likelihood, or a lower bound on
it, and infer the common factor from the fitted model.

\subsubsection{The Common-Factor Model}
The methods of this first group all assume the same thing. A single latent
factor is described as a random vector $\bg \in \R^{D}$ following a fixed prior $p(\bg)$, and the $K$ views are conditionally independent given it, i.e.,\vspace{12pt}
\vspace{-24pt}
\begin{adjustwidth}{-\extralength}{0cm}
\begin{linenomath}
\begin{equation}
  p(\bx_1,\ldots,\bx_K,\bg)
  \;=\; p(\bg)\prod_{k=1}^{K} p_{\theta_k}(\bx_k \mid \bg),
  \qquad
  p(\bx_1,\ldots,\bx_K)
  \;=\; \int p(\bg) \prod_{k=1}^{K} p_{\theta_k}(\bx_k \mid \bg)\, d\bg ,
  \label{eq:shared_only}
\end{equation}
\end{linenomath}
\end{adjustwidth}
where each observation model $p_{\theta_k}(\bx_k \mid \bg)$ is parameterized by a
neural network with parameters $\theta_k$. The first factorization
in~\eqref{eq:shared_only} is the $K$-view conditional
independence~\eqref{eq:wyner_K_ci} of Section~\ref{sec:wyner_K}. Note that no
variable other than $\bg$ appears. Whatever in view $k$ the common factor does
not account for has to be absorbed into $p_{\theta_k}(\bx_k \mid \bg)$.

Before any of these stochastic methods is described, it is worth stating what
\eqref{eq:shared_only} does and does not fix. It is a feasibility condition on
$\bg$ and not a definition of it. The feasible set spans a range, as does the
feasible set of Wyner's common information in Section~\ref{sec:wyner_K}, which
is defined by the same conditional independence. At one end $\bg$ is a minimally
informative factor rendering the views conditionally independent. At the other
end, the choice $\bg = (\bx_1,\ldots,\bx_K)$ satisfies the condition as well,
since each $p(\bx_k \mid \bg)$ is then a point mass and the product factorizes
for that reason alone. Many feasible factors lie between the two ends. Wyner
resolved the same ambiguity by minimizing the mutual information between the
views and the candidate factor over that feasible set, which selects the least
informative factor in it.

The methods below fit~\eqref{eq:shared_only} by maximizing the likelihood, or a
lower bound on it, i.e., $p(\bx_1,\ldots,\bx_K)$ evaluated on a collection of
training samples, over the model parameters. Since $\bg$ is integrated out, that
objective alone sees it only through the view marginal it induces. Two fitted
models attaining the same empirical likelihood may therefore sit at different
points of the range described above, which leaves unclear which $\bg$ either of them is learning.

\subsubsection{Deep Variational CCA}
For $K=2$ in~\eqref{eq:shared_only}, Wang et~al.~\cite{wangvcca2016}
parameterize the two observation models by neural networks, which gives the
model of~\eqref{eq:prob_cca} with the linear maps replaced by networks. The
resulting marginal likelihood $p(\bx_1,\bx_2)$ in~\eqref{eq:shared_only}
requires integrating the product of two network-parameterized densities against
$p(\bg)$, and the posterior
$p(\bg \mid \bx_1) = p(\bg)\,p(\bx_1 \mid \bg)/p(\bx_1)$ carries an integral of
the same kind in its normalizing constant
$p(\bx_1) = \int p(\bg)\,p(\bx_1 \mid \bg)\,d\bg$. Neither integral has a closed
form.

This limitation can be bypassed by the approach used for the variational
autoencoder (VAE)~\cite{kingma2014}. A stochastic encoder $q(\bg \mid \bx_1)$, also
implemented by a neural network, is introduced to stand in for the posterior.
Assuming the $N$ training samples are drawn independently, the likelihood of the
collection factorizes as $\prod_{n=1}^{N} p\big(\bx_1^{(n)},\bx_2^{(n)}\big)$,
and its logarithm is the sum
$\sum_{n=1}^{N} \log p\big(\bx_1^{(n)},\bx_2^{(n)}\big)$. It is therefore enough
to work with $\log p(\bx_1,\bx_2)$. The derivation of the VAE approach proceeds as
follows. First, Bayes' rule gives
\linebreak $p(\bx_1,\bx_2) = p(\bx_1,\bx_2,\bg) / p(\bg \mid \bx_1,\bx_2)$, and this holds
for every $\bg$ for which the ratio is well defined. Then, multiplying and
dividing its right-hand side by $q(\bg \mid \bx_1)$ and taking the logarithm
gives
\begin{linenomath}
\begin{equation}
  \log p(\bx_1,\bx_2)
  \;=\; \log \frac{q(\bg \mid \bx_1)}{p(\bg \mid \bx_1,\bx_2)}
  \;+\; \log \frac{p(\bx_1,\bx_2,\bg)}{q(\bg \mid \bx_1)} ,
  \label{eq:elbo_split}
\end{equation}
\end{linenomath}
which is still an identity at every $\bg$ for which the ratios are well defined.

Next, we take the expectation of
both sides of~\eqref{eq:elbo_split} under $q(\bg \mid \bx_1)$. The left side is
unchanged, since $\log p(\bx_1,\bx_2)$ does not depend on $\bg$, and the 
expectation of the first term on the right is
by definition a Kullback--Leibler divergence. Hence, we obtain
\begin{linenomath}
\begin{equation}
  \log p(\bx_1,\bx_2)
  \;=\; D_{\mathrm{KL}}\big(q(\bg \mid \bx_1)\,\|\,p(\bg \mid \bx_1,\bx_2)\big)
  \;+\; \E_{q(\bg \mid \bx_1)}\!\left[
    \log \frac{p(\bx_1,\bx_2,\bg)}{q(\bg \mid \bx_1)} \right] .
  \label{eq:elbo_gap}
\end{equation}
\end{linenomath}
The first term of~\eqref{eq:elbo_gap} is always non-negative, so the second term is a lower bound on
$\log p(\bx_1,\bx_2)$. Substituting the factorization~\eqref{eq:shared_only}
into that term and separating the part carrying $p(\bg)/q(\bg \mid \bx_1)$ gives
the evidence lower bound
\begin{linenomath}
\begin{equation}
  \log p(\bx_1,\bx_2) \;\ge\;
  -D_{\mathrm{KL}}\big(q(\bg \mid \bx_1)\,\|\,p(\bg)\big)
  + \E_{q(\bg \mid \bx_1)}\big[\log p(\bx_1 \mid \bg) + \log p(\bx_2 \mid \bg)\big],
  \label{eq:vcca}
\end{equation}
\end{linenomath}
which is what deep variational CCA maximizes over the networks parameterizing
$q$, $p(\bx_1 \mid \bg)$ and $p(\bx_2 \mid \bg)$.

Conditioning the posterior on a single view is deliberate, designed for scenarios
where only one view is present downstream. As an alternative, Wang et~al.~\cite{wangvcca2016} also form a convex combination of two bounds of the form above, one for each view. In their setup, the approximate posterior is Gaussian with a diagonal covariance whose mean and variances are the outputs of a view encoder network, which allows computing the Kullback--Leibler term of~\eqref{eq:vcca} in closed form. The remaining expectation is approximated by sampling. For the $n$-th training sample, the encoder of the first view returns a mean $\bmu\big(\bx_1^{(n)}\big)$ and standard deviations $\bsigma\big(\bx_1^{(n)}\big)$. Then, a draw from $q\big(\bg \mid \bx_1^{(n)}\big)$ is formed as $\bg^{(n)} = \bmu\big(\bx_1^{(n)}\big)
+ \bsigma\big(\bx_1^{(n)}\big) \odot \bepsilon$, where
$\bepsilon \sim \mathcal{N}(\mathbf{0},\bI_D)$ and $\odot$ is the entrywise
product. Hence, the sample-based estimate is a differentiable function of the
network parameters, since the randomness enters only through $\bepsilon$. Unlike
the whitening constraints of~\eqref{eq:dcca} and~\eqref{eq:dgcca}, which couple
all training samples, summing~\eqref{eq:vcca} over the training samples gives an
objective that does not couple them. The two bounds share the decoders, so both
encoders are tied to one latent space. However, the two encoders approximate
different posteriors, $p(\bg \mid \bx_1)$ and $p(\bg \mid \bx_2)$, so their
outputs for a jointly observed pair of views need not agree. The formulation
does not say which of the two to take as the representation, and when both views
are available at test time the authors concatenate them.

\subsubsection{The Multimodal Variational Autoencoders}
A parallel line arrives at~\eqref{eq:shared_only} for general $K$ from
multimodal generation rather than from CCA, and it takes up the question the
previous paragraph left open. Since there are $2^K$ possible subsets of views
that might be observed, a separate inference network for each is not
affordable. Instead, a joint posterior is assembled from the per-view posteriors
$q(\bg \mid \bx_k)$. Wu and Goodman~\cite{wu2018mvae}
assemble it as a product of experts (MVAE)
\begin{linenomath}
\begin{equation}
  q(\bg \mid \bx_1,\ldots,\bx_K) \;\propto\;
  p(\bg)\prod_{k=1}^{K} \tilde q(\bg \mid \bx_k),
\end{equation}
\end{linenomath}
where $\tilde q(\bg \mid \bx_k)$ is the network output, related to the per-view
posterior approximation by $q(\bg \mid \bx_k) \propto \tilde q(\bg \mid \bx_k)\,p(\bg)$.
Shi {et~al.}~\cite{shi2019mmvae}
assemble it as a mixture of experts (MMVAE)
\begin{linenomath}
\begin{equation}
  q(\bg \mid \bx_1,\ldots,\bx_K) \;=\;
  \frac{1}{K}\sum_{k=1}^{K} q(\bg \mid \bx_k),
\end{equation}
\end{linenomath}
and Sutter, Daunhawer and
Vogt~\cite{sutter2021} as a mixture over all subsets (MoPoE-VAE)
\begin{linenomath}
\begin{equation}
  q(\bg \mid \bx_1,\ldots,\bx_K) \;=\;
  \frac{1}{2^{K}} \sum_{\mathcal{S} \subseteq \{1,\ldots,K\}}
  q_{\mathcal{S}}(\bg) ,
  \qquad
   q_{\mathcal{S}}(\bg) \;\propto\; \prod_{k \in \mathcal{S}}
  q(\bg \mid \bx_k).
\end{equation}
\end{linenomath}

Two features matter here. The generative assumption~\eqref{eq:shared_only} is
identical in MVAE, MMVAE and MoPoE-VAE, so this is a development of the
inference and not of the model, and any difference in what is recovered comes
from the posterior and the estimator. Moreover, both Shi {et~al.}\ and Sutter {et~al.}\ open by asking for a latent space that separates what the views share from what is private to each. Shi {et~al.}\ list that separation first among the criteria such a model should meet. Yet Shi {et~al.}\ put no private
variable in the model, and Sutter {et~al.}\ add modality-specific latent
spaces only in one experiment, outside their main~\mbox{formulation}.

\subsubsection{Variational CCA with Private Factors}

A more complete generative model also considers a private factor for each view,
responsible for the variation that the common factor does not explain. In
several applications, having access to the private factors may also be of
interest for downstream tasks, and separating the two is the latent
disentanglement named in the introduction of this section. Moreover, in many
applications, such as image generation, a generative model is assessed by the
quality of the samples it generates, which requires the decoders to capture the
variation of each view. The limitation of considering a common factor alone shows up in reconstruction. For Gaussian observation models, the expectation term of the
bound is a reconstruction error, and the best reconstruction it admits is
$\E[\bx_k \mid \bg]$, so all realizations of view $k$ sharing the same $\bg$
receive the same reconstruction. As Wang {et~al.} report, assuming that the
common factor is sufficient to generate the views can be too restrictive when
the input carries large variations that the common factor cannot explain.
Moreover, such large variations make the bound harder to optimize. When this
happens, it may be beneficial to model private factors explicitly, alongside the
common one.

Wang {et~al.}\ extend VCCA in~\eqref{eq:vcca} by introducing a private
factor $\bc_k \in \R^{L_k}$ for each view, with
$p(\bg,\bc_1,\bc_2) = p(\bg)\,p(\bc_1)\,p(\bc_2)$ and
$\bc_k \sim \mathcal{N}(\mathbf{0},\bI_{L_k})$, and letting each view
observation model be $p(\bx_k \mid \bg, \bc_k)$. The variational posterior $q(\bg,\bc_1,\bc_2 \mid \bx_1,\bx_2)$ is taken in the factored form $q(\bg \mid \bx_1)\,q(\bc_1 \mid \bx_1)\,q(\bc_2 \mid \bx_2)$, one network for each factor, so the private factor of a view is inferred from that view alone. Repeating the derivation of~\eqref{eq:vcca} with this posterior adds one Kullback--Leibler term per private factor and replaces the two expected log-likelihoods with those of $p(\bx_k \mid \bg, \bc_k)$. Hence, the best possible  reconstruction is no longer given by
$\E[\bx_k \mid \bg]$ but by $\E[\bx_k \mid \bg, \bc_k]$. However, the formulation admits a degeneracy of its own, which the authors also state. The path from $\bx_k$
through $\bc_k$ back to $\bx_k$ can reconstruct view $k$ on its own while $\bg$
is ignored, since nothing in the bound requires the fitted $p(\bx_k \mid \bg, \bc_k)$ to depend on $\bg$. They argue that applying dropout to the samples of $\bg$, $\bc_1$ and $\bc_2$ helps prevent this, and find dropout on these samples to give most of their empirical gain. No theoretical guarantee is provided that the
fitted model separates the common factor from the private ones.

We left deterministic nonlinear CCA having seen that requiring each encoding to
reconstruct its own view helps to avoid trivially correlated encodings. Moving
from DCCAE to the probabilistic formulations of this subsection removes the
conflict those encodings face, since reconstruction now arrives as a likelihood
rather than as a competing term. However, generative models of high quality
need decoders that capture the variation of each view, which leads to models
that carry both common and private factors. VCCA with private factors, and the
more recent approaches of Karami and Schuurmans~\cite{karami2021} and Damianou {et~al.}~\cite{damianou2021}, still leave open the question of when common and private factors can be reliably separated in the latent space. A crucial assumption throughout this subsection is that the observed views are conditionally independent given the common factor $\bg$, which is the conditional independence that defines the feasible set of Wyner's common information. Next, we turn to two works that provide conditions and algorithms with disentanglement and identifiability guarantees, under deterministic mappings.


\subsection{What Makes the Common Factor Recoverable}
\label{sec:idgcca}

Every method so far has left the same question open. Maximizing either
correlation or a likelihood produces a latent factor, but in neither case could we
say whether what is learned is purely common, or whether all of the common
information is extracted. Two works close that gap. To answer either question
one has to define first what the common factor is and how it enters the views,
so both works start from an explicit generative model. Both assume that the views are generated from a common factor together with a private factor per view, both
take the views conditionally independent given the common factor, and both give
conditions under which the learned factor is the common one up to an invertible
reparameterization. They differ in how much they assume and in how much they
recover, and that trade is the subject of this subsection.

\subsubsection{The Generative Model}
The works of Lyu {et~al.}~\cite{lyu2022} and Karakasis and Sidiropoulos~\cite{karakasis2023} adopt the same mixing model,
\begin{linenomath}
\begin{equation}
  \bx_k = \boldsymbol{v}_k\!\left(\begin{bmatrix} \bg \\ \bc_k \end{bmatrix}\right),
  \qquad k = 1,\ldots,K,
  \label{eq:gen_model}
\end{equation}
\end{linenomath}
with $\bg \in \R^{D}$ the common latent random vector, $\bc_k \in \R^{L_k}$ the
private one of view $k$, and
$\boldsymbol{v}_k : \R^{D + L_k} \rightarrow \R^{M_k}$ the function that maps
them into view $k$. Note that~\eqref{eq:gen_model} is the model behind
VCCA with private factors, with the observation model $p(\bx_k \mid \bg, \bc_k)$ made deterministic. What separates the two works is the assumption placed on the
joint law of $\bg$ and the $\bc_k$, and on the maps
$\boldsymbol{v}_k$. 

The two works also share the measure of agreement. It is the criterion $R^{(K)}$ of~\eqref{eq:dgcca}, written with
$\widehat\bg$ for the learned consensus to distinguish it from the common factor
$\bg$ of~\eqref{eq:gen_model},
\begin{linenomath}
\begin{equation}
  \bar R^{(K)}\big(\boldsymbol{f}_1,\ldots,\boldsymbol{f}_K,\widehat\bg\big)
  := \sum_{k=1}^{K}
     \E\Big[\big\| \widehat\bg - \boldsymbol{f}_k(\bx_k) \big\|^2\Big] .
  \label{eq:consensus_nored}
\end{equation}
\end{linenomath}

\subsubsection{Invertibility Recovers the Common Factor, Independence Adds the
Private Ones}

Lyu {et~al.}~\cite{lyu2022} consider the two-view case ($K=2$). They assume that the latent factors are jointly group independent
\begin{linenomath}
\begin{equation}
  p(\bg, \bc_1, \bc_2) = p(\bg)\, p(\bc_1)\, p(\bc_2)
  \label{eq:group_indep}
\end{equation}
\end{linenomath}
and that each $\boldsymbol{v}_k$ is smooth and invertible. The output of each encoder $\boldsymbol{f}_k : \R^{M_k} \rightarrow \R^{D + L_k}$ is split into a shared and a private block,
$\boldsymbol{f}_k(\bx_k) = [\boldsymbol{f}_{k,s}(\bx_k)^\top,\,
\boldsymbol{f}_{k,p}(\bx_k)^\top]^\top$ with
$\boldsymbol{f}_{k,s} : \R^{M_k} \rightarrow \R^{D}$ and
$\boldsymbol{f}_{k,p} : \R^{M_k} \rightarrow \R^{L_k}$.

Their learning criterion combines three terms,
\vspace{-4pt}
\begin{adjustwidth}{-\extralength}{0cm}
\begin{linenomath}
\begin{equation}
\begin{aligned}
  \min_{\{\boldsymbol{f}_k\},\{\boldsymbol{w}_k\},\,\widehat\bg} \;
  \max_{\{\phi_k\},\{\tau_k\}} \;\;
  & \bar R^{(2)}\big(\boldsymbol{f}_{1,s},
    \boldsymbol{f}_{2,s}, \widehat\bg\big)
    \;+\; \lambda\, L^{(2)}(\boldsymbol{f}_1,\boldsymbol{w}_1,
    \boldsymbol{f}_2,\boldsymbol{w}_2)
    + \beta\, V^{(2)}(\boldsymbol{f}_1,\phi_1,\tau_1,
    \boldsymbol{f}_2,\phi_2,\tau_2)\\
 \quad \quad \quad \text{s.t.} \quad 
  & \quad \quad \quad \quad \quad \quad\E\big[\widehat\bg\,\widehat\bg^\top\big] = \bI_D, \quad
    \E[\widehat\bg] = \mathbf{0}_D ,
\end{aligned}
  \label{eq:lyu}
\end{equation}
\end{linenomath}
\end{adjustwidth}
where $\lambda, \beta \ge 0$ are weights, $\bar R^{(2)}$ is the consensus criterion in~\eqref{eq:consensus_nored},
$L^{(2)}$ is the reconstruction penalty in~\eqref{eq:dccae_recon} with decoders
$\boldsymbol{w}_k : \R^{D + L_k} \rightarrow \R^{M_k}$, and $V^{(2)}$ is given
in~\eqref{eq:indep_penalty} below. 

The first term asks the shared blocks of the two views to agree. The second promotes invertibility of each $\boldsymbol{f}_k$, which the authors adopt to avoid trivial solutions, such as those of Section~\ref{sec:dcca}. The third drives the shared and private blocks of a single view apart,
\begin{linenomath}
\begin{equation}
  V^{(K)}(\boldsymbol{f}_1,\phi_1,\tau_1,\ldots,
  \boldsymbol{f}_K,\phi_K,\tau_K) := \sum_{k=1}^{K}
  \frac{\Cov\big(\phi_k(\boldsymbol{f}_{k,s}(\bx_k)),\,
        \tau_k(\boldsymbol{f}_{k,p}(\bx_k))\big)}
       {\sqrt{\Var[\phi_k(\boldsymbol{f}_{k,s}(\bx_k))]}\;
        \sqrt{\Var[\tau_k(\boldsymbol{f}_{k,p}(\bx_k))]}} ,
  \label{eq:indep_penalty}
\end{equation}
\end{linenomath}
where $\phi_k : \R^{D} \rightarrow \R$ and $\tau_k : \R^{L_k} \rightarrow \R$.
Each summand of~\eqref{eq:indep_penalty}, for fixed $\boldsymbol{f}_k$ and
maximized over $\phi_k$ and $\tau_k$ as in~\eqref{eq:lyu}, is the
Hirschfeld--Gebelein--R\'enyi maximal correlation of
Section~\ref{sec:common_info} between the two blocks of $\boldsymbol{f}_k$. 
The maximal correlation vanishes if and only if the two blocks are independent~\cite{renyi1959}, so it serves as an independence test, and driving it to zero enforces independence of $\boldsymbol{f}_{k,s}(\bx_k)$ and $\boldsymbol{f}_{k,p}(\bx_k)$.

Next, we present their sufficient conditions for recovering the common latent
factor, as well as for recovering the private latent factors and hence
achieving effective latent disentanglement. Their analysis is carried out not on~\eqref{eq:lyu} but on the problem~\eqref{eq:dcca} applied to the shared blocks, with the invertibility and independence terms imposed as requirements rather than penalties, 
\begin{linenomath}
\begin{equation}
\begin{aligned}
  \max_{\{\boldsymbol{f}_k\}} \quad
  &  \quad \quad \quad \quad \quad \quad \quad  R^{(2)}\big(\boldsymbol{f}_{1,s}, \boldsymbol{f}_{2,s}\big) \\
  \text{s.t.} \quad
  & \E\big[\boldsymbol{f}_{k,s}(\bx_k)\,\boldsymbol{f}_{k,s}(\bx_k)^\top\big] = \bI_D, \quad
    \E\big[\boldsymbol{f}_{k,s}(\bx_k)\big] = \mathbf{0}_D, \\
  & \boldsymbol{f}_k \ \text{invertible}, \quad
    \boldsymbol{f}_{k,s}(\bx_k) \perp\!\!\!\perp \boldsymbol{f}_{k,p}(\bx_k),
    \qquad k = 1,2 .
\end{aligned}
  \label{eq:lyu_exact}
\end{equation}
\end{linenomath}
{For two views, MAXVAR coincides with CCA }(Section~\ref{sec:gcca}). Hence, once the linear maps are absorbed into the encoders, the consensus term of~\eqref{eq:lyu} and the correlation objective of~\eqref{eq:lyu_exact} {lead to the same shared blocks, up to an invertible linear map.}

Their Theorems~1 and~2 combine into the following result.
\begin{Theorem}[\cite{lyu2022}]
\label{thm:lyu}
Consider the generative model~\eqref{eq:gen_model} with $K = 2$, each
$\boldsymbol{v}_k$ smooth and invertible, and let~\eqref{eq:group_indep} hold.
Let
$\mathcal{X}_k = \big\{\boldsymbol{v}_k([\bg^\top,\bc_k^\top]^\top) :
\bg \in \mathcal{G},\ \bc_k \in \mathcal{U}_k\big\}$, with $\mathcal{G}$ and
$\mathcal{U}_k$ the supports of $p(\bg)$ and $p(\bc_k)$. Moreover, let
$\widehat{\boldsymbol{f}}_k$ be any solution of~\eqref{eq:lyu_exact} whose
constraints hold for all $\bx_k \in \mathcal{X}_k$, and assume that the
first-order derivative of
$\widehat{\boldsymbol{f}}_k \circ \boldsymbol{v}_k$ exists. Then
$\widehat{\boldsymbol{f}}_{k,s}(\bx_k) = \boldsymbol{\gamma}(\bg)$, where
$\boldsymbol{\gamma} : \R^{D} \rightarrow \R^{D}$ is an unknown invertible
function, whether or not the independence constraint
of~\eqref{eq:lyu_exact} is enforced. If it is enforced, then in addition
$\widehat{\boldsymbol{f}}_{k,p}(\bx_k) = \boldsymbol{\delta}_k(\bc_k)$, where
$\boldsymbol{\delta}_k : \R^{L_k} \rightarrow \R^{L_k}$ is an unknown invertible
function.
\end{Theorem}

The independence constraint is therefore needed only for the private factors.
Correlation maximization together with invertibility already delivers the common
one, which makes that constraint, and hence the
penalization of~\eqref{eq:indep_penalty}, optional when $\bg$ is the only target,
as the authors also note. Even so, the full encoders still have to be learned,
together in practice with the decoders that promote their invertibility, since
invertibility is a requirement on the whole of $\boldsymbol{f}_k$ and not on its
shared block alone.

The second conclusion is the strongest in this subsection, and it is bought with
the strongest assumption. Condition~\eqref{eq:group_indep} forbids any
dependence among the private factors and any dependence between a private factor
and the common one. Note that~\eqref{eq:group_indep} is the assumption VCCA with
private factors already made in Section~\ref{sec:vcca}. What Lyu {et~al.}\
add are deterministic, smooth and invertible maps $\boldsymbol{v}_k$, and that is
what turns the same assumption into a guarantee.

Their analysis assumes the latent dimensions to be known in advance.
When they are not, the guarantee no longer applies, and common components can
be modeled as private, since~\eqref{eq:indep_penalty} acts between the two
blocks of a single view and never between $\boldsymbol{f}_{1,p}$ and
$\boldsymbol{f}_{2,p}$. In practice, the authors select the dimensions on a
validation set.

 \subsubsection{Partial Invertibility and Conditional Independence Recover the Common Factor}
Karakasis and Sidiropoulos~\cite{karakasis2023} focus on recovering the common
factor alone, since in many applications the private factors correspond to noise
or interference that is of no interest. Motivated by this, they ask how much the
assumptions of Lyu {et~al.}\ can be weakened. They keep~\eqref{eq:gen_model}
for general $K$ and replace~\eqref{eq:group_indep} with\vspace{12pt}
\begin{linenomath}
\begin{equation}
  p(\bg, \bc_1, \ldots, \bc_K) = p(\bg) \prod_{k=1}^{K} p(\bc_k \mid \bg),
  \label{eq:cond_indep}
\end{equation}
\end{linenomath}
which implies
\begin{linenomath}
\begin{equation}
  p(\bx_1,\ldots,\bx_K)
  = \int p(\bg) \prod_{k=1}^{K} p(\bx_k \mid \bg)\, d\bg.
  \label{eq:views_cond_indep}
\end{equation}
\end{linenomath}
Condition~\eqref{eq:cond_indep} is strictly weaker than~\eqref{eq:group_indep},
since it permits each private factor to depend on the common factor, and hence
the private factors to depend on one another through it. Figure~\ref{fig:latent_models} draws a comparison of the two structures. At the same time, the
views remain conditionally independent given $\bg$, as~\eqref{eq:views_cond_indep}
states. As we will see, the assumptions on the mixing functions
$\boldsymbol{v}_k$ can also be relaxed, since they need only be partially
invertible with respect to the common factor and not with respect to the private
ones. Not having to estimate the private factors, Karakasis and Sidiropoulos take
their encoders to be $\boldsymbol{f}_k : \R^{M_k} \rightarrow \R^{D}$ and their
decoders $\boldsymbol{w}_k : \R^{D} \rightarrow \R^{M_k}$.

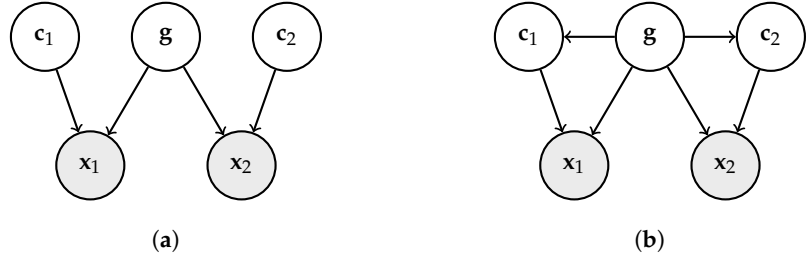
\begin{figure}[H]
\centering
\begin{tikzpicture}[
  every node/.style={font=\small},
  lat/.style={circle,draw,thick,minimum size=9mm,inner sep=0pt},
  obs/.style={circle,draw,thick,fill=black!8,minimum size=9mm,inner sep=0pt}]
\begin{scope}
\node[lat] (c1) at (-1.6,1.7) {$\bc_1$};
\node[lat] (g)  at (0,1.7)    {$\bg$};
\node[lat] (c2) at (1.6,1.7)  {$\bc_2$};
\node[obs] (x1) at (-1.0,0)   {$\bx_1$};
\node[obs] (x2) at (1.0,0)    {$\bx_2$};
\draw[->,thick] (g)  -- (x1);
\draw[->,thick] (g)  -- (x2);
\draw[->,thick] (c1) -- (x1);
\draw[->,thick] (c2) -- (x2);
\node at (0,-1.0) {(\textbf{a})};
\end{scope}
\begin{scope}[xshift=6.4cm]
\node[lat] (c1) at (-1.6,1.7) {$\bc_1$};
\node[lat] (g)  at (0,1.7)    {$\bg$};
\node[lat] (c2) at (1.6,1.7)  {$\bc_2$};
\node[obs] (x1) at (-1.0,0)   {$\bx_1$};
\node[obs] (x2) at (1.0,0)    {$\bx_2$};
\draw[->,thick] (g)  -- (x1);
\draw[->,thick] (g)  -- (x2);
\draw[->,thick] (c1) -- (x1);
\draw[->,thick] (c2) -- (x2);
\draw[->,thick] (g) -- (c1);
\draw[->,thick] (g) -- (c2);
\node at (0,-1.0) {(\textbf{b})};
\end{scope}
\end{tikzpicture}
\caption{The two systems of assumptions, drawn for $K = 2$. Each view is
produced from the common factor and its own private factor through the mixing
$\boldsymbol{v}_k$ of~\eqref{eq:gen_model}. In (\textbf{a}) the three latent
factors are mutually independent, which is~\eqref{eq:group_indep} and the
assumption of Lyu {et~al.} In (\textbf{b}) two further edges are added, so
each private factor may depend on the common one, which
is~\eqref{eq:cond_indep} and the assumption of Karakasis and Sidiropoulos. The
views stay conditionally independent given $\bg$ in both cases, since every path
between $\bx_1$ and $\bx_2$ passes through $\bg$.
\label{fig:latent_models}}
\end{figure}

A direct consequence of~\eqref{eq:views_cond_indep} is that the minimum mean-squared-error estimate of view $k$ given $\bg$ and the remaining views is a function of $\bg$ alone. This motivates considering the problem
\begin{linenomath}
\begin{equation}
\begin{aligned}
  \min_{\{\boldsymbol{f}_k\},\, \widehat\bg,\, \{\boldsymbol{w}_k\}} \;\;
  & B^{(K)}(\boldsymbol{w}_1,\ldots,\boldsymbol{w}_K,\widehat\bg)
    := \sum_{k=1}^{K}
    \E\Big[\big\| \boldsymbol{w}_k(\widehat\bg) - \bx_k \big\|^2\Big] \\
  \text{s.t.} \;\;
  & \E\big[\widehat\bg\,\widehat\bg^\top\big] = \bI_D, \quad
    \E[\widehat\bg] = \mathbf{0}_D, \quad
    \bar R^{(K)}(\boldsymbol{f}_1,\ldots,\boldsymbol{f}_K,\widehat\bg) = 0 .
\end{aligned}
  \label{eq:idgcca_exact}
\end{equation}
\end{linenomath}
Comparing~\eqref{eq:lyu_exact} with~\eqref{eq:idgcca_exact} is instructive. Both
ask the learned common factor to have uncorrelated elements of unit variance.
The perfect consensus imposed by the constraint
$\bar R^{(K)}(\boldsymbol{f}_1,\ldots,\boldsymbol{f}_K,\widehat\bg) = 0$
in~\eqref{eq:idgcca_exact} corresponds to the perfect correlation attained at the
optimum of~\eqref{eq:lyu_exact}, and under the respective generative
assumptions, both are achievable.

What is flipped is the role of the reconstruction error. In~\eqref{eq:lyu}, it
is used to promote invertibility, and under the generative assumptions it
vanishes. In~\eqref{eq:idgcca_exact} it is the objective rather than a
constraint, and it cannot vanish when the private factors affect the views,
since $\widehat\bg$ does not carry private information. Yet the reconstruction
error remains informative. The total variance of a view is the mean squared
error of its best constant predictor. Hence, if the learned $\widehat\bg$ pushes
the reconstruction error of a view below that total variance, then
$\widehat\bg$ carries information about the view.

The two requirements correspond to the two questions this subsection opened
with. The constraint $\bar R^{(K)} = 0$ asks every view to agree on
$\widehat\bg$ for every realization, which is what makes it purely common. The
objective asks $\widehat\bg$ to be as informative as possible about each view,
which keeps it from being a trivial consensus. In that respect, the formulation
is close in spirit to DCCAE, where reconstruction also assigns meaning to the
encodings. Since the reconstruction cannot be perfect, however, nothing so far
prevents a coarsening that discards common information the best possible
reconstruction does not need. Karakasis and Sidiropoulos provide a sufficient
condition that rules this out, under which the learned factor is an invertible
reparameterization of $\bg$ and therefore informationally equivalent to~it.

\begin{Theorem}[\cite{karakasis2023}]
\label{thm:idgcca}
Consider a solution of~\eqref{eq:idgcca_exact} under the generative model
\eqref{eq:gen_model} and~\eqref{eq:cond_indep}, and assume that
\begin{enumerate}
\item[(i)] The mixing functions $\boldsymbol{v}_k$ are partially invertible with
respect to $\bg$, meaning that there exist maps $\boldsymbol{u}_k$ from
$\R^{M_k}$ to $\R^{D}$ with
$\boldsymbol{u}_k\big(\boldsymbol{v}_k([\bg^\top,\bc_k^\top]^\top)\big) = \bg$~
for every realization of $(\bg,\bc_k)$ and every $k$;
\item[(ii)] There is mean dependence between $\bg$ and at least one view, that
is, $\E[\bx_{k'} \mid \bg] \neq \E[\bx_{k'}]$ for some $k'$.
\end{enumerate}
{Then} 
 the learned encodings are nontrivial functions of the common factor alone,
$\boldsymbol{f}_k(\bx_k) = \boldsymbol{\psi}(\bg)$ for some map
$\boldsymbol{\psi}$ from $\R^{D}$ to $\R^{D}$ that is not constant. If in addition
\begin{enumerate}
\item[(iii)] The stacked conditional expectation
$\E\big[[\bx_1^\top,\ldots,\bx_K^\top]^\top \mid \bg\big]$ is an invertible
function of $\bg$,
\end{enumerate}
{then} $\boldsymbol{\psi}$ is invertible and $\bg$ is identifiable up to invertible nonlinearities.
\end{Theorem}

Assumption (i) is what replaces the full invertibility of
Theorem~\ref{thm:lyu}. It requires $\bg$ to be recoverable from each view,
regardless of whether the private factor is. Assumption (ii) is the condition of
the previous paragraph in its exact form, since
$\E[\bx_{k'} \mid \bg] \neq \E[\bx_{k'}]$ holds precisely when reconstructing
view $k'$ from $\bg$ improves on its best constant predictor. It is a very mild
condition, since it fails only if every view is mean-independent of $\bg$, i.e., 
$\E[\bx_k \mid \bg] = \E[\bx_k]$ for every $k$, while for the theorem a single view suffices. Assumption (iii) is what closes the gap left above, ruling out the coarsening and giving $\boldsymbol{\psi}$ an inverse.

The works of Lyu {et~al.}\ and Karakasis and Sidiropoulos give conditions
under which, and criteria with which, the common factor can be recovered from
the observed views, and these criteria lead to practical algorithms. We now step
back and relate these results to the common information of
Section~\ref{sec:common_info}. A first step is to check that the generative
structure shared by the two works specifies the common factor uniquely, up to
invertible reparameterizations. The model of Lyu {et~al.}\ is a special case
of that of Karakasis and Sidiropoulos, since full invertibility implies partial
invertibility with respect to $\bg$, and group
independence~\eqref{eq:group_indep} implies the conditional
independence~\eqref{eq:cond_indep}. It therefore suffices to consider the latter.

\begin{Proposition}
\label{prop:intersection}
Consider the generative model~\eqref{eq:gen_model} under~\eqref{eq:cond_indep},
and let each $\boldsymbol{v}_k$ be partially invertible with respect to $\bg$, as
in assumption (i) of Theorem~\ref{thm:idgcca}. Let $\mathbf{w}$ be a random vector such
that $\mathbf{w} = \boldsymbol{f}_k(\bx_k)$ with probability one for every $k$, and such
that the views are conditionally independent given $\mathbf{w}$. Then
$\mathbf{w} = \boldsymbol{\varphi}(\bg)$ with probability one, for an invertible map
$\boldsymbol{\varphi}$.
\end{Proposition}
\begin{proof}
Given $\bg$, the views are conditionally independent
by~\eqref{eq:views_cond_indep}, and $\mathbf{w}$ is a function of each of them, so $\mathbf{w}$ is
conditionally independent of itself given $\bg$. A random vector that is
independent of itself is constant with probability one, so $\mathbf{w}$ is a function of
$\bg$. Conversely, $\bg = \boldsymbol{u}_k(\bx_k)$ is a function of each view,
and the views are conditionally independent given $\mathbf{w}$, so by the same argument
$\bg$ is a function of $\mathbf{w}$. Hence, $\mathbf{w}$ and $\bg$ are functions of each~other.
\end{proof}

We can now draw the parallels with Section~\ref{sec:common_info}.

\begin{Connection}[{\color{black}Analogy}: Identifiability at the intersection of Gács--Körner and Wyner]
\label{con:nonlinear_id}
{\color{black}Recall from Section~\ref{sec:wyner} that, for discrete views, the feasible sets
of the Gács--Körner and Wyner common information may intersect. A nonempty
intersection contains only the common part $W^*$ of Section~\ref{sec:gk}, up to
relabeling, and the chain~\eqref{eq:ci_chain} then holds with equality,
$K_{GK}(X;Y) = I(X;Y) = C(X;Y)$. The two works place their common factor $\bg$ in the intersection of the two feasible sets as well. Their invertibility conditions imply that $\bg$ is a deterministic function of each view, so the views agree on it with probability one, which places $\bg$ in the feasible set of the Gács--Körner common information. Their latent independence structures imply that the views are conditionally independent given $\bg$, as in~\eqref{eq:views_cond_indep}, which places $\bg$ in the feasible set of Wyner's common information. By Proposition~\ref{prop:intersection}, the intersection contains only $\bg$, up to invertible reparameterizations, which is the continuous counterpart of uniqueness up to relabeling. The two common informations, however, are infinite, since $\bg$ is continuous and recoverable from each view, so they cannot single out $\bg$ as they do in the discrete case.}
\end{Connection}

\begin{Connection}[{\color{black}Analogy}: A Gács--Körner problem for continuous views]
\label{con:idgcca_gk}
{\color{black}The Gács--Körner problem maximizes the entropy $H(W)$ over the variables on which
all views agree with probability one. Problem~\eqref{eq:idgcca_exact} keeps the
agreement as the constraint $\bar R^{(K)} = 0$ and replaces the entropy with the
reconstruction error of the views, to be minimized. The replacement is needed,
since the entropy of a continuous variable is infinite, and its differential
entropy changes under invertible maps, which are exactly the
reparameterizations the identifiability results leave free. As in the first
half of the proof of Proposition~\ref{prop:intersection}, every
$\widehat\bg$ on which the views agree is a function of $\bg$, so the candidates
are ordered by how much of $\bg$ they retain. Both objectives favor the finer
candidates. In the discrete case, any coarsening of the common variable lowers its entropy. In the continuous case, condition (iii) ensures that the best reconstruction of the views needs all of $\bg$, so any coarsening of $\bg$ increases the reconstruction error. Hence, both problems select the finest common variable, up to relabeling in the first case and up to invertible reparameterizations in the second. }
\end{Connection}

Overall, what can be recovered depends on the assumed generative model, not
only on the learning criterion. Correlation and consensus objectives alone admit
many solutions, from the trivial consensus of Section~\ref{sec:dcca} to
coarsenings of the common factor. Lyu {et~al.}\ and Karakasis and
Sidiropoulos single out the common factor by assuming that it is a common part of
the views, as in Gács--Körner, that also renders them conditionally independent,
as in Wyner. In the continuous setting, the information values can no longer
select the common factor, and the identifiability conditions take their place.
It remains to see how problem~\eqref{eq:idgcca_exact} can be solved in practice.

\subsubsection{A Tractable Relaxation}

Given $N$ training samples, the constraints in~\eqref{eq:idgcca_exact} take a
finite-dimensional form. Writing $\bG$ for the matrix whose rows are the
realizations of $\widehat\bg$, the second-moment constraint becomes
$\bG^\top\bG = N\bI_D$ and the mean constraint
$\bG^\top\mathbf{1}_N = \mathbf{0}_D$, while $\bar R^{(K)} = 0$ becomes the $KN$
vector equations $\boldsymbol{f}_k(\bx_k^{(n)}) = \bG(n,:)^\top$. The
optimization variables are $\bG$ together with the parameters of the encoders
and decoders. Updating $\bG$ under the orthogonality constraint calls for
optimization on the Stiefel manifold, up to the scaling by $\sqrt{N}$, which is
expensive and sensitive to local minima, and the equality constraints add to the
difficulty.

Karakasis and Sidiropoulos suggest two modifications that
make~\eqref{eq:idgcca_exact} tractable. The first removes $\widehat\bg$ from the
objective. At any feasible point of~\eqref{eq:idgcca_exact}, every encoding
equals $\widehat\bg$, so the encoding of any view can be used in place of
$\widehat\bg$ to reconstruct the others. Replacing $\widehat\bg$ in this way, and
excluding each view from its own reconstruction, gives the cross-reconstruction
loss\vspace{12pt}
\begin{linenomath}
\begin{equation}
  Q^{(K)}(\boldsymbol{f}_1,\boldsymbol{w}_1,\ldots,
  \boldsymbol{f}_K,\boldsymbol{w}_K)
  := \frac{1}{K-1}\sum_{k=1}^{K} \sum_{\substack{j=1 \\ j \neq k}}^{K}
     \E\Big[\big\| \boldsymbol{w}_k\big(\boldsymbol{f}_j(\bx_j)\big)
     - \bx_k \big\|^2 \Big],
  \label{eq:cross_recon}
\end{equation}
\end{linenomath}
which coincides with $B^{(K)}$ at every feasible point, since there the inner sum
is $K-1$ copies of the same term, and the factor $1/(K-1)$ makes the two values
agree. Under assumption (i) of Theorem~\ref{thm:idgcca}, nothing is lost by
reconstructing view $k$ from another view instead of from $\bg$. Since $\bg$ is a
function of $\bx_j$, conditional independence gives
$\E[\bx_k \mid \bx_j] = \E[\bx_k \mid \bg]$, so both achieve the same best
possible error.

Excluding $\bx_k$ from the reconstruction of itself is a deliberate choice, and
it is the mechanism that prevents leakage. Under~\eqref{eq:cond_indep},
information about $\bc_k$ beyond what $\bg$ carries is available only through
view $k$, so excluding that path removes any reward for the encodings to carry
it, regardless of how strong the private components may be. Together with the agreement of the encodings, this keeps private information out of them, as the authors argue. This is the structural difference from~\eqref{eq:dccae}, where each view reconstructs itself and hence private variation can be rewarded.

The second modification replaces the agreement constraint $\bar R^{(K)} = 0$ with a penalty, following a Lagrangian approach. This removes the equality constraints and leaves only the orthogonality and zero-mean constraints on $\widehat\bg$. Together, the two modifications give
\begin{linenomath}
\begin{equation}
\begin{aligned}
  \min_{\{\boldsymbol{f}_k\},\, \widehat\bg,\, \{\boldsymbol{w}_k\}}\;
  & (1-\lambda)\, \bar R^{(K)}(\boldsymbol{f}_1,\ldots,
    \boldsymbol{f}_K,\widehat\bg)
    \;+\; \lambda\, Q^{(K)}(\boldsymbol{f}_1,\boldsymbol{w}_1,\ldots,
    \boldsymbol{f}_K,\boldsymbol{w}_K) \\
  \text{s.t.} \;\;
  & ~~~~~~~~~~~~~~~~\E\big[\widehat\bg\,\widehat\bg^\top\big] = \bI_D, \quad
    \E[\widehat\bg] = \mathbf{0}_D ,
\end{aligned}
  \label{eq:idgcca}
\end{equation}
\end{linenomath}
where $\lambda \in (0,1)$ sets the trade-off, with smaller values enforcing the
agreement of the encodings more strictly. Note that Theorem~\ref{thm:idgcca} is
stated for~\eqref{eq:idgcca_exact}, while~\eqref{eq:idgcca} is what is solved in
practice.

\subsubsection{Algorithm}
Problem~\eqref{eq:idgcca} admits alternating minimization over the two groups of
variables identified above. With $\bG$ fixed, updating the encoders and decoders
is unconstrained and is handled by stochastic optimization. With the encoders
and decoders fixed, the update of $\bG$ minimizes
$\sum_{k=1}^{K}\|\bF_k - \bG\|_F^2$, where
$\bF_k(n,:) = \boldsymbol{f}_k(\bx_k^{(n)})^\top$, under the orthogonality and
centering constraints. Let $\bY = \sum_{k} \bF_k$, and let $\bar\bY$ be its columnwise centered
version. Expanding the squares, the update reduces to solving the orthogonal Procrustes problem
\begin{linenomath}
\begin{equation}
  \max_{\bG \in \R^{N \times D}} \Tr\big(\bG^\top \bar\bY\big)
  \quad \text{s.t.} \quad \bG^\top\bG = N\bI_D ,
  \label{eq:procrustes}
\end{equation}
\end{linenomath}
solved in closed form as $\bG^\star = \sqrt{N}\,\bU\bV^\top$ from the SVD
$\bar\bY = \bU\bSigma\bV^\top$. The columns of $\bU$ lie in the column space of
$\bar\bY$, which is orthogonal to $\mathbf{1}_N$, so $\bG^\star$ also satisfies
the centering constraint when $\bar\bY$ has rank at least $D$. Per iteration the
cost is $\mathcal{O}\big(N D^2 + |\mathcal{B}| \sum_k p_k\big)$, with $p_k$ the
number of parameters of $\boldsymbol{f}_k$ and $\boldsymbol{w}_k$ and
$|\mathcal{B}|$ the batch size.

\subsubsection{Trivial Solutions}

Theorems~\ref{thm:lyu} and~\ref{thm:idgcca} are population statements,
while the objectives, losses and constraints are all evaluated on finitely many
training samples. With a sufficiently rich function class, the empirical problems
admit solutions whose encodings merely index the training samples, as in the
construction of Section~\ref{sec:dcca}, so overfitting remains possible here as
it was there. The two lines of work offer complementary guarantees. For kernel
CCA, consistency is established~\cite{fukumizu2007}, as it is for the functional
CCA of Leurgans {et~al.}~\cite{leurgans1993}, but the limit is defined by
maximal correlation alone and need not correspond to any latent factor. For the
neural network formulations, the guarantee is instead a population one. Under
the generative models discussed above, the population problems identify the
common factor, up to invertible reparameterizations, but neither work shows that
the encoders learned from finitely many samples converge to it. To the best of our knowledge, the closest result is a finite-sample bound of Lyu {et~al.}, which limits how much the learned shared encoding depends on the private factor, for encoders of bounded complexity.

\subsection{Connection to Self-Supervised Learning}
\label{sec:ssl}

Everything in this section has assumed that the views are given by the problem,
such as the acoustic and articulatory measurements of the speech example.
Self-supervised learning drops that assumption. Its goal is to learn, from
unlabeled data, representations that are useful for downstream tasks such as
classification. To that end, self-supervised learning manufactures the views,
usually by applying two random augmentations to the same image, designed to leave
its semantic content unchanged, and then learns encodings that capture that
content. The semantic content that an augmentation preserves thus plays the role
of the common factor, and what the augmentation changes plays the role of the
private one. What self-supervised learning optimizes then has the structure we have been working with throughout. One term rewards agreement between the two encodings, and a second term stops that agreement from being achieved on a constant. This subsection highlights how strong the similarities are between self-supervised learning and nonlinear CCA, and where they stop. An early work in this line, by Becker and Hinton, already named canonical correlation as what it extends, and the identifiability questions of Section~\ref{sec:idgcca} reappear in the self-supervised literature in nearly the same form.

\subsubsection{The Original Argument}

Becker and Hinton~\cite{becker1992} trained two small networks, which they call
modules, on inputs that share only a common cause. In their multiview setup, two
neighboring patches of a stereo image serve as views, and the only property they
share is the local depth of the scene. The idea is to train the networks to agree, so that, without any labels, they learn to report depth. How agreement is measured, however, matters. The simplest measure admits the trivial solution of Section~\ref{sec:dcca}. As Becker and Hinton note, penalizing the squared difference of the outputs ``usually causes each module to produce the same constant output that is unaffected by the input to the module and therefore conveys no information about~it''.

What they propose instead is that the outputs ``agree closely relative to how
much they both vary as the input is varied''. With $a$ and $b$ the outputs of the
two networks, this corresponds to maximizing
\begin{equation}
\label{eq:imax}
I(a,b) = \tfrac{1}{2}\log\frac{V(a+b)}{V(a-b)},
\end{equation}
where $V$ is the variance over the training cases. When $a$ and $b$ have equal
variances $\sigma^2$ and correlation $\rho$, we have
$V(a \pm b) = 2\sigma^2(1 \pm \rho)$, so the ratio equals $(1+\rho)/(1-\rho)$. Hence, maximizing~\eqref{eq:imax} maximizes the correlation of the two outputs, and a constant output is no longer optimal. The paper itself describes the procedure as ``a multilayer nonlinear extension of statistical methods such as canonical correlation or alternating conditional expectation because it maximizes the normalized covariation between the outputs''.

In~\eqref{eq:imax}, agreement and the prevention of collapse are combined in a
single ratio. A small denominator rewards agreement, and a large numerator rules
out a constant output. Modern self-supervised methods separate the two roles into
distinct terms of the loss.

\subsubsection{How Collapse Can Be Prevented}

The self-supervised methods keep this two-part structure, an agreement term and
a term that prevents collapse, and differ in how the latter is realized. Let
$\bx$ and $\bx'$ be two augmentations of the same image, encoded as
$\boldsymbol{z} = \boldsymbol{f}_1(\bx)$ and
$\boldsymbol{z}' = \boldsymbol{f}_2(\bx')$, and let
$\bZ, \bZ' \in \mathbb{R}^{N\times D}$ collect the encodings of a minibatch of
$N$ samples. SimCLR~\cite{chen2020simclr} prevents collapse by contrasting the
encodings of matching augmentations with those of the other images in the same
minibatch. Writing $\boldsymbol{z}_1,\dots,\boldsymbol{z}_{2N}$ for the rows of
$\bZ$ and $\bZ'$ together, a positive pair $(i,j)$ is one image under its two
augmentations, and the remaining $2(N-1)$ points serve as its negatives. Then,
the contrastive loss for the pair $(i,j)$ is
\begin{equation}
\label{eq:ntxent}
\ell_{i,j} = -\log\frac{\exp\!\big(\mathrm{sim}(\boldsymbol{z}_i,\boldsymbol{z}_j)/\tau\big)}
{\sum_{m=1}^{2N}\mathbf{1}_{[m\neq i]}\exp\!\big(\mathrm{sim}(\boldsymbol{z}_i,\boldsymbol{z}_m)/\tau\big)},
\end{equation}
with $\mathrm{sim}$ the cosine similarity and $\tau$ a temperature. The
denominator is what keeps the encoders from collapsing. This loss is a form of
the InfoNCE loss of van den Oord {et~al.}~\cite{oord2018cpc}, {\color{black}who show that
minimizing it maximizes a lower bound on the mutual information between the
compared variables}, a bound that becomes tighter as the number of negatives
grows. As a result, SimCLR benefits from large batches, which provide more
negatives.

Barlow Twins~\cite{zbontar2021barlow} removes the negatives and constrains the
second-order statistics instead. Let $\widetilde\bZ$ and $\widetilde\bZ'$ be the
two batches of encodings, centered over the batch and with unit-norm columns. The
cross-correlation matrix between the two views and the loss built from it are \begin{equation}
\label{eq:barlow}
\bT = \widetilde\bZ^{\!\top}\widetilde\bZ' \in \mathbb{R}^{D\times D},
\qquad
\mathcal{L}_{\mathrm{BT}} = \sum_{p} (1-\bT_{pp})^2
+ \lambda \sum_{p} \sum_{q\neq p} \bT_{pq}^2 ,
\end{equation}
so that $\bT_{pq}$ is the sample correlation between coordinate $p$ of
one view and coordinate $q$ of the other, taken over the $N$ samples in the
batch, and $\lambda > 0$ weighs the two terms. Minimizing
$\mathcal{L}_{\mathrm{BT}}$ drives $\bT$ toward the identity. The diagonal terms
are the invariance requirement and the off-diagonal terms are what the authors
call a soft whitening constraint. The authors read the objective as an instance
of the information bottleneck, which seeks a representation that keeps as much
information as possible about the original image, while being as uninformative
as possible about the particular augmentation applied to it. In our terms, the
representation retains the common factor and discards the private one.

Comparing~\eqref{eq:barlow} with the whitening constraints of
Section~\ref{sec:cca} is worth doing carefully. Scaling each centered column to
unit norm divides each coordinate by its standard deviation over the batch, which
plays the role of the unit-variance part of the whitening constraints. The diagonal
condition $\bT_{pp}=1$ asks matched coordinates to be perfectly correlated across
the two views, which is what CCA maximizes. The off-diagonal condition asks
mismatched coordinates to be uncorrelated across views, whereas whitening asks
the coordinates of each view to be uncorrelated with one another. Since the
off-diagonal term never compares two coordinates of the same view, those
coordinates may remain correlated. The two conditions meet only at an exact
solution, where $\bT = \bI$ forces $\widetilde\bZ' = \widetilde\bZ$, so that
$\widetilde\bZ^{\top}\widetilde\bZ = \bI$, which is within-view whitening. This is
consistent with the authors' description of the off-diagonal term as a soft
whitening constraint.

VICReg~\cite{bardes2022vicreg} separates agreement from the terms that rule out
constant encodings, which~\eqref{eq:barlow} combines in a single cross-correlation
matrix. Its objective is
\begin{equation}
\label{eq:vicreg}
\lambda\, s(\bZ,\bZ') + \mu\big[v(\bZ)+v(\bZ')\big] + \nu\big[c(\bZ)+c(\bZ')\big],
\end{equation}
where $s$ is the mean squared distance between paired encodings and $v$
penalizes every coordinate whose standard deviation within the batch drops
below a target $\gamma$,\vspace{12pt}
\begin{equation}
\label{eq:vicreg_var}
v(\bZ) = \frac{1}{D}\sum_{p=1}^{D}
\max\!\Big(0,\ \gamma - \sqrt{\operatorname{Var}(\boldsymbol{z}^{(p)}) + \epsilon}\Big),
\end{equation}
with $\epsilon$ a small constant for numerical stability, while $c$ penalizes the
squared off-diagonal entries of the covariance matrix of a single view, and
$\lambda, \mu, \nu > 0$ weigh the three terms. The relation to CCA is close, but
not exact. For whitened encodings,
$\E\|\boldsymbol{z}-\boldsymbol{z}'\|^2 = 2D - 2\Tr\left[\E[\boldsymbol{z}\boldsymbol{z}'^{\top}]\right]$,
so minimizing the squared distance between the views boils down to CCA. VICReg keeps this agreement term but replaces the whitening constraints with two penalties, one that only bounds each standard deviation from below and one that decorrelates the coordinates of each view.

Neither the objective of Barlow Twins nor that of VICReg, however, rules out the
trivial solutions of Section~\ref{sec:dcca}. With sufficiently rich encoders, mapping most samples of
a batch to zero and the remaining few to suitably scaled, decorrelated values,
identically in both views, satisfies all the terms of~\eqref{eq:barlow}
and~\eqref{eq:vicreg}, exactly as it satisfies the whitening constraints there.

\subsubsection{When the Views Are Given After All}
In multimodal learning, the views are again given by the data rather than
constructed. CLIP~\cite{radford2021clip} trains an image encoder and a text
encoder jointly, maximizing the cosine similarity of the $N$ matched
image--caption pairs in a batch against the $N^2-N$ mismatched ones under a
symmetric cross-entropy loss. ImageBind~\cite{girdhar2023imagebind} extends this
to six modalities, namely images, text, audio, depth, thermal, and inertial
measurement unit readings, but does not align all pairs. It aligns each modality
to images only, using the InfoNCE loss on pairs $(\mathcal{I},\mathcal{M})$ of an
image and an observation of another modality, and reports that modalities
$\mathcal{M}_1$ and $\mathcal{M}_2$ become aligned even though no
$(\mathcal{M}_1,\mathcal{M}_2)$ pair was ever seen in training. The topology is
worth noting. Every modality is aligned to images alone, rather than to every
other modality as in the all-pairs sum of correlations of SUMCOR. By analogy with
MAXVAR, where all views are fitted to a single common representation $\bG$, the
image embedding plays the role of $\bG$, and the emergent alignment is then what
one would expect of a common representation that every view is required to
match.

\subsubsection{What Is Recovered}

Zimmermann et
al.~\cite{zimmermann2021} consider views generated as follows. A latent is drawn
uniformly from the latent space, a second latent is drawn from a conditional
centered at the first, and the two views are obtained by applying the same
mixing to the two latents. The mixing is never assumed known, only injective and
differentiable. Under this model, they show that, in the limit of infinitely
many negative samples, an encoder that minimizes a contrastive loss recovers the
latent variables that generated the views. When the latent space is a
hypersphere and the conditional is von Mises--Fisher, the minimizer
of~\eqref{eq:ntxent} recovers the latent up to an orthogonal map. When the latent
space is a convex body, the cosine similarity of~\eqref{eq:ntxent} is replaced by
a distance induced by a norm, and the latent is recovered up to an affine map.
What must be known is the form of the conditional that produced the positive
pair, because the similarity used in the loss has to be based on the same metric
as that conditional, so a contrastive objective identifies only the generative
model it was built for.

Von K\"ugelgen et~al.~\cite{vonkugelgen2021} consider views generated as
follows. The latent is split into content $\mathbf{c}\in\mathbb{R}^{n_c}$ and
style $\mathbf{s}$, and style may depend on content. A second latent is obtained
by keeping the content and changing a random subset of the style variables, with
every style variable changed with positive probability. The two views
$\mathbf{x}$ and $\tilde{\mathbf{x}}$ are obtained by applying the same mixing to
the two latents, and the mixing is smooth and invertible with a smooth inverse.
An encoder $f$ that maps a view to $(0,1)^{n_c}$ is trained to minimize
\begin{equation}\label{eq:alignmaxent}
\mathbb{E}\big[\|f(\mathbf{x})-f(\tilde{\mathbf{x}})\|_2^2\big]-h\big(f(\mathbf{x})\big),
\end{equation}
where $h$ is the differential entropy. The first term enforces alignment, and the
second rules out collapse. Under this model, they show that every smooth
minimizer of~\eqref{eq:alignmaxent} recovers content alone, up to an invertible
map, i.e., $f(\mathbf{x})=\boldsymbol{\psi}(\mathbf{c})$ for an invertible
$\boldsymbol{\psi}$. The number of content variables must be known, but the
distribution of the latent need not factorize. In the limit of many negative
samples, the denominator of~\eqref{eq:ntxent} acts as the entropy term, so
contrastive learning asymptotically recovers content as well.

In general, maximizing the differential entropy of an unconstrained output would
be meaningless. For example, scaling an output in $\mathbb{R}^{n_c}$ by a factor
$a>1$ increases $h$ by $n_c\log a$. In the setup of von K\"ugelgen et~al., the
output is confined to $(0,1)^{n_c}$, a set of unit volume, and since the
differential entropy of a vector confined to a set is at most the logarithm of
the volume of that set, with equality only for the uniform distribution, $h$ can be
at most zero, and the maximum is attained only by a uniform output. Under their
model, an encoder that is both aligned and uniform exists, so every minimizer
of~\eqref{eq:alignmaxent} is aligned and has a uniform output. These two properties together force the map from content to the output to be invertible, as von K\"ugelgen et~al. show.

In the result of von K\"ugelgen et~al., the invariance of content is a property
of the augmentation, which the user designs. When the views are given, as in
Section~\ref{sec:idgcca}, the corresponding conditions are hypotheses about the
data that have to be argued for case by case, and the guarantees of a
self-supervised method do not carry over to given views. This is also why the two fields focus on different things. Self-supervised learning focuses on designing the augmentation, and multiview learning focuses on justifying the assumptions.


\subsection{Summary}
\label{sec:nonlinear:summary}

This section replaced the linear maps of Sections~\ref{sec:cca}
and~\ref{sec:sp} with functions drawn from broader classes. The main points are
summarized below.

\begin{itemize}
\item \textbf{Rich function classes allow overfitting and trivial solutions.}
Optimal scaling, kernel CCA, and deep CCA restrict the function class in
different ways, and without a sufficient restriction, the learned representation
carries no information about the data, because learning criteria are evaluated
on finitely many samples. In optimal scaling, nominal and monotone transformations can merge distinct values, and degenerate results have been reported under extreme nonlinearities and high noise. In kernel CCA without regularization, every canonical correlation can equal one regardless of the data. In the deep formulations, with
a sufficiently rich class of networks, the empirical constraints can be met by
representations unrelated to the views, because agreement is imposed only at the
$N$ training samples. 

\item \textbf{Probabilistic formulations assume conditional independence but do
not single out the common factor.} The variational methods assume that the views
are conditionally independent given a shared latent factor, and fit the resulting
model by maximizing a lower bound on the likelihood. {\color{black}Conditional independence
places the factor in the feasible set of Wyner's common information, but does not
single it out, since the feasible set ranges from a minimally informative factor
to the views themselves, and two fitted models with the same likelihood may sit
at different points of that range.} Adding private factors makes the model more
expressive, but the private factor of a view can reconstruct that view while the
shared factor is ignored. As a result, the question of when the shared and private factors can be reliably separated remains~open.

\item \textbf{Agreement alone does not determine the representation.}
Proposition~\ref{prop:intersection} shows that a representation that agrees
across views and also renders the views conditionally independent is the common
factor, up to an invertible map. The criteria of this section reward agreement,
but do not impose conditional independence, which enters only as an assumption
of the generative model. Hence, even in the population setups, every function of the common factor agrees across views, including functions that discard part of the
common factor. As a result, additional requirements are needed for the representation to retain the whole common factor. Reconstructing each view from its own representation makes the representation more informative, but private variation may leak into the representation. Reconstructing each view from the
representations of the other views keeps private variation out, while requiring
the representation to retain what predicts the other views.

\item \textbf{The common factor is identifiable up to an invertible map, in the
population.} Lyu et~al. show identifiability for two views, each an invertible
function of the common factor and a private factor, with the common and private
factors mutually independent. Karakasis and Sidiropoulos show identifiability for
$K$ views. They require only that the common factor be a function of each view,
and allow the private factors to be dependent, provided that they are
conditionally independent given the common factor. In exchange, they require at
least one view to be mean dependent on the common factor, which makes the
learned representation a nontrivial function of the common factor alone. Moreover, they show that when the stacked conditional mean of the views given the common factor is an invertible function of the common factor, the common factor is recovered up to an invertible map. At any fixed $N$, the trivial solutions
remain available.

\item \textbf{Self-supervised learning seeks nontrivial agreement between
constructed views.} Many of its losses pair an agreement term with a term
intended to prevent collapse to a constant representation, which agrees
perfectly across views but carries no information. The second term is realized by negative samples in contrastive learning, by a soft whitening of the
cross-correlation matrix in Barlow Twins, or by a variance floor in VICReg. The identifiability result of von K\"ugelgen et~al. assumes that the augmentation leaves content unchanged and changes style, and the choice of augmentation implicitly defines what content and style are.

\end{itemize}

%
%

\section{Further Extensions and Open Problems}
\label{sec:open}

This section is in two parts. The first collects results that bear directly on
the problem of this tutorial but fall outside what the preceding sections
develop, so that a reader who wants to go further knows where to go. The second outlines open directions.

\subsection{Further Extensions}
\label{sec:ext}

\subsubsection{Unifying Generalized CCA with Subspace Clustering}
Generalized CCA assumes that all $K$ views share one common subspace. In many
deployments they do not, and the views fall into clusters, with the views inside a
cluster sharing a subspace. Karakasis and Sidiropoulos formulate this as subspace
clustering of subspaces (SCoS)~\cite{karakasis2026subspace}, in which the
assignment of views to clusters and the common subspace of each cluster are
estimated together. The formulation contains GCCA as the single-cluster case, and
it extends subspace clustering from vector-valued samples to matrix-valued ones,
which is why it unifies the two.

\subsubsection{Tensor CCA and the High-Order Correlation Criterion}
Every criterion in Section~\ref{sec:cca} is assembled from pairwise
correlations, so the extension to $K$ views sums or otherwise combines
quantities defined two views at a time. Luo et~al.~\cite{luo2015tcca} replace that with a criterion defined on all $K$ views at once, the high-order correlation of a covariance tensor built from the views. They show that maximizing it is
equivalent to computing the best rank-one approximation of that
tensor. The CP machinery of Section~\ref{sec:common_info}
therefore returns as the solver rather than as a representation of a
distribution. Wong et~al.~\cite{wong2020dtcca} give the deep version, learning nonlinear per-view maps whose images are correlated in high order.
Girka et~al.~\cite{girka2023tgcca} address a different tensor setting, in which each view is
itself a higher-order tensor, and constrain the canonical vectors to admit an
orthogonal rank-$R$ CP decomposition within the regularized GCCA
framework. To the best of our knowledge, none of the three comes with an identifiability result.

\subsubsection{Identifiability Without Paired Samples}
The results of Section~\ref{sec:nonlinear} assume that sample $n$ of view $k$
and sample $n$ of view $l$ come from the same draw. Timilsina
et~al.~\cite{timilsina2024nips} remove that assumption for linear mixtures and
establish identifiability of the shared components by minimizing a divergence
between the cross-modality distributions, so the alignment they need is
distributional rather than sample by sample.

\subsubsection{Identifiability Without Known Latent Dimensions}
The dimensions of the common and private factors are fixed in advance throughout
this tutorial. Shrestha and Fu~\cite{shrestha2025iclr} show that content-style identifiability holds
without prior knowledge of either dimension, provided sparsity constraints are
imposed on the learned representations. Their setting is
unaligned multi-domain data, in which a content variable is shared across all
domains and each domain has its own style. As in the previous paragraph, identifiability is established by matching the latent distributions across domains.

\subsubsection{Disentangling the Private Factors from the Common One}
Section~\ref{sec:idgcca} showed that, under the assumptions of Lyu et
al.~\cite{lyu2022}, an independence constraint between the shared and private
blocks of each encoder also recovers the private factors. Sristi et~al.~\cite{sristi2026ctae} pursue a separation of this kind in multi-region neural recordings, partitioning the latent space of each brain region into orthogonal shared and private subspaces. There both subspaces are estimated with coupled transformer encoders and decoders, which also handle the non-stationary temporal dynamics that a static model~\mbox{ignores}.

\subsubsection{Preventing Collapse in Deep CCA}
Section~\ref{sec:nonlinear} shows that a rich enough transformation class can
manufacture agreement between the views. He et~al.~\cite{he2024collapse} report the training counterpart of this, in which the performance of deep CCA drops sharply as training proceeds, and they relate it to the weight matrices of the networks becoming low-rank. As a remedy, they propose a noise regularizer that encourages each network to keep its weight matrices full rank, and the same regularizer also applies to deep GCCA.

\subsubsection{Adaptive Estimation}
Wang and Poor also implement their blind detector adaptively, tracking the
signal subspace as snapshots arrive with the projection approximation subspace
tracking algorithm of Yang~\cite{yang1995past}, rather than recomputing it from
a whole block. The CCA estimators of Section~\ref{sec:cca} have no such
counterpart and need the full record before processing. Peppas and Sidiropoulos
provide one, an adaptive CCA that updates the canonical directions as samples
arrive, applied to time synchronization for signal alignment~\cite{peppas2025adaptive}.

\subsubsection{Applications Beyond the Three Examples}
The formulation of Section~\ref{subsec:crossing-disciplines} has been carried
into settings the tutorial does not develop, and we name a few of them here,
grouped by field. In neuroimaging, a two-stage GCCA separates the task-related
response shared across subjects from resting-state activity private to each
one~\cite{karakasis2020multi, karakasis2022}. In network analysis, a related
problem is graph alignment, where two graphs must be matched with no node
correspondence given, and Karakasis et~al.~\cite{karakasis2021align} show how
subspace intersection can be used to align them. In wireless communications, an
underlay scheme places short-range secondary transmissions beneath a primary
signal and recovers them without coordinating with the primary. Ibrahim et
al.~\cite{ibrahim2022underlay} show how CCA can be combined with repetition
coding for this purpose, and the construction is carried to analog FM radio with
several receivers~\cite{razavi2024fm}. Packet repetition is also part of many
popular protocols in frequency-hopped IoT networks. Peppas et
al.~\cite{peppas2025signal, peppas2023harnessing} exploit that repetition to
create a multiview setup in which CCA applies directly to align packets received
in different hops, and the same alignment is carried to multicarrier waveforms
to make them resilient to interference~\cite{peppas2025wcl}.

\subsection{Open Problems}
\label{sec:op}

\subsubsection{Relaxing the Invertibility Condition}

The identifiability results of Section~\ref{sec:nonlinear} assume each mixing
function is partially invertible with respect to the common factor. The common
factor is then a function of any single view and can be recovered from one view
at a time. The next weaker assumption is that some joint function of all $K$
views determines the common factor, with no single view sufficient on its own.
To the best of our knowledge, neither an identifiability theory nor an algorithm
is available in that case. Note that the weakening costs more than analytical
convenience. A model identified this way cannot be applied to a single view at
test time, so the weaker condition returns us to the multiview fusion problem
set aside in Section~\ref{subsec:crossing-disciplines}. 

\subsubsection{Identifiability of the Private Factors}
In the most general model of Section~\ref{sec:nonlinear}, the private factors
are required only to be conditionally independent given the common factor, and,
to the best of our knowledge, no identifiability theory for them is written
directly in terms of that conditional independence. The available guarantee, by
Lyu et~al., instead assumes that the private factors are independent of the
common factor and of each other, and enforces this in learning through an
independence constraint. What is missing is a framework that recovers the
private factors under conditional independence alone.

\subsubsection{Estimating the Clustering of the Views}
Removing the pairing assumption and removing the clustering assumption are
different problems, and they should not be conflated. The first has been
addressed for linear mixtures. The second has not. Section~\ref{sec:sp:synthesis}
sets out the deployment that raises it, in which several groups of cell-edge
users are each received by a different subset of stations and some of the
stations are themselves mobile. Ibrahim and Sidiropoulos select a subset of
stations and apply generalized CCA to it~\cite{ibrahim2020selective}, which
addresses the problem for a fixed candidate subset and leaves the estimation of
the partition open. Note also that SCoS assumes the number of clusters is given,
and neither that deployment nor the underlay setting above provides such a
number.

\subsubsection{Finite-Sample Behavior and Minimax Limits}
Kernel CCA comes with consistency results (Section~\ref{sec:kcca_reg}), but, to
the best of our knowledge, no comparable result exists for deep CCA itself. The minimax rates, i.e., the smallest estimation error for the
common subspaces achievable by any estimator, as a function of the number of
samples, the number of views, the subspace dimensions and the noise level, are a
separate question and remain unknown. The multi-user
extension of Tandon, Sankar and Poor described in Section~\ref{sec:common_info}
is a starting point, since the rate on the common link bounds how much
information about the common factor reaches any decoder. Converting that rate
into an estimation error bound is open. The obstacle is that the Gray--Wyner
region is stated in entropy and mutual information while estimation theory
works in mean squared error, and the I-MMSE relationship of Guo, Shamai and
Verd\'{u}~\cite{guo2005} is a natural bridge between the two formalisms.

\subsubsection{Invariance, Domain Shift and Causality}
Correlation across views does not establish a causal relation. The factor these
methods recover is the one that co-varies across the views, and none of the
formulations surveyed here distinguishes the factor of interest from any other
variable that the views happen to share, such as a confounder introduced by the
acquisition process. The invariance that CCA enforces is indexed by the view, while the invariance that domain generalization requires is indexed by the environment, and a factor stable across views need not be stable across environments. Whether the first can be made to serve as the second is open, and answering it needs a model in which the views and the environments appear as separate indices.
\section{Conclusions}

This tutorial followed one question through information theory, statistics, signal processing, and machine learning. Given several views of the same phenomenon, what do the views have in common, and under what conditions can the common part be recovered from data?

Section~\ref{sec:common_info} gave two answers to the first part of the question for scalar-valued views. The G\'acs--K\"orner common information is defined through the highest-entropy function that each view can compute on its own, and Wyner's common information through the least informative variable that renders the views conditionally independent. The two quantities bound the mutual information of the views from below and from above, and all three coincide when the views share a common part given which they are conditionally~\mbox{independent}.

Section~\ref{sec:cca} turned to vector-valued views. CCA finds, within each view, the linear combinations that are most correlated with each other. The solution follows in closed form from the SVD of the {\color{black}coherence matrix}, whose singular values are the canonical correlations, the cosines of the principal angles between the subspaces of the two views. In this sense, CCA is the second-order counterpart of the common information of Section~\ref{sec:common_info}, and Section~\ref{sec:cca} related both notions to a single set of canonical correlations. For jointly Gaussian views, the canonical correlations determine both the mutual information and Wyner's common information. {\color{black}The relation to the G\'acs--K\"orner notion is weaker. A canonical correlation equal to one certifies that G\'acs--K\"orner common information exists, and the directions with canonical correlation one play the role of the common part by analogy.} On the side of recovery, CCA returns the common subspace exactly in the noiseless case, under rank conditions that require the common and private subspaces of the views to be linearly independent. Generalized CCA extends the analysis to $K$ views, where identifiability, cast as an intersection of subspaces, generically improves as views are added.

Section~\ref{sec:sp} reached the same problem from multiuser detection, where the goal was to recover selected users under strong interference. Our first stop was the result of Wang and Poor, who showed that, in CDMA with a single antenna, the signal subspace allows the zero-forcing and MMSE detectors to be built blindly, with the signature of the desired user as the only prior knowledge. With multiple antennas, the directions of the users are set by propagation rather than by design, and pilots take the role of the signatures. However, when only the pilots of the desired user are known, strong interference degrades the channel estimate, and longer pilot sequences are required, at the cost of rate. A second view allows the interference to be suppressed without pilots, and changes the question from which directions carry the most energy to which directions the views share, which is the question CCA answers. When the second view is provided by a second base station, CCA recovers the subspace of the cell-edge users in the noiseless case. When the second view is provided by repetition at the transmitter, the canonical directions converge, up to scaling and as the number of samples grows, to the linear MMSE detector of the desired user, and no channel estimation is required.

Section~\ref{sec:nonlinear} replaced the linear maps of CCA with nonlinear ones, drawn from optimal scaling, from a reproducing kernel Hilbert space, or from a class of neural networks. The richer the function class, the easier agreement becomes. Without regularization, every kernel canonical correlation can equal one regardless of the data, and a sufficiently rich class of networks can meet the deep CCA constraints at the training samples with representations unrelated to the views. Even in the population, when each view determines the common factor, every function of the common factor can be reproduced by all views, including functions that discard part of the common factor, so agreement alone does not determine the representation.{\color{black} The probabilistic formulations add conditional independence of the views given a shared latent factor, which places that factor in the feasible set of Wyner's common information without singling the factor out.} Under explicit generative models, Lyu et~al. for two views and Karakasis and Sidiropoulos for $K$ views showed that the common factor is identifiable up to an invertible map. Self-supervised learning pursues the same agreement between views that are constructed rather than given, and in the result of von K\"ugelgen et~al. the recovered factor is whatever the augmentations leave unchanged.

Of the open problems collected in Section~\ref{sec:open}, we single out four. Partial invertibility of each mixing function with respect to the common factor has not been relaxed, and relaxing that assumption to a joint function of all the views would return the problem to multiview fusion. No identifiability theory for the private factors is written directly in terms of their conditional independence. Finite-sample analyses cover only a few formulations, and minimax limits are unknown. Finally, the recovered factor is defined by agreement across views, and none of the formulations surveyed here separates a common cause from a confounder that the views share.

The practical lesson is the same across the tutorial. A high canonical correlation, or a high agreement between learned representations, does not necessarily indicate shared structure, since the agreement can be produced by the transformations themselves. The recovered factor can be read as the common factor only when a well-posed learning criterion is used under conditions that make the common factor identifiable, and these conditions should be checked against the application before the result is interpreted.

\vspace{6pt}

\authorcontributions{{Conceptualization, P.A.K. and W.S.; methodology, P.A.K.; formal analysis, P.A.K.; investigation, P.A.K.; visualization, P.A.K.; writing---original draft preparation, P.A.K.; writing---review and editing, P.A.K. and W.S.; supervision, W.S.; funding acquisition, W.S. All}
 authors have read and agreed to the published version of the manuscript.}

\funding{{This research was funded by the National Science Foundation, grant number 2201641.}} 

\institutionalreview{Not applicable.}

\informedconsent{Not applicable.}

\dataavailability{{No new data} 
 were created or analysed in this study. Data sharing is not applicable to this article.}

\acknowledgments{The authors thank {Nicholas} 
 D.\ Sidiropoulos for the many
discussions and collaborations that shaped some of the material presented in
this tutorial. During the preparation of this manuscript, the authors used
{Claude Opus 5.5} 
 (Anthropic) to generate the code for Figures~\ref{fig:perception}--\ref{fig:wyner_dsbs} and~\ref{fig:speech}. The figures were
subsequently edited and annotated by the authors. The authors have reviewed the
generated content and take full responsibility for the content of this
publication.}

\conflictsofinterest{The authors declare no conflicts of interest.}

\abbreviations{Abbreviations}{
The following abbreviations are used in this manuscript:\\

\noindent
\begin{tabular}{@{}ll}
{ACE}    & {Alternating Conditional Expectation} \\
ADC   & Analog-to-Digital Converter \\
BOLD  & Blood Oxygen Level-Dependent \\
BSC   & Binary Symmetric Channel \\
BSS   & Blind Source Separation \\
CCA   & Canonical Correlation Analysis \\
CDMA  & Code Division Multiple Access \\
CP    & Canonical Polyadic \\
DCCA  & Deep Canonical Correlation Analysis \\
DCCAE & Deep Canonically Correlated Autoencoders \\
DGCCA & Deep Generalized Canonical Correlation Analysis \\
DSBS  & Doubly Symmetric Binary Source \\
fMRI  & Functional Magnetic Resonance Imaging \\
GCCA  & Generalized Canonical Correlation Analysis \\
{GENVAR} & {Generalized Variance} \\
HGR   & Hirschfeld--Gebelein--Rényi \\
ICA   & Independent Component Analysis \\
{IoT}    & {Internet of Things} \\
KCCA  & Kernel Canonical Correlation Analysis \\
{KL}     & {Kullback--Leibler} \\
{MAXVAR} & {Maximum Variance} \\
{MFCC}   & {Mel-Frequency Cepstral Coefficient} \\
{MINVAR} & {Minimum Variance} \\
MMSE  & Minimum Mean-Square Error \\
{MSE}    & {Mean-Square Error} \\
MUD   & Multiuser Detection \\
NMF   & Nonnegative Matrix Factorization \\
NTF   & Nonnegative Tensor Factorization \\
PCA   & Principal Component Analysis \\
{PDF}    & {Probability Density Function} \\
PID   & Partial Information Decomposition \\
PLS   & Partial Least Squares \\
PMF   & Probability Mass Function \\
RDF   & Rate--Distortion Function \\
RKHS  & Reproducing Kernel Hilbert Space \\
SCoS  & Subspace Clustering of Subspaces \\
SNR   & Signal-to-Noise Ratio \\
{SSQCOR} & {Sum of Squared Correlations} \\
{SUMCOR} & {Sum of Correlations} \\
SVD   & Singular Value Decomposition \\
VAE   & Variational Autoencoder \\
{VCCA}   & {Variational Canonical Correlation Analysis} \\
WSN   & Wireless Sensor Network \\
\end{tabular} 
}


\begin{adjustwidth}{-\extralength}{0cm}
\reftitle{References}

\PublishersNote{}
\end{adjustwidth}
\end{document}